%% file: main.tex
\RequirePackage{etoolbox}

\newtoggle{full}
\newtoggle{showoverflow}
\newtoggle{anonymous}
\newtoggle{submission}
\newtoggle{llncs}
\newtoggle{fullpage}

\providetoggle{forcefull}
\providetoggle{forceconf}
\iftoggle{forcefull}{
  \toggletrue{full}
}{
  \iftoggle{forceconf}{
    \togglefalse{full}
  }{
    \toggletrue{full}
  }
}

\toggletrue{fullpage}

\ifboolexpr{togl{full} and (not togl{submission})}{
  \togglefalse{llncs}
}{
  \toggletrue{llncs}
}

\documentclass[envcountsame,envcountsect,runningheads]{llncs}

\iftoggle{showoverflow}{
  \overfullrule=10pt
}{}

\iftoggle{fullpage}{
  \usepackage{fullpage}
}{}

\input{sections/header}

\iftoggle{full}{
  \newcommand*{\appref}[1]{Appendix~\ref{#1}}
}{
  \newcommand*{\appref}[1]{the full version}
}

\title{Verifiable quantum advantage in extremely low depth}
\titlerunning{Verifiable quantum advantage in low depth}

\date{}

\iftoggle{anonymous}{
  \author{}
  \institute{}
}{
	\author{
		Alexandru Gheorghiu
	}
	\hypersetup{
		pdfauthor = {Alexandru Gheorghiu}
	}

	\authorrunning{A.\ Gheorghiu}

	\institute{
		IBM Research \\
		\href{mailto:agheorghiu@ibm.com}{agheorghiu@ibm.com}
	}
}

\begin{document}
\iftoggle{full}{
  {\renewcommand{\rm}{\textrm}
  \maketitle}
}{
  \maketitle
}

\iftoggle{full}{
  \nottoggle{submission}{
    \vspace{-0.25cm}
  }{}
}{}

\sloppy

\begin{abstract}
We give a sampling problem that is solvable by shallow quantum circuits,
hard for polynomial-time classical algorithms under lattice-based assumptions,
and efficiently verifiable by a classical computer. The quantum sampler admits
two implementations: one uses
log-logarithmic-depth quantum circuits with one- and two-qubit gates, i.e.,
$\QNCloglog$ circuits, while the other uses constant-depth quantum circuits with
unbounded fan-in gates, i.e., $\QAC$ circuits.
Our construction can be seen as compiling the Learning with Errors (LWE)-based single-round proof of
quantumness of Arabadjieva et al.~\cite{ArabadjievaEtAl2024} to very low depth. The price paid for this
compilation is the reliance on less standard, though well-motivated,
assumptions: in addition to the lattice knowledge assumption used in~\cite{ArabadjievaEtAl2024},
we require a strengthened variant of the \emph{adaptive-hardcore-bit property}
of LWE, for which we provide supporting evidence. Unlike previous
low-depth proofs of quantumness, the quantum computation here requires no
mid-circuit measurements or feed-forward: it consists only of running a shallow
circuit and sampling from its output distribution.
This shows that shallow quantum circuits have sufficient structure to solve
certain classically hard tasks whose solutions can be verified efficiently.
\end{abstract}

\section{Introduction}\label{sec:introduction}

A basic goal in quantum complexity theory is to understand how little quantum
computation is needed to obtain an advantage over efficient classical
algorithms. For such an advantage to provide a convincing \emph{test of
quantumness}, however, a classical verifier should also be able to check it
efficiently. Constant-depth quantum circuits are among the most restricted
models of quantum computation, yet under plausible assumptions they can still
solve classically hard problems. We ask whether their outputs can also retain
enough structure for efficient classical verification. A positive answer
would broaden our understanding of shallow quantum circuits and mark a step
toward more practical tests of quantumness.

Let us start by taking stock of the evidence for the power of shallow quantum circuits.
Sampling the output distribution of a
constant-depth quantum circuit can be hard for polynomial-time classical
algorithms under plausible complexity-theoretic assumptions. This was shown
by Terhal and DiVincenzo~\cite{TerhalDiVincenzo2004} and later extended to
constant-depth architectures suited to quantum simulators, with hardness
holding also on average~\cite{MillerSandersMiyake2017,BermejoVegaEtAl2018,
HaferkampEtAl2020}. The later results adapt hardness arguments
developed for the broader class of \emph{instantaneous quantum
polynomial-time} (IQP) circuits, whose depth is low but not
constant~\cite{BremnerJozsaShepherd2011,BremnerMontanaroShepherd2016}.
Relative to a black-box \emph{random oracle}, sampling hardness can even be made
unconditional: a constant-depth circuit that queries the oracle can sample
from the Fourier spectrum of the oracle's truth table~\cite{Aaronson2010,BassirianEtAl2026}.
Beyond the oracle setting, one can prove unconditional
separations between quantum circuits and shallow classical
circuits~\cite{BravyiGossetKoenig2018,CoudronStarkVidick2021,
GrierSchaeffer2020,HsiehEtAl2026}. Such separations can also be robust to
noise~\cite{BravyiGossetKoenigTomamichel2020} and hold on average rather than
only in the worst case~\cite{BeneWattsEtAl2019}.

A complementary line of research by
H{\o}yer and {\v S}palek~\cite{HoyerSpalek2005} showed that
constant-depth quantum circuits with unbounded fan-out---the class
$\QNCf$---can approximate the quantum Fourier transform and the arithmetic
needed for Shor's algorithm~\cite{Shor1997}.\footnote{Throughout this
paper, $\QNC$, $\QNCloglog$, $\QAC$, and $\QNCf$ denote classes of
\emph{circuits} rather than classes of problems. Let us also note why
we restrict attention to sampling rather than decision problems: in a sublogarithmic-depth
circuit over bounded-arity gates, each output bit depends on at most
$2^{o(\log\lambda)}=\lambda^{o(1)}$ input qubits, so any decision
problem solved by such circuits can be simulated classically in
subexponential time. In particular, for constant-depth circuits over bounded-arity gates the corresponding decision task is solvable in classical polynomial time, and for log-logarithmic depth in quasipolynomial time $2^{\polylog(\lambda)}$.} With polynomial-time classical pre- and
post-processing, this yields an efficient factoring algorithm. More
recently, Grier, Morris, and Wu showed that
constant-depth quantum circuits with unbounded fan-in gates, the class
$\QAC$, can compute functions in $\TC$, the class of constant-depth
threshold circuits, when given multiple copies of the
input~\cite{GrierMorrisWu2026}. This is significant because proving lower bounds
against $\TC$ remains a major challenge for current circuit lower-bound
techniques, which face barriers such as natural
proofs~\cite{RazborovRudich1997,NaorReingold2004}.

Shallow quantum circuits can also implement sophisticated cryptographic
primitives. For instance, Foxman et al.~\cite{FoxmanEtAl2026} showed, under
cryptographic assumptions, that \emph{pseudorandom unitaries} admit exact
$\QNCf$ implementations and inverse-polynomial-accuracy
$\QAC$ implementations. In related work, Gheorghiu showed that
\emph{pseudoentangled states}, states that ``hide'' their entanglement structure, can be prepared by $\QNC$
circuits, i.e., constant-depth circuits with bounded fan-in and bounded
fan-out gates~\cite{Gheorghiu2026}.

Despite these results, our understanding of shallow quantum computation has
an important \emph{verification} gap. No task is known that
is simultaneously solvable by $\QNC$ circuits, hard for general
polynomial-time classical computation, and efficiently verifiable by a
classical algorithm. Unconditional shallow-circuit separations rule out only
correspondingly shallow classical circuits, while generic sampling hardness
does not provide an efficient predicate for recognizing a valid sample.
The
verifiable problems known to be both solvable by shallow quantum circuits
and classically hard either require unbounded fan-out gates, as in the
H{\o}yer--{\v S}palek setting~\cite{HoyerSpalek2005}, or use mid-circuit
measurements and classical feed-forward, as in the interactive \emph{proofs
of quantumness} with shallow circuits of Hirahara and Le
Gall~\cite{HiraharaLeGall2021} and of Liu and
Gheorghiu~\cite{LiuGheorghiu2022}, and in the non-interactive protocol of Chia and Hung~\cite{ChiaHung2023}, which in addition uses a black-box random oracle as part of the quantum circuit.\footnote{We also note that the Chia and Hung result is not merely a proof of quantumness, but a protocol for certifying the depth of quantum provers~\cite{ChiaHung2023}.} This leaves the following
basic question:

\begin{quote}
  \begin{center}
\emph{Is there a (white-box) problem solvable by $\QNC$ circuits, classically
hard, and efficiently verifiable?}
  \end{center}
\end{quote}

\noindent We make partial progress toward answering this question by providing
a sampling task with the following properties:
\begin{itemize}[label=$\bullet$]
  \item it is classically hard under plausible lattice assumptions;
  \item its samples are classically verifiable in polynomial time;
  \item it is solvable by polynomial-width, log-logarithmic-depth
        quantum circuits (i.e., $\QNCloglog$ circuits);
  \item it is solvable by polynomial-width, constant-depth quantum circuits
        with unbounded fan-in gates (i.e., $\QAC$ circuits).
\end{itemize}
Since $\QAC\subseteq\QNCf$~\cite{HoyerSpalek2005,TakahashiTani2016}, the sampling task is also solvable by
$\QNCf$ circuits, though we do not emphasize this point: as mentioned, a
similar task can already be solved using the H{\o}yer--{\v S}palek result. Our
reason for targeting $\QAC$ is that it appears to be a much weaker
class. Whether polynomial-size constant-depth $\QAC$ circuits can
compute parity---equivalently, fan-out---is a long-standing open
problem: the best known circuits have exponential
size~\cite{Rosenthal2021}. Polynomial-size $\QAC$ circuits are known only for
fan-out on polylogarithmically many qubits~\cite{GrierMorrisWu2026}, and
circuits with few ancilla qubits provably cannot compute parity, even
approximately~\cite{NadimpalliEtAl2024}. By contrast, adding the
fan-out gate makes threshold arithmetic exactly computable in constant
depth and collapses the resulting depth
hierarchy~\cite{HoyerSpalek2005,TakahashiTani2016}. As far as is known,
$\QAC$ lacks precisely the operation that gives $\QNCf$ its power.

An important subtlety with such a result concerns circuit \emph{width}. A
polynomial-size circuit on $w$ qubits can always be simulated classically
by maintaining its $2^w$-dimensional state vector. In particular, a circuit
whose width is polylogarithmic in a parameter $\lambda$ has simulation cost at
most $2^{\polylog(\lambda)}$. It is relatively straightforward to construct
such quantum circuits that solve a classically hard problem, produce a verifiable
solution, and have log-logarithmic depth. The idea is to use the Cleve--Watrous
parallelization of Shor's algorithm to factor a semiprime with polylogarithmic
bit length~\cite{cleve-watrous}. More concretely, suppose we wish to factor a
semiprime represented using $\polylog(\lambda)$ bits. The best known
classical algorithms have quasipolynomial cost in $\lambda$, namely
$2^{\polylog(\lambda)}$.\footnote{This is the case when the semiprime has
$\omega(\log^3\lambda)$ bits.}  With the Cleve--Watrous
parallelization, however, the corresponding quantum circuit has depth
$\log(\polylog(\lambda))=O(\log\log\lambda)$. This is unsatisfactory, however: the construction merely
expresses the input length as polylogarithmic in a larger parameter, and with
respect to the actual input size the depth remains logarithmic. One could
instead ask for polynomially many small semiprimes to be factored in parallel,
thereby obtaining polynomial total width while preserving log-logarithmic
depth, but the classical running time would remain
quasipolynomial.\footnote{This argument applies to $\QAC$ as well:
factoring these small semiprimes can be done by a constant-depth
$\QAC$ circuit. This is achieved with fan-out gates that act on only
polylogarithmically many qubits. Fan-out at that scale can be
performed in polynomial-width $\QAC$, as shown recently by
Grier, Morris, and Wu~\cite{GrierMorrisWu2026}.}

We argue that to obtain a genuine separation at polynomial width, we need a task with \emph{near-exponential} classical hardness in the security parameter: ideally, no $2^{o(\lambda)}$-time classical algorithm should solve it, while bounded-arity quantum circuits of depth $O(\log\log\lambda)$ can. Our construction achieves these properties under the assumptions stated below.

\begin{theorem}[Informal]
\label{thm:informal}
Assume the subexponential knowledge-of-lattice-point (LK-$1/4$) assumption
(\Cref{ass:encoded-lk}) and the adaptive-hardcore-bit assumption with carry
predicates (\Cref{ass:carry-ahcb}). There is a family of sampling problems,
indexed by a security parameter $\lambda$, whose samples can be verified in
classical polynomial time. This is illustrated as a challenge--response
task in \Cref{prot:overview}. An honest quantum sampler is accepted with probability
$1-\negl(\lambda)$ and has two polynomial-size
implementations: a $\QNCloglog$ circuit over one- and two-qubit gates and a
constant-depth $\QAC$ circuit with unbounded fan-in gates. By contrast,
any uniform classical sampler running in time
$2^{o(\lambda)}$ is accepted with probability at most
$3/4+\negl(\lambda)$. Both quantum implementations consist of a single
unitary circuit followed by a final measurement layer and require neither
mid-circuit measurements nor classical feed-forward.
\end{theorem}
\noindent \Cref{tab:comparison} summarizes how our results compare to previous work
on quantum advantage from shallow circuits.

\begin{table}[!ht]
  \centering
  \begin{adjustbox}{max width=\textwidth}
  \footnotesize
  \renewcommand{\arraystretch}{1.35}
  \setlength{\tabcolsep}{5pt}
  \begin{tabular}{@{}>{\raggedright\arraybackslash}m{3.4cm}ccccccc@{}}
    \toprule
    & \shortstack{Classical\\hardness}
    & Verifiable
    & \shortstack{Quantum\\depth}
    & \shortstack{Mid-circuit\\meas.}
    & Fan-out
    & Fan-in
    & \shortstack{Black\\box} \\
    \midrule
    Sampling hardness \cite{TerhalDiVincenzo2004,MillerSandersMiyake2017,BermejoVegaEtAl2018,HaferkampEtAl2020}
    & \cellcolor{tabgood}exponential
    & \cellcolor{tabbad}no
    & \cellcolor{tabgood}$O(1)$
    & \cellcolor{tabgood}no
    & \cellcolor{tabgood}bounded
    & \cellcolor{tabgood}bounded
    & \cellcolor{tabgood}no \\
    Fourier sampling \cite{Aaronson2010,BassirianEtAl2026}
    & \cellcolor{tabgood}exponential
    & \cellcolor{tabbad}no
    & \cellcolor{tabgood}$O(1)$
    & \cellcolor{tabgood}no
    & \cellcolor{tabgood}bounded
    & \cellcolor{tabgood}bounded
    & \cellcolor{tabbad}yes \\
    Unconditional separations \cite{BravyiGossetKoenig2018,CoudronStarkVidick2021,HsiehEtAl2026}
    & \cellcolor{taborange}shallow circuits
    & \cellcolor{tabgood}yes
    & \cellcolor{tabgood}$O(1)$
    & \cellcolor{tabgood}no
    & \cellcolor{tabgood}bounded
    & \cellcolor{tabgood}bounded
    & \cellcolor{tabgood}no \\
    Interactive shallow Clifford circuits \cite{GrierSchaeffer2020}
    & \cellcolor{taborange}$\NC$ circuits
    & \cellcolor{tabgood}yes
    & \cellcolor{tabgood}$O(1)$
    & \cellcolor{tabbad}yes
    & \cellcolor{tabgood}bounded
    & \cellcolor{tabgood}bounded
    & \cellcolor{tabgood}no \\
    Parallelized Shor \cite{Shor1997,cleve-watrous}
    & \cellcolor{tabgood}exponential
    & \cellcolor{tabgood}yes
    & \cellcolor{tabbad}$O(\log\lambda)$
    & \cellcolor{tabgood}no
    & \cellcolor{tabgood}bounded
    & \cellcolor{tabgood}bounded
    & \cellcolor{tabgood}no \\
    Factoring with fan-out \cite{HoyerSpalek2005}
    & \cellcolor{tabgood}exponential
    & \cellcolor{tabgood}yes
    & \cellcolor{tabgood}$O(1)$
    & \cellcolor{tabgood}no
    & \cellcolor{tabbad}unbounded
    & \cellcolor{tabgood}bounded
    & \cellcolor{tabgood}no \\
    Interactive proofs of quantumness \cite{HiraharaLeGall2021,LiuGheorghiu2022}
    & \cellcolor{tabgood}exponential
    & \cellcolor{tabgood}yes
    & \cellcolor{tabgood}$O(1)$
    & \cellcolor{tabbad}yes
    & \cellcolor{tabgood}bounded
    & \cellcolor{tabgood}bounded
    & \cellcolor{tabgood}no \\
    Proofs of quantum depth \cite{ChiaHung2023}
    & \cellcolor{tabgood}exponential
    & \cellcolor{tabgood}yes
    & \cellcolor{tabgood}$O(1)$
    & \cellcolor{tabbad}yes
    & \cellcolor{tabgood}bounded
    & \cellcolor{tabgood}bounded
    & \cellcolor{tabbad}yes \\
    Single-round proofs of quantumness \cite{ArabadjievaEtAl2024}
    & \cellcolor{tabgood}exponential
    & \cellcolor{tabgood}yes
    & \cellcolor{tabbad}$O(\log\lambda)$
    & \cellcolor{tabgood}no
    & \cellcolor{tabgood}bounded
    & \cellcolor{tabgood}bounded
    & \cellcolor{tabgood}no \\
    \midrule
    This work ($\QNCloglog$)
    & \cellcolor{tabgood}exponential
    & \cellcolor{tabgood}yes
    & \cellcolor{tabok}$O(\log\log\lambda)$
    & \cellcolor{tabgood}no
    & \cellcolor{tabgood}bounded
    & \cellcolor{tabgood}bounded
    & \cellcolor{tabgood}no \\
    This work ($\QAC$)
    & \cellcolor{tabgood}exponential
    & \cellcolor{tabgood}yes
    & \cellcolor{tabgood}$O(1)$
    & \cellcolor{tabgood}no
    & \cellcolor{tabgood}bounded
    & \cellcolor{taborange}unbounded
    & \cellcolor{tabgood}no \\
    \bottomrule
  \end{tabular}
  \end{adjustbox}
  \vspace{0.1in}
  \caption{Comparison with previous work on quantum advantage from shallow
  circuits. ``Classical hardness'' is the classical time complexity that
  the task is conjectured to require, measured in a security parameter
  $\lambda$ in which the input size is polynomial. The unconditional separations hold
  without assumptions but only against shallow classical circuits (the
  interactive variant of Grier and Schaeffer~\cite{GrierSchaeffer2020} is
  hard for $\NC$ circuits), and the
  hardness of Fourier sampling is unconditional relative to a random
  oracle; the remaining rows rely on complexity-theoretic or cryptographic
  assumptions, with the last two using those of \Cref{thm:informal}, which
  rule out $2^{o(\lambda)}$-time samplers.
  ``Verifiable'' indicates whether a classical polynomial-time verifier can
  check the output. The remaining columns list the quantum resources used:
  circuit depth, mid-circuit measurements with classical feed-forward, gate
  fan-out and fan-in, and black-box (random-oracle) access. Green marks an
  ideal property, yellow an acceptable caveat, and orange and red
  increasingly significant limitations for the goal of verifiable quantum
  advantage in low depth. Note that unbounded fan-in is orange because it is
  viewed as being weaker than unbounded fan-out, which is in red.}
  \label{tab:comparison}
\end{table}

The LK-$1/4$ and carry-predicate assumptions are nonstandard, so we briefly
motivate their origins. Our starting point is the LWE-based \emph{proof of
quantumness} of Brakerski et al.~\cite{BCMVV2021} (BCMVV), together with the
single-round variant of Arabadjieva et al.~\cite{ArabadjievaEtAl2024} (AGGM).
A proof of quantumness is a protocol between a classical polynomial-time
\emph{verifier} and a purportedly quantum polynomial-time \emph{prover}: the
verifier issues challenges that an efficient quantum prover can answer, but
that no efficient classical prover should answer with comparable success.
The original BCMVV protocol has two rounds (four messages) and requires the quantum prover
to perform mid-circuit measurements. Single-round variants, consisting of
one challenge and one response, can be obtained either in the
random-oracle model~\cite{BrakerskiKoppulaVaziraniVidick2020,
YamakawaZhandry2024,ChiaHung2023} or under a
\emph{knowledge assumption}, as in AGGM. The random-oracle route is ill-suited to a
low-depth result: because the oracle is a black box, it provides no
circuit-depth guarantee. Indeed, a light-cone argument shows that any exact
one- and two-qubit-gate implementation of a function with $n$ input bits and
one output bit, where the output depends on every input bit, must have depth
$\Omega(\log n)$. We therefore follow AGGM's fully white-box approach,
which is based on the LK assumption.

In cryptography, knowledge assumptions are \emph{non-falsifiable}:\footnote{This is a technical term; it does not mean that the assumption literally cannot be falsified. Indeed, the LK assumption is false for quantum algorithms and this is part of the reason why it's possible to use it in the design of a test of quantum advantage.} unlike
the hardness of factoring or LWE, they are not efficiently testable by an
external challenger. Instead, a knowledge assumption makes a claim about
the structure of algorithms---namely, that any algorithm accomplishing a
certain task must ``know'' an associated witness. This is formalized
through an efficient \emph{extractor}, tailored to the algorithm, that
recovers the witness from the algorithm's input and random tape. For
intuition, consider a one-way function whose image is a small random
subset of its codomain. If an algorithm, given a description of the
function, outputs a point in its image, a knowledge assumption asserts
that it must have known a corresponding preimage. Such functions are known
as \emph{extractable one-way
functions}~\cite{CanettiDakdouk2009,BitanskyEtAl2014}.

Knowledge assumptions were introduced by Damg{\aa}rd, who used them to
construct a public-key encryption scheme designed to resist
chosen-ciphertext attacks (CCA)~\cite{Damgard1991}. Variants were subsequently
refined for three-round zero-knowledge
arguments~\cite{HadaTanaka1998,BellarePalacio2004} and used in succinct
non-interactive arguments of knowledge (SNARKs) and practical verifiable-computation systems
such as Pinocchio~\cite{GennaroEtAl2013,ParnoEtAl2013}. Thus, although
non-falsifiable, knowledge assumptions have served as a cryptographic
design tool for more than three decades.

The LK assumption is a lattice-based knowledge assumption: roughly, an
algorithm that, given a description of a lattice, outputs a noisy lattice
point must know the lattice point closest to it. It was introduced in work
on CCA-secure somewhat homomorphic encryption~\cite{loftus}; related
lattice-knowledge assumptions have since been used in efficient
constructions of lattice-based designated-verifier zero-knowledge
SNARKs~\cite{GennaroEtAl2018,IshaiSuWu2021}.

Even under the LK assumption, the honest AGGM prover must evaluate a
logarithmic-depth circuit. We use \emph{randomized encodings} to compile the prover's computation to one that can instead be evaluated either in
log-logarithmic depth or in constant depth with unbounded fan-in gates.
Classical hardness rests on a subexponential form of the LK assumption
of AGGM, as well as an \emph{adaptive-hardcore-bit assumption for LWE with carry
predicates}. Informally, BCMVV proved the classical hardness of their interactive test using an adaptive-hardcore-bit property which
follows from LWE. This states that no efficient algorithm, given an LWE public
key, can predict a self-chosen parity of the secret non-negligibly better
than guessing, subject to an admissibility condition that excludes
degenerate choices. Our assumption states that adaptively chosen carry
bits do not help: the algorithm may augment its parity with bits that flip
when partial sums of the LWE equations are added to values of its choice,
making the predicted bit \emph{nonlinear} in the LWE secret. The carry bits account
for the structure that our randomized encoding places in the prover's
equation; the assumption is not currently known to follow from standard
LWE, though we provide partial evidence toward such a reduction.

\section{Technical Overview}\label{sec:overview}

We first recall the BCMVV and AGGM protocols and the LWE-based trapdoor
claw-free function, then describe the randomized encoding, the encoded
quantum strategy, verification and classical hardness, and the circuit
depth, and finally summarize the quantum advantage task as a two-message protocol.

\subsubsection{BCMVV and AGGM}
We begin with a brief recap of the proof of quantumness of Brakerski et
al.~\cite{BCMVV2021}. The protocol uses
\emph{trapdoor claw-free functions} (TCFs) constructed from LWE. A TCF
$f$ can be generated and evaluated efficiently and is $2$-to-$1$. Thus,
every output $y\in\Im(f)$ has exactly two preimages $x_0$ and $x_1$;
the pair $(x_0,x_1)$, for which $f(x_0)=f(x_1)=y$, is called a
\emph{claw}. Claw-freeness requires that no efficient algorithm can find
such a triple $(x_0,x_1,y)$. For the TCFs used in BCMVV, this property is
based on the conjectured intractability of LWE. The description of $f$ is
also generated together with a secret trapdoor that enables efficient
inversion: given the trapdoor and any $y\in\Im(f)$, a polynomial-time
algorithm can recover the corresponding claw $(x_0,x_1)$.
We use this idealized description for clarity; the noisy LWE instantiation
below has small correctness and boundary errors.

BCMVV use a somewhat stronger variant of a standard TCF. In addition to
being unable to find a claw, no efficient algorithm should be able, with
probability non-negligibly greater than $1/2$, to produce an image with a
valid preimage as well as an admissible string $d$ satisfying the equation
$d\cdot(x_0\oplus x_1) = 0$. This requirement is known as the
\emph{adaptive-hardcore-bit property} (AHCB). Intuitively, given an
input-output pair, it is intractable to predict even one bit of information
about the other preimage in the claw. Brakerski et al.~showed, based on the
intractability of LWE, that their TCF construction satisfies this
property~\cite{BCMVV2021}. We can now describe their protocol.

The verifier
first samples a TCF satisfying the AHCB property together with its trapdoor,
denoted $(f,\tau)$, and sends the description of $f$ to the prover. The
prover responds with a point $y\in\Im(f)$. The verifier then chooses
uniformly between two challenges. In the \emph{preimage test}, it asks for
an $x$ such that $f(x)=y$. In the \emph{equation test}, it asks for an
admissible string $d$ such that $d\cdot(x_0\oplus x_1)=0$, where
$(x_0,x_1)$ are the preimages of $y$. Both conditions can be checked
efficiently.
For the preimage test, the verifier evaluates $f$ on the prover's response
and checks that the result equals $y$. For the equation test, it uses the
trapdoor $\tau$ to recover $(x_0,x_1)$ from $y$ and checks the equation
for $d$.

Any classical prover that succeeds in this protocol with probability
non-negligibly greater than $3/4$ can be used to violate the AHCB property.
By contrast, in this idealized description there is a simple quantum strategy
that succeeds with probability $1$.
Suppressing normalization, the quantum prover first prepares the state
\[
\sum_{x} \ket{x}\ket{f(x)},
\]
and measures the second register, obtaining some $y\in\Im(f)$. Since
$f$ is a TCF, the first register collapses to
\begin{equation} \label{eq:clawstate}
\frac{1}{\sqrt{2}}(\ket{x_0} + \ket{x_1}).
\end{equation}
The prover retains this state and sends $y$ to the verifier. Upon
receiving the challenge, it measures the state in~\Cref{eq:clawstate} in
either the computational basis (for the preimage test) or the Hadamard basis
(for the equation test).
Despite its simple structure, the protocol is interactive and requires the
prover to keep the state in~\Cref{eq:clawstate} coherent until it receives
the verifier's challenge.

Arabadjieva et al.~\cite{ArabadjievaEtAl2024} show how to
replace this interactive protocol with a single-round one using a knowledge
assumption. To see the main idea, suppose that the TCF had the following
additional property: whenever a classical algorithm produces an image
$y\in\Im(f)$, an efficient \emph{extractor} can recover a corresponding
preimage from the algorithm's input and random coins. In essence, this says
that a classical prover that evaluates the TCF must ``know'' the preimage on
which the evaluation was performed. The preimage challenge would then no
longer be necessary. Instead, the prover could return $(y,d)$ in a single
response, and the verifier would only need to check that $y$ is a valid image 
and that the equation for $d$ is correct.

Existing knowledge assumptions do not, however, establish both extractability
and the AHCB property for the same function family. AGGM instead use two
function families whose public descriptions are computationally
indistinguishable. The first is the original TCF $f$, which satisfies the
AHCB property based on LWE. The second is an injective function $g$ for
which the \emph{knowledge-of-lattice-point} (LK) assumption provides the
required extractor. Rather than choosing a challenge after receiving $y$,
the verifier chooses one of these functions at the outset and hides this
choice in its description. An efficient classical prover cannot tell which
function it received and must return $(y,d)$ without a subsequent
challenge. For $f$, the verifier uses its trapdoor to recover the claw and
check the equation for $d$; for $g$, it uses its trapdoor to check that
$y$ is a valid image.

If a classical prover succeeds for $g$, the knowledge assumption gives an
extractor that recovers a preimage of $y$. Indistinguishability transfers
this extraction guarantee to $f$, where the extracted preimage and an
accepted equation would violate the AHCB property. An honest quantum prover,
by contrast, can obtain $y$ and $d$ by measuring the two registers of its
state and can therefore send both values at once without waiting for another
message from the verifier.
For this reason, the LK extraction guarantee is assumed only for classical
algorithms.\footnote{While the LK assumption already provides a separation between classical and quantum computations, it is unclear whether 
it is sufficient, on its own, to obtain a proof of quantumness. This is why the AGGM
protocol uses both LK and LWE.}

In terms of depth, the AGGM prover must coherently evaluate the LWE-based TCF.
Standard implementations have either logarithmic depth or constant depth with
unbounded fan-out~\cite{gheorghiuHoban}. To obtain a shallow implementation
with bounded fan-out, we adapt the randomized encoding for binary
affine maps proposed in~\cite{Gheorghiu2026} to affine maps over a ring.
Before explaining this construction, let us first describe the two functions to be encoded: the LWE-based TCF of BCMVV and the extractable function of AGGM.

\subsubsection{The LWE-based TCF}
Throughout, we write all vectorial quantities in boldface.
In LWE, one is given a matrix $\bA\in\Zq^{m\times n}$ and a vector
$\bA\bs+\be\bmod q$, where $\bs\in\Zq^n$ is secret, $\be$ is a
short error vector sampled from the truncated discrete Gaussian
$D_{\Zq,B}^{m}$ of parameter $B$, and $q$ is the modulus. The search variant of LWE asks one to recover $\bs$ from
$(\bA,\bA\bs+\be\bmod q)$~\cite{regev}.
The BCMVV instantiation chooses the secret from $\bits^n$. Its LWE-based
TCF is defined by
\[
 f(b,\bx)=\bA\bx+b\cdot(\bA\bs+\be)+\boldeta\pmod q,
\]
where $\boldeta\leftarrow D_{\Zq,B_P}^{m}$ for a width $B_P\gg B$.

It is useful to make the function noise explicit by writing
\[
 \bu=\bA\bs+\be\pmod q,\qquad
 f(b,\bx;\boldeta)=\bA\bx+b\bu+\boldeta\pmod q.
\]
The bit $b$ selects one of two noisy branches. Suppose that branch $b=0$
is evaluated at $(\bx_0,\boldeta_0)$. Taking
$\bx_1=\bx_0-\bs$ and $\boldeta_1=\boldeta_0-\be$ on branch $b=1$
gives
\[
 f(0,\bx_0;\boldeta_0)
 =f(1,\bx_1;\boldeta_1).
\]
Such pairs form claws of the function. The function noise $\boldeta$ is
sampled from a range much wider than that of the LWE error $\be$, so
translating it by $\be$ preserves all but a negligible fraction of its
mass. Consequently, almost every output produced by one branch also occurs
on the other. The $\bx$-components of the two inputs differ by the secret
$\bs$, so finding both would reveal $\bs$ and thereby break search-LWE.
Although we omit the details here, one can efficiently sample the LWE instance
$(\bA,\bu)$ together with a trapdoor that enables efficient inversion of
$f$.

The extractable function used by AGGM has the same public affine form,
\[
 g(b,\bx;\boldeta)=\bA\bx+b\bu+\boldeta\pmod q,
\]
but here the full matrix $[\bA\mid\bu]$ is generated as a single trapdoor
matrix, rather than taking $\bu=\bA\bs+\be$. Except with negligible
probability, every image of this function has a unique preimage over the input and noise
ranges considered in the protocol, and the LK assumption says that a
classical algorithm producing a valid image must know that preimage. Under
LWE, the public descriptions of $f$ and $g$ are computationally indistinguishable
to efficient classical provers, so a prover cannot adapt its strategy to the
function it receives.

\subsubsection{A randomized encoding for affine maps over a ring}
The remaining task is to obtain low-depth implementations of these two
functions. Direct evaluation does not suffice:
the functions effectively compute long sums modulo $q$, which require
logarithmic depth. Instead, we construct low-depth randomized encodings of
these functions. A randomized encoding of a function $F$ is a new
randomized function $\widehat F$ that retains key properties of $F$. Let
\[
 \bM=[\bA\mid\bu]\in\Zq^{m\times t},\qquad
 \bv=(\bx,b)\in\Zq^t,\qquad t=n+1.
\]
The affine map that we wish to encode, which computes either $f$ or $g$, is
\[
 \bv\longmapsto\bM\bv+\boldeta.
\]
Each output coordinate is a sum involving essentially every coordinate of
$\bv$, and each coordinate of $\bv$ is reused in every row of $\bM$.
Rather than compute these sums and distribute the inputs directly, the
randomized encoding replaces $\bM\bv+\boldeta$ by a longer string and defers
the global work to a classical decoder. We need three properties. First,
the decoder must recover the original output.
Second, the encoded string must reveal nothing else: its distribution should
depend on $(\bv,\boldeta)$ only through $\bM\bv+\boldeta$. Third, given
an input and a compatible encoded string, it must be possible to reconstruct
the randomness used to create that string. Shallow implementation also
requires locality: every symbol of the longer string should
depend on only a constant number of source values, and no source value
should be used in more than a constant number of symbols.
An encoding with these properties for \emph{binary} affine maps was given
in~\cite{Gheorghiu2026}; here we generalize it to maps over $\Zq$.

To see the basic idea behind the encoding, first consider
\[
 S(a_1,\ldots,a_t)=a_1+\cdots+a_t\pmod q.
\]
This is the modulo-$q$ analog of the parity function. Choose independent
uniform \emph{masks} $h_1,\ldots,h_{t-1}\in\Zq$, and set $h_0=h_t=0$. The randomized encoding for
$S$ outputs
\[
 c_j=a_j+h_{j-1}-h_j\pmod q,
 \qquad j=1,\ldots,t.
\]
One can think of $h_j$ as moving a random amount from position $j$ to
position $j+1$: it is subtracted from $c_j$ and added back into
$c_{j+1}$. This is why the $h_j$ are called \emph{path masks}: the
positions $1,\ldots,t$ form a path, each mask sits on an edge between
consecutive positions and shuttles a random amount one step along it,
and summing along the path cancels every mask; see
\Cref{fig:path-masks}. Each $c_j$ therefore
uses only one input and its two neighboring masks.

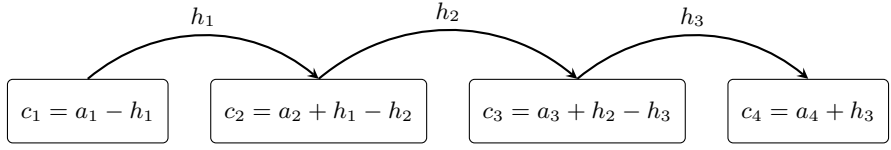
\begin{figure}[t]
\centering
\begin{tikzpicture}[>=stealth, node distance=0.55cm,
  sym/.style={draw, rounded corners=2pt, minimum height=2.6em,
              inner sep=5pt, font=\small}]
 \node[sym] (c1) {$c_1=a_1-h_1$};
 \node[sym, right=of c1] (c2) {$c_2=a_2+h_1-h_2$};
 \node[sym, right=of c2] (c3) {$c_3=a_3+h_2-h_3$};
 \node[sym, right=of c3] (c4) {$c_4=a_4+h_3$};
 \draw[->, thick] (c1.north) to[bend left=40]
   node[above, midway, font=\small] {$h_1$} (c2.north);
 \draw[->, thick] (c2.north) to[bend left=40]
   node[above, midway, font=\small] {$h_2$} (c3.north);
 \draw[->, thick] (c3.north) to[bend left=40]
   node[above, midway, font=\small] {$h_3$} (c4.north);
\end{tikzpicture}
\caption{The path-mask encoding of the sum $a_1+a_2+a_3+a_4$ modulo
$q$, for $t=4$. The symbols form a path, and the mask $h_j$ lives on
the $j$-th edge: it moves a random amount from position $j$ to position
$j+1$, subtracted from $c_j$ and added to $c_{j+1}$. Summing along the
path cancels every mask and returns the sum.}
\label{fig:path-masks}
\end{figure}

Adding the encoded symbols causes the masks to cancel in
adjacent pairs:
\[
 \sum_{j=1}^{t}c_j
 =\sum_{j=1}^{t}a_j+
   (h_0-h_1)+(h_1-h_2)+\cdots+(h_{t-1}-h_t)
 =S(a_1,\ldots,a_t).
\]
Thus, the decoder recovers $S$, even though each $c_j$ is computed
locally.

The tuple $(c_1,\ldots,c_t)$ also hides the individual summands. For any
fixed total, it is uniform over all tuples whose coordinates have that total.
Given only the total, one can therefore sample
$c_1,\ldots,c_{t-1}$ uniformly and choose $c_t$ to make the sum correct.
Thus, two inputs with the same sum give exactly the same distribution of
encoded tuples.

Applying this path-mask encoding independently to every row of $\bM\bv$
would handle the long sums, but it would still use the same input
$\bv_j$ in every row. We need to avoid this reuse in order to preserve
constant locality. To do this, we add a second layer of random masks.

For each input coordinate $j$, choose fresh uniform values
$\br_{j,1},\ldots,\br_{j,m}\in\Zq$ and output
\[
 \bw_{j,0}=\bv_j+\br_{j,1},\qquad
 \bw_{j,i}=\br_{j,i}-\br_{j,i+1}
 \quad(i=1,\ldots,m-1).
\]
From these symbols, the classical decoder can derive a separately masked
version of $\bv_j$ for row $i$:
\[
 \widetilde{\bv}_{j,i}
 =\bw_{j,0}-\sum_{i'=1}^{i-1}\bw_{j,i'}
 =\bv_j+\br_{j,i}.
\]
The long prefix sum is performed only by the decoder. The encoder
touches $\bv_j$ once, in $\bw_{j,0}$, and every other $\bw$-symbol is
just the difference of two adjacent masks.

The decoder's masked value $\widetilde{\bv}_{j,i}$ introduces an unwanted correction
$\bM_{i,j}\br_{j,i}$. We hide and collect these corrections using the
path-mask encoding above. Independently for each row $i$, choose
uniform masks $\bh_{i,1},\ldots,\bh_{i,t-1}$, set
$\bh_{i,0}=\bh_{i,t}=0$, and produce
\[
 \bz_{i,j}
 =\bM_{i,j}\br_{j,i}
  +\bh_{i,j-1}-\bh_{i,j}.
\]
Adding across row $i$ cancels the $\bh$-masks and leaves precisely the
unwanted correction:
\[
 \sum_j\bz_{i,j}=\sum_j\bM_{i,j}\br_{j,i}.
\]
The decoder can therefore recover the $i$-th coordinate of the matrix
product:
\[
 \sum_j\bM_{i,j}\widetilde{\bv}_{j,i}-\sum_j\bz_{i,j}
 =\sum_j\bM_{i,j}\bv_j.
\]

Finally, the function noise is inserted locally. We replace the first
$\bz$-symbol in row $i$ by
\[
 \widehat{\bz}_{i,1}=\bz_{i,1}-\boldeta_i
\]
and leave the other symbols unchanged. Since the decoder subtracts the
$\bz$-symbols, its output becomes
\[
 \sum_j\bM_{i,j}\widetilde{\bv}_{j,i}-\sum_j\widehat{\bz}_{i,j}
 =(\bM\bv)_i+\boldeta_i.
\]
This completes the encoding of the affine map.

The two mask layers address different obstacles. The $\br$-masks prevent
one input from feeding every matrix row, while the $\bh$-masks break each
row's correction sum into independently computable local terms.
The encoded output consists of all the $\bw$- and
$\widehat{\bz}$-symbols. The masks make its distribution depend only on
the decoded value, just as in the addition example. More concretely, for a
fixed input the map from the uniform $\br$-masks to the $\bw$-string is
a bijection, so $\bw$ is uniform and hides $\bv$. Once $\bw$ and
the decoded value are fixed, the $\bh$-masks make each row of
$\widehat{\bz}$ uniform subject only to the required row sum.

The masks are also uniquely reconstructible. Given $\bv$ and an encoded
output, one first decodes it to $\by=\bM\bv+\boldeta$ and obtains
$\boldeta=\by-\bM\bv$. Restoring
$\bz_{i,1}=\widehat{\bz}_{i,1}+\boldeta_i$, one recovers the
$\br$-masks successively from the $\bw$-symbols and then the
$\bh$-masks successively from the $\bz$-symbols.

Most importantly, the quantum part of this computation is entirely local.
Every encoded symbol depends on at most three source words. Conversely,
each input word appears once, each $\br$-mask appears in at most two
$\bw$-symbols and one $\bz$-symbol, and each $\bh$-mask appears in
two adjacent $\bz$-symbols; each function-noise word also appears once.
All of the prefix sums and row sums occur later in the classical decoder.
Once the encoded symbols have been measured, that decoder may copy and
reuse them freely.

\subsubsection{The encoded quantum strategy}
We can now explain the quantum prover's strategy in the protocol.
As mentioned, the prover will be asked to evaluate an encoding of an affine function.
We track the prover's state through the computation, suppressing
normalization and, until they matter, the scratch registers used for the
arithmetic. Write $\widehat F(b,\bx,\boldeta;\br,\bh)$ for the encoded
output, the string of all $\bw$- and $\widehat\bz$-symbols.

The prover first prepares a uniform superposition over the branch bit, the
input, the function noise, and the two collections of masks, next to an
output register initialized to zero:
\[
 \sum_{b,\bx,\boldeta,\br,\bh}
 \ket{b,\bx,\boldeta,\br,\bh}\ket{\mathbf 0}.
\]
In the shallow construction, the function noise is uniform over a
centered interval of power-of-two size. This both preserves the overlap
between the two TCF branches and allows its superposition to be prepared
exactly using Hadamard gates followed by a fixed modular translation.

Second, the prover evaluates the encoding into the output register:
\[
 \sum_{b,\bx,\boldeta,\br,\bh}
 \ket{b,\bx,\boldeta,\br,\bh}
 \ket{\widehat F(b,\bx,\boldeta;\br,\bh)}.
\]
Every symbol of $\widehat F$ depends on only a constant number of source
words, so this step is local.

Third, the prover erases the function-noise register. This step is needed
because the two sides of a TCF claw use different values of $\boldeta$:
if the noise remained, it would become part of the final Hadamard relation
below. Since $\boldeta_i$ enters a single symbol,
$\widehat\bz_{i,1}=\bM_{i,1}\br_{1,i}-\bh_{i,1}-\boldeta_i$, the local
reversible update
\[
 \boldeta_i\longleftarrow
 \boldeta_i+\widehat{\bz}_{i,1}+\bh_{i,1}
             -\bM_{i,1}\br_{1,i}
\]
maps $\boldeta_i$ to zero on every basis state of the superposition.
After this update, and after uncomputing the arithmetic scratch, the state
is
\[
 \sum_{b,\bx,\boldeta,\br,\bh}
 \ket{b,\bx,\mathbf 0,\br,\bh}
 \ket{\widehat F(b,\bx,\boldeta;\br,\bh)};
\]
the noise value survives only inside the output register.

Finally, the prover applies Hadamards to the registers containing
$b,\bx,\br,\bh$ and measures everything. All measurements occur at the
end; for the analysis it is convenient to view the output register as
measured first, which is equivalent because the measurements act on
different registers and commute. Measuring the output register yields an
encoded string, and the surviving superposition consists of exactly the
sources that produce it. Perfect privacy says that the two original
preimages of the decoded value induce identical distributions on encoded
strings, and randomness reconstruction says that exactly one mask
assignment is compatible with each of them. Except for a negligible
fraction of boundary and shifted-noise cases, the unmeasured registers are
therefore in the state
\[
 \frac{1}{\sqrt2}
 \left(\ket{0,\boldzeta_0}+\ket{1,\boldzeta_1}\right),
 \qquad
 \boldzeta_b=\bigl(J_q(\bx_b),\,J_q(\br^{(b)}),\,J_q(\bh^{(b)})\bigr),
\]
where $J_q$ denotes componentwise binary representation and
$\br^{(b)},\bh^{(b)}$ are the unique masks compatible with the input
$\bx_b$ and the measured output; the erased function noise appears in
neither string. The string $\boldzeta_b$ is the \emph{extended
preimage} on branch $b$, and $(b,\boldzeta_b)$ the complete extended
preimage.
The Hadamard measurement of these registers
produces a uniformly random \emph{direction} $\bD$ and a bit $C$ satisfying
\[
 C=\bD\cdot(\boldzeta_0\oplus\boldzeta_1).
\]
The randomized encoding has thus preserved the interference between the
two preimages used by the original protocol while making the coherent
computation local.

\subsubsection{Verification and classical hardness}
The verifier chooses uniformly between the TCF $f$ and its
indistinguishable extractable counterpart $g$. When it chooses $g$, its
trapdoor recovers the decoded preimage, and it checks that the encoded
output is a valid encoding of that preimage, masks included
(\Cref{def:validity-relation}). When it chooses
$f$, its trapdoor recovers the two original preimages, and randomness
reconstruction recovers the corresponding $\boldzeta_0,\boldzeta_1$; the
verifier then checks the stated equation
$C=\bD\cdot(\boldzeta_0\oplus\boldzeta_1)$, together with an admissibility
condition on $\bD$ (\Cref{def:admissible}, analyzed in
\appref{app:encoded-good-directions}).

It is worth writing out what this equation says, because its structure
dictates the assumption we make. The two extended preimages differ in
three ways. The inputs differ by the LWE secret:
$\bx_1=\bx_0-\bs$. Recall that the encoded input is
$\bv=(\bx,b)$; the shift means that $\bv_j$ changes by $-s_j$ for
$j\leq n$, while the branch coordinate $\bv_{n+1}=b$ changes by $+1$.
The $\br$-masks differ by coordinates of $\bs$: once the encoded output
is fixed, the first $\bw$-symbol of column $j$ determines
$\br_{j,1}=\bw_{j,0}-\bv_j$, and the remaining $\bw$-symbols fix the
differences of adjacent masks, so each input shift propagates down its
column with opposite sign:
\[
 r^{(1)}_{j,i}=r^{(0)}_{j,i}+s_j\quad(j\leq n),
 \qquad
 r^{(1)}_{n+1,i}=r^{(0)}_{n+1,i}-1.
\]
The $\bh$-masks differ by partial sums of the LWE equations: telescoping
the first $j$ symbols of row $i$ gives
$\bh_{i,j}=\sum_{j'\leq j}(\bM_{i,j'}\br_{j',i}-\bz_{i,j'})$,
and across the two branches the $\br$-shifts contribute
$\sum_{j'\leq j}\bA_{i,j'}s_{j'}$ while restoring the erased noise
($\boldeta_1=\boldeta_0-\be$) contributes $e_i$, so
\[
 h^{(1)}_{i,j}=h^{(0)}_{i,j}+\sigma_{i,j}(\bs),
 \qquad
 \sigma_{i,j}(\bs)=e_i+\sum_{j'\leq j}\bA_{i,j'}s_{j'}\pmod q.
\]
The parity $\bD\cdot(\boldzeta_0\oplus\boldzeta_1)$
compares such values bitwise. Because the secret is binary, the $\bx$-
and $\br$-contributions are linear in $\bs$, with coefficients the
prover can compute from its own transcript. The $\bh$-contributions are
nonlinear: XORing the binary representations of two values that differ by
a modular shift produces the carry pattern of that shift. For instance,
adding $1$ to $3$ turns $0011$ into $0100$ and flips three bits, while
adding $1$ to $4$ turns $0100$ into $0101$ and flips one: the flipped
bits show how far the carry propagates, and this depends jointly on
the shift and on the value being shifted. This is similar in spirit to
the \emph{crypto dark matter} proposals of Boneh et
al.~\cite{BonehEtAl2018} and the alternating-moduli constructions that
followed~\cite{DinurEtAl2021}, which build cryptographic primitives by
composing linear functions over two different moduli, a composition
that is nonlinear over either modulus; here the two moduli are $q$
and $2$. Altogether, the
checked equation takes the form
\[
 C=\langle\ba,\bs\rangle\oplus\beta\oplus
 \bigoplus_{i}\bigoplus_{j\leq n-1}\bd_{i,j}\cdot
 \bigl(J_q(h^{(0)}_{i,j})\oplus J_q(h^{(0)}_{i,j}+\sigma_{i,j}(\bs))\bigr).
\]
The last column $j=n$, where $\sigma_{i,n}(\bs)=u_i$ is public,
contributes only to the offset bit $\beta$.
Each carry term is anchored at the branch-$0$ mask $h^{(0)}_{i,j}$, the
value at which the carry pattern of adding $\sigma_{i,j}(\bs)$ is read off;
we call this value the \emph{base point} of the term, and, as the example
above shows, the same shift produces different carries at different base
points. The
linear coefficient $\ba$, the offset bit $\beta$, the bit masks
$\bd_{i,j}$, and the base points are all computable from the public
key, one extended preimage, and $\bD$, without the secret; the
displayed equation is written for branch $0$, and the branch-$1$
preimage gives an equivalent sign-flipped form. Computability
without the secret is what will turn
a successful classical prover into an adversary for our assumption. We
call each term of the double XOR, i.e., $\bd_{i,j}\cdot
 \bigl(J_q(h^{(0)}_{i,j})\oplus J_q(h^{(0)}_{i,j}+\sigma_{i,j}(\bs))\bigr)$, a
\emph{carry predicate}: a choice of a row $i$, a prefix length $j$, a
base point $\xi$ (here $\xi=h^{(0)}_{i,j}$), and a selection $\bd$ of
bit positions.

The LK assumption provides an extractor that obtains a decoded preimage from
any successful classical strategy, and randomness reconstruction extends
it to the full string $\boldzeta_b$, masks included. Indistinguishability
lets this extraction event be transferred from $g$ to $f$. Since the
verifier chooses the two functions uniformly, an average success
probability exceeding $3/4$ by more than the reduction error would imply
that, with probability correspondingly larger than $1/2$, a classical
prover both holds an extended preimage and supplies a correct equation of
the displayed form.

We therefore assume an \emph{adaptive hardcore bit with carry predicates}:
no efficient algorithm, given the public key, can produce a linear
function of the secret, a list of carry predicates of its choice, and a
prediction of the resulting parity that is correct non-negligibly more
often than a random guess, subject to an admissibility
condition that excludes degenerate choices. With no carry bits, the corresponding statement is the BCMVV adaptive-hardcore-bit
property and follows from LWE.
\textbf{Our assumption essentially states that allowing adaptively chosen carry bits does not
help.} To see what the carry bits add, consider a single carry predicate
with base point $\xi=0$: it contributes bits of $J_q(\sigma_{i,j}(\bs))$, the binary
representation of the partial sum itself. The top bit of
$\sigma_{i,j}(\bs)$ is a threshold function of the secret, far outside the
linear parities $\langle\ba,\bs\rangle\oplus\beta$ that the BCMVV
theorem controls.
Unfortunately, unlike for the ordinary adaptive hardcore bit, we were not able to derive this statement from standard LWE.
Nevertheless, we provide some evidence for why such a reduction might be true in \appref{app:encoded-ahcb-evidence}.

The carry-predicate assumption also implies that the keys of $f$ and
$g$ are indistinguishable (the \emph{key-mode indistinguishability} of
\Cref{lem:aggm-counting}, with $f$ the equation key and $g$ the image key). A
direct reduction converts a distinguisher for the two keys into a
predictor for a parity of the binary LWE secret. Thus the two
cryptographic assumptions used for classical soundness are LK extraction for
$g$ and the carry-predicate assumption for $f$, with the latter also
supplying key-mode indistinguishability.
The resulting classical soundness bound is $3/4+\negl(\lambda)$.

\subsubsection{Circuit depth}
To see why the prover can perform the evaluation in shallow depth,
let $q$ be an $L$-bit prime, where $L=\Theta(\log^2\lambda)$; the
parameters are fixed in \Cref{sec:encoded-parameters}. Each
output symbol of the
randomized encoding is the result of only a constant number of
additions, subtractions, and multiplications on $L$-bit words. With
bounded-arity gates these operations have depth $O(\log L)$, and the
encoding's constant input locality allows all symbol calculations to run
in parallel. Since $L=\Theta(\log^2\lambda)$, this gives depth
$O(\log\log\lambda)$.

The same local arithmetic can be shown to be implementable in constant
depth with fan-out gates of size $\poly(L)=\log^{O(1)}\lambda$.
Grier, Morris, and Wu show how to exactly implement fan-out of this size with
polynomial-size standard-$\QAC$ circuits~\cite{GrierMorrisWu2026}.
Together with the constant input locality, this lets all local blocks
run in parallel on disjoint registers and yields an exact constant-depth
$\QAC$ implementation. The construction uses only polynomially many
source, mask, output, and scratch words, so both implementations have
polynomial width. They perform exactly the same encoded computation,
produce the same transcript distribution, and use only final measurements.

The superposition over inputs and masks can also be prepared with
polynomially many parallel blocks of size $\poly(L)$.
This can once again be done either with bounded-arity gates in log-logarithmic
depth or with a $\QAC$ circuit.

\subsubsection{The protocol in brief}
\Cref{prot:overview} summarizes the steps described in this overview. The formal version, with the parameters
and the exact checks, is \Cref{prot:encoded}; the formal statement of
\Cref{thm:informal} is \Cref{thm:encoded-formal}. The honest quantum prover
is accepted with probability $1-\negl(\lambda)$, and, under our
assumptions, no uniform classical prover of time $2^{o(\lambda)}$ is accepted
with probability exceeding $3/4+\negl(\lambda)$.

\begin{protocolframe}
\begin{protocol}[Simplified overview protocol]
\label{prot:overview}
\leavevmode
\smallskip

\noindent\textbf{Input:} a security parameter $\lambda$.

\begin{enumerate}
 \item \textbf{Challenge.} The verifier samples, with equal probability,
  either the claw-free function
  $f(b,\bx;\boldeta)=\bA\bx+b\bu+\boldeta\pmod q$, where
  $\bu=\bA\bs+\be$ is an LWE sample, or its extractable counterpart
  $g$ of the same public form, where $[\bA\mid\bu]$ is a single
  trapdoor matrix; the two descriptions are computationally
  indistinguishable. It keeps the trapdoor and sends $(\bA,\bu)$ to
  the prover.
 \item \textbf{Response.} The prover prepares a uniform superposition
  over the branch bit $b$, the input $\bx$, the function noise
  $\boldeta$, and the encoding masks $\br,\bh$; evaluates the
  randomized encoding $\widehat F(b,\bx,\boldeta;\br,\bh)$ into an
  output register; erases the function-noise register in place; applies
  Hadamards to the registers holding $(b,\bx,\br,\bh)$; and measures
  all registers. This computation admits both a $\QNCloglog$
  implementation with one- and two-qubit gates and a constant-depth
  $\QAC$ implementation. The prover sends the measured encoded output
  $\widehat\by$, the direction $\bD$, and the bit $C$.
 \item \textbf{Verification.}
  \begin{itemize}
   \item If the verifier sampled $g$: it accepts if $\widehat\by$
    decodes to a valid image of $g$, checked using the trapdoor.
   \item If the verifier sampled $f$: it recovers the two extended
    preimages $\boldzeta_0,\boldzeta_1$ of $\widehat\by$---the binary
    encodings of the two claw preimages together with their
    reconstructed masks---using the trapdoor and randomness
    reconstruction. It accepts if $\bD$ passes the admissibility test
    and
    \[
     C=\bD\cdot(\boldzeta_0\oplus\boldzeta_1).
    \]
  \end{itemize}
\end{enumerate}
\end{protocol}
\end{protocolframe}

\section{Related Work}\label{sec:related}

\subsubsection{Sampling advantage from shallow circuits}
As mentioned in the introduction, hardness of classically sampling the
outputs of constant-depth quantum circuits was first established by Terhal
and DiVincenzo~\cite{TerhalDiVincenzo2004}. Worst-case and average-case
hardness arguments were developed for the low- but not constant-depth IQP
model~\cite{BremnerJozsaShepherd2011,BremnerMontanaroShepherd2016} and were
subsequently brought to constant depth proper, for measurement-based and
quenched-Hamiltonian
architectures~\cite{MillerSandersMiyake2017,BermejoVegaEtAl2018,
HaferkampEtAl2020}. None of these sampling tasks is known to be
efficiently verifiable. Proposals that hide a verification secret inside
an IQP circuit~\cite{ShepherdBremner2009,BremnerChengJi2025} have faced secret-extraction
attacks~\cite{KahanamokuMeyer2023,GrossHangleiter2025}, so their classical security remains
unsettled. In the black-box setting, Fourier sampling relative to a random
oracle is solvable by constant-depth circuits with oracle gates and is
unconditionally hard
classically~\cite{Aaronson2010,RazTal2022,BassirianEtAl2026}; it supports
certified randomness, but with a verifier that runs in exponential
time~\cite{BassirianEtAl2026}.

\subsubsection{Unconditional separations}
Bravyi, Gosset, and K{\"o}nig proved unconditionally that a relation
problem solvable by $\QNC$ circuits requires logarithmic depth
classically~\cite{BravyiGossetKoenig2018}. This separation was made
noise-tolerant~\cite{BravyiGossetKoenigTomamichel2020}, extended to
hardness against $\mathsf{AC}^0$ circuits~\cite{BeneWattsEtAl2019} and,
in interactive form, against $\NC$~\cite{GrierSchaeffer2020}, developed
into certifiable randomness against low-depth
adversaries~\cite{CoudronStarkVidick2021}, and generalized to a depth
hierarchy for shallow circuits~\cite{HsiehEtAl2026}. Unconditional and
exponentially large quantum advantages are also known in \emph{sample}
complexity~\cite{BenedettiBuhrmanWeggemans2026,BenedettiEtAl2025}. These
tasks are all efficiently verifiable, but their hardness holds only against
restricted classical models rather than against all polynomial-time
algorithms.

\subsubsection{Cryptographic tests of quantumness}
BCMVV introduced the LWE-based test of quantumness underlying our
result~\cite{BCMVV2021}; Mahadev's verification protocol and the
computational Bell test of Kahanamoku-Meyer et al.\ developed closely
related forms of classical verification~\cite{Mahadev2018,KMCVY2022}.
Single-round variants were obtained in the random-oracle
model~\cite{BrakerskiKoppulaVaziraniVidick2020,YamakawaZhandry2024} or fully white-box
by using knowledge assumptions~\cite{ArabadjievaEtAl2024}. On the
depth side, Hirahara and Le Gall~\cite{HiraharaLeGall2021} and Liu and
Gheorghiu~\cite{LiuGheorghiu2022} compressed the honest
prover of such protocols to constant-depth circuits interleaved with
mid-circuit measurements and classical
feed-forward, and Chia and Hung
combined these ingredients with a random oracle to certify the depth of a
quantum device~\cite{ChiaHung2023}. Compared with these works, our honest
prover runs a single white-box shallow circuit with one final measurement
layer, at the cost of less standard assumptions.

\subsubsection{Fan-out, fan-in, and shallow arithmetic}
The distinction between unbounded fan-in and unbounded fan-out goes back to
early work on parallel quantum
computation~\cite{MooreNilsson2002,GreenHomerMoorePollett2002}. H{\o}yer
and {\v S}palek showed that fan-out yields constant-depth implementations
of the arithmetic behind Shor's algorithm~\cite{HoyerSpalek2005}, building
on the logarithmic-depth parallelization of Cleve and
Watrous~\cite{cleve-watrous}, and Takahashi and Tani proved the exact collapse
of the fan-out depth hierarchy~\cite{TakahashiTani2016}. For fan-in,
Grier, Morris, and Wu showed that $\QAC$ circuits can compute $\TC$
functions when given many copies of the input~\cite{GrierMorrisWu2026};
our $\QAC$ compilation uses their exact polylogarithmic fan-out
construction. Foxman et al.\ construct pseudorandom unitaries in these
shallow classes~\cite{FoxmanEtAl2026}, and their use of $\QAC$
parallels ours: Rosenthal showed that fan-out has approximate constant-depth
$\QAC$ implementations of exponential size~\cite{Rosenthal2021}, and
Foxman et al.\ instantiate a refinement of that construction on blocks
of logarithmically many qubits, where its size becomes polynomial and
its error inverse-polynomial. In both cases $\QAC$ supplies fan-out
only at small scales---for us, the polylogarithmic arithmetic blocks of
the randomized encoding---and the construction is arranged so that no
larger fan-out is ever needed.

\subsubsection{Randomized encodings}
Randomized encodings with constant locality originate in the work of
Applebaum, Ishai, and Kushilevitz on cryptography in
$\mathsf{NC}^0$~\cite{AIK2006}. Our encoding is a $q$-ary path-mask
analog of the binary encoding recently used for constant-depth
pseudoentanglement~\cite{Gheorghiu2026}. 

\section{Discussion}\label{sec:discussion}

Our results leave open the central question raised in the introduction: can
verifiable quantum advantage be achieved with $\QNC$ circuits? The two
implementations we gave approach this question from complementary directions. The
bounded-arity implementation has depth $O(\log\log\lambda)$, while the
constant-depth implementation uses unbounded fan-in gates. Thus, neither
gives a constant-depth circuit over only one- and two-qubit gates.

There are a few ways one could imagine pushing this construction to a true $\QNC$ one.
First, the reason for the log-logarithmic depth, or the unbounded fan-in, is that we must perform
arithmetic on $L$-bit numbers, where $L=\polylog(\lambda)$. If it were possible to
construct analogous functions over a constant-size modulus, then $L=O(1)$ and we would
have a $\QNC$ implementation. Unfortunately, constructing trapdoor claw-free functions over a constant-size
modulus has proven difficult. The closest construction is the lossy trapdoor function family of Dao and Jain,
based on the Dense-Sparse Learning Parity with Noise (LPN) assumption~\cite{daojain}. In fact, this is precisely what allows
the pseudoentanglement construction of~\cite{Gheorghiu2026} to have a $\QNC$ implementation.
However, the lossy trapdoor functions are not $2$-to-$1$, and the
trapdoor there cannot invert the lossy function. Neither property is needed for a pseudoentanglement
construction, but both are needed for a verifiable test of quantumness (at least in this formulation).

A second possibility is to change the randomized encoding, rather than the TCFs. Specifically, one
would need a polynomial-length randomized encoding with the same output- and
input-locality properties at the level of \emph{individual bits} rather than individual $q$-ary symbols.
In particular, every encoded output bit should depend on only a constant number
of input and random bits, and every input or random bit should occur in only a
constant number of output formulas. To carry over the present proof directly,
the encoding should also retain efficient decoding, perfect privacy, and
efficient unique randomness reconstruction. The coherent superpositions
over the encoding randomness and function noise would have to be preparable in
constant depth with bounded-arity gates. Classical soundness would
additionally require an analog of the carry-predicate assumption for the
resulting encoding. We were not able to find a randomized encoding satisfying all of these properties.

A separate cryptographic open problem is to derive the carry-predicate version of the adaptive-hardcore-bit statement (\Cref{ass:carry-ahcb}) from standard LWE. At our parameters, the standard
BCMVV reductions yield key-mode indistinguishability and negligible
advantage for the adaptive hardcore bit on the
original preimages, i.e., the \emph{linear} version of the property we need.
However, we were unable to show tolerance to
carry predicates, which let the adversary's equation involve carry bits of
the partial sums of the LWE equations.

As mentioned, though, we do provide some evidence that the carry-predicate version of the
adaptive-hardcore-bit property should indeed hold against classical adversaries. To explain why, we first
recall what violating the assumption would mean. A successful adversary, given the
public key, chooses a linear function of the binary LWE secret, a
selection of carry bits of the partial sums
$e_i+\sum_{j'\leq j}\bA_{i,j'}s_{j'}$ evaluated at base points of
its choice, and a correct prediction of the combined parity, subject to an
admissibility condition that excludes degenerate choices. In
\appref{app:encoded-ahcb-evidence} we show that various efficient classical strategies for doing this must fail if LWE is hard. Below, $n$
denotes the length of the binary LWE secret.

First, no efficient classical algorithm can, once its base
points are chosen, predict these parities,
carry bits included, for \emph{uniformly random} challenges, non-negligibly
better than guessing.
If one could, rerunning it on many random challenges from the same saved
state and applying Goldreich--Levin
reconstruction~\cite{GoldreichLevin1989} would recover the secret in
full, thereby solving LWE.
Second, a different reduction tolerates an adversary that chooses its own
challenges once its base points are fixed, provided it can be rerun
from that point to produce many equations that are overwhelmingly often
correct and its challenge distribution is not too concentrated:
challenges that retain all but $n/4$ bits of the maximum
possible min-entropy still lead to a contradiction. An adversary
therefore cannot escape the reduction merely by biasing its choice of
challenges.
Third, we show that, perhaps unsurprisingly, if the adversary makes no
use of the carry predicates, so that its predicted bit is a linear
function of the secret, then to succeed it would have to break the
ordinary adaptive-hardcore-bit property, even when its choices are
fully adaptive. A sharper version also holds: if such an adversary's
predictions have merely inverse-polynomial bias, then having slightly more than
$n/2$ bits of min-entropy in its linear function already yields a
contradiction. The proofs turn these
predictors directly into key-mode distinguishers and adversaries against
the ordinary BCMVV
adaptive-hardcore-bit game. The underlying difficulty in
proving the full carry-predicate version outright is the nonlinear dependence of the checked
equation on the secret: the privacy and randomness-reconstruction
properties of the randomized encoding are not enough to map it back to
the BCMVV hardcore bit.

A possible way to avoid the adaptive hardcore bit altogether would be
to base the construction on the proof of quantumness of
Kahanamoku-Meyer et al.~\cite{KMCVY2022} instead of BCMVV: their
computational Bell test replaces the hardcore-bit check with a
CHSH-style test and needs only the claw-free structure of the trapdoor
family. However, that protocol has three rounds of interaction instead
of two, and it is unclear whether knowledge assumptions can collapse it
to a single round, as is possible with
BCMVV~\cite{ArabadjievaEtAl2024}.

Several two-round (four-message) proofs of quantumness avoid the
adaptive-hardcore-bit property. Kalai et al.~compile two-player nonlocal
games into single-prover protocols using QFHE~\cite{KalaiEtAl2023};
Alnawakhtha et al.~use rotated measurements to base soundness directly on
LWE~\cite{AlnawakhthaEtAl2024}; Arora et al.~use an LWE-based encrypted-CNOT
operation~\cite{AroraEtAl2024}; Miller uses LWE to hide an extended GHZ
state~\cite{Miller2024}; and the oblivious-state-preparation framework of
Bartusek and Khurana yields such a protocol from dual-mode
TCFs~\cite{BartusekKhurana2025}. The latter work also gives a more interactive
construction from any plain TCF. Other AHCB-free protocols with more
interaction include the nonlocal-game compiler of Bacho et al., which uses
any plain TCF but has round complexity linear in the security
parameter~\cite{BachoEtAl2025}, and the construction of Morimae and Yamakawa from
full-domain trapdoor permutations, which uses sequential interactive
hashing~\cite{MorimaeYamakawa2023}.

These alternatives do not appear to yield our low-depth, single-round
result directly. Several require operations beyond coherent evaluation of
an affine map---such as homomorphic evaluation, rotated-basis measurements,
an encrypted CNOT, or oblivious state preparation---and it is unclear
whether randomized encodings can parallelize those operations to our target
depth. More fundamentally, their transcripts do not have BCMVV's
preimage-versus-equation challenge structure. The existing LK extractor
therefore does not make one challenge redundant, and reducing these
protocols to a single round appears to require a new, protocol-specific
knowledge assumption. By comparison, LK has an established history, while
our remaining carry-predicate AHCB statement is falsifiable and is therefore
a more concrete target for a proof or an attack than a new knowledge
assumption.

In terms of practicality, our construction likely cannot be implemented on near-term devices due
to the large width of the circuits involved. A conservative estimate shows that the width of the circuit is
$\widetilde O(\lambda^4\log^6\lambda)$. Much of this overhead comes from
the encoding masks, together with the workspace used to perform local arithmetic
in parallel. Reducing the expansion of the randomized encoding and the
workspace needed for its coherent evaluation, without increasing the depth
or weakening security, is an important direction for future work.

\subsubsection*{\texorpdfstring{Statement on AI usage}{Statement on AI usage}}
The main ideas for the results in the paper are my own. I have, however, made extensive use of OpenAI's ChatGPT 5.6 Sol and Anthropic's
Claude Fable in deriving the proofs and in writing this paper (mainly copy-editing, improving phrasing and so on).
The carry-predicate form of \Cref{ass:carry-ahcb} was proposed by
Claude. I had initially expressed the assumption as an
adaptive-hardcore-bit property of the randomized encoding of the
original trapdoor claw-free function. I then asked Claude for a more self-contained statement, and it reformulated the assumption in terms of
carry bits of the partial sums of the LWE equations, stated directly
from the public key, with no reference to the encoding. Unfortunately, neither model
was able to help me prove it based on LWE (even in conjunction with other assumptions),
though they were also unable to disprove it. I have independently checked all the proofs.

\section{Preliminaries}\label{sec:preliminaries}

We fix notation, recall the basic notions of quantum information we
use, and define the circuit classes considered in this paper.

\subsection{Notation}

Vectorial quantities (including bit strings viewed as vectors)
and matrices are typeset in boldface. For example, $\bx,\boldeta,\bd,\bD$
are vectors, while $\bA,\bK,\bM$ are matrices. Scalars, random variables,
sets, register names, and algorithms remain in ordinary math font. Indexed
matrix entries retain the bold base symbol, as in $\bA_{i,j}$, while
coordinates of vectors are set in italic where the bold form would be
cluttered, as in $s_j$ for the $j$-th coordinate of $\bs$. Entries of
the mask arrays $\br,\bh$ are written like matrix entries, $\br_{j,i}$ and
$\bh_{i,j}$; when a branch superscript is attached we drop the boldface and
write $r^{(b)}_{j,i}$, $h^{(b)}_{i,j}$. We write
$\boldsymbol\delta_j$ for the $j$-th standard basis vector, and, for a
positive integer $N$, $[N]=\{1,\ldots,N\}$.

For an integer $q\geq2$, we write $\Zq$ for the ring of integers modulo
$q$ and $L=\lceil\log_2 q\rceil$ for its bit length (in our
construction $q$ is chosen with $2^{L-1}\le q<2^L$, so this agrees with the
parameter $L$ of \Cref{sec:encoded-parameters}). For a scalar, vector, or
array of elements of $\Zq$, $J_q(\cdot)$ denotes its entrywise canonical
$L$-bit binary encoding, concatenated in the natural order; thus
$J_q(\bx)\in\bits^{nL}$ for $\bx\in\Zq^n$. In our construction, $q$ is an
odd prime, so some $L$-bit strings represent integers at least $q$;
such a string is not the canonical encoding of any ring element, and the
verifier rejects any response in which one appears. Binary inner products are over $\mathbb F_2$. For a
residue $a\in\Zq$, $\centerq(a)$ denotes its centered integer
representative: for even $q$, it lies in $[-q/2,q/2)$, and for odd
$q$, it lies in $[-(q-1)/2,(q-1)/2]$. We apply this convention
coordinatewise to vectors.

A \emph{negligible} function $f$, denoted $f=\negl(\lambda)$, is a
function $f:\mathbb N\to\mathbb R$ such that
$f(\lambda)=o(\lambda^{-c})$ for every constant $c>0$. For
$\lambda\in\mathbb N$, we write $1^\lambda$ for the string of
$\lambda$ ones.

\subsection{Quantum information}

We briefly recall the standard quantum formalism and refer to
\cite{NielsenChuang2010} for a complete treatment. The state space of
$n$ qubits is the $2^n$-dimensional complex Hilbert space
$(\mathbb C^2)^{\otimes n}$, with the \emph{computational basis}
$\{\ket{x}:x\in\bits^n\}$ indexed by $n$-bit strings. A pure state is
a unit vector $\ket\psi=\sum_x\alpha_x\ket x$ in this space. A quantum
circuit is a sequence of \emph{gates}, unitary operators acting on one or
more qubits, followed by measurements. Measuring a register in the
computational basis returns outcome $x$ with probability
$|\alpha_x|^2$. The Hadamard gate acts on one qubit as
$H\ket b=\frac1{\sqrt2}(\ket0+(-1)^b\ket1)$ for $b\in\bits$; measuring
a register in the \emph{Hadamard basis} means applying $H$ to each of its
qubits and then measuring in the computational basis.

We use the following standard calculation from LWE-based proofs of
quantumness~\cite{BCMVV2021}. Measuring the state
$\frac1{\sqrt2}(\ket0\ket{\boldzeta_0}+\ket1\ket{\boldzeta_1})$, for
$\boldzeta_0,\boldzeta_1\in\bits^n$, entirely in the Hadamard basis yields a
uniformly random string $\bd\in\bits^n$ together with a bit $c$
satisfying $c=\bd\cdot(\boldzeta_0\oplus\boldzeta_1)$.

\subsection{Circuit classes}

We use the standard definitions of the circuit classes $\QNC$, $\QAC$,
and $\QNCf$~\cite{complexityZoo}. Throughout, these names refer to classes of
\emph{circuits} rather than classes of problems.

A quantum circuit layer is a tensor product of gates acting on pairwise
disjoint sets of qubits. The class $\QNC$ consists of polynomial-size,
constant-depth circuits over one- and two-qubit gates. We write
$\QNCloglog$ when the depth is $O(\log\log\lambda)$. The class $\QAC$
additionally permits generalized Toffoli gates
\[
 \ket{x_1,\ldots,x_t,y}\longmapsto
 \ket{x_1,\ldots,x_t,y\oplus\textstyle\bigwedge_i x_i}
\]
with an unbounded number of controls; equivalently, it permits generalized
controlled-$Z$ after a one-qubit basis change on the target. This is
unbounded \emph{fan-in}. It does not let one qubit participate in several
gates in the same layer.

The distinct class $\QNCf$ augments bounded-arity constant-depth circuits
with the unbounded \emph{fan-out} gate
\[
 \ket{s,x_1,\ldots,x_t}\longmapsto
 \ket{s,x_1\oplus s,\ldots,x_t\oplus s}.
\]
Fan-out prepares GHZ states in constant depth, so the content of one
register can feed many gates in the same layer. In particular, the
fan-out gates on polylogarithmically many qubits used by our constant-depth
arithmetic are available directly, so both of our
circuit implementations carry over to $\QNCf$. We keep
$\QAC$ and $\QNCf$ separate throughout.

Circuits are classically compiled from the public instance in polynomial
time. Ancillas start in computational-basis states, arbitrary one-qubit
gates are allowed, and all measurements occur in one final layer.

\input{sections/02_encoded_model_assumptions}

\input{sections/06_randomized_encoding}

\iftoggle{full}{
  \bibliographystyle{alpha}
}{
  \bibliographystyle{splncs04}
}
\bibliography{references}

\iftoggle{full}{
  \appendix
  \makeatletter
  \renewcommand*{\theHsection}{appendix.\Alph{section}}
  \renewcommand*{\theHsubsection}{appendix.\Alph{section}.\arabic{subsection}}
  \makeatother

  \input{sections/appendix_encoded_good_directions}
  \input{sections/appendix_encoded_ahcb_evidence}
  \input{sections/appendix_encoded_parameters}
}{}

\end{document}

%% file: sections/header.tex
\usepackage{iftex}
\ifPDFTeX
  \usepackage[utf8]{inputenc}
\fi
\usepackage[T1]{fontenc}
\usepackage{lmodern}
\usepackage{microtype}

\usepackage{amsmath,amssymb,mathtools}
\usepackage{booktabs,tabularx,array}
\usepackage[noadjust]{cite}
\usepackage{url}
\usepackage[table]{xcolor}
\usepackage{adjustbox}
\usepackage{tikz}
\usetikzlibrary{positioning}

\colorlet{tabgood}{green!60!white!25}
\colorlet{tabok}{yellow!45!white!45}
\colorlet{taborange}{orange!35}
\colorlet{tabbad}{red!55!white!30}
\usepackage[inline]{enumitem}

\usepackage[
  pdfpagelabels=true,
  \iftoggle{llncs}{}{pagebackref}
]{hyperref}
\hypersetup{
  linktoc=page,
  pdfpagemode=UseNone,
  colorlinks,
  linkcolor={red!50!black},
  citecolor={blue!50!black},
  urlcolor={blue!80!black},
  pdfsubject={Verifiable quantum advantage in extremely low depth},
  pdfkeywords={proofs of quantumness, LWE, shallow quantum circuits, randomized encodings}
}

\iftoggle{llncs}{}{
  \renewcommand*{\backref}[1]{}
  \renewcommand*{\backrefalt}[4]{
    \ifcase #1
      (Not cited.)
    \or
      (Page~#2.)
    \else
      (Pages~#2.)
    \fi
  }

}

\makeatletter
\AtBeginDocument{
  \hypersetup{pdftitle={\@title}}
}
\makeatother

\usepackage[capitalize]{cleveref}

\spnewtheorem{assumption}[theorem]{Assumption}{\bfseries}{\itshape}
\spnewtheorem{protocol}[theorem]{Protocol}{\bfseries}{\rmfamily}

\usepackage{mdframed}
\newmdenv[
  topline=true,bottomline=true,leftline=false,rightline=false,
  linewidth=0.8pt,
  innerleftmargin=0pt,innerrightmargin=0pt,
  innertopmargin=6pt,innerbottommargin=8pt,
  skipabove=12pt,skipbelow=12pt
]{protocolframe}

\crefname{assumption}{assumption}{assumptions}
\Crefname{assumption}{Assumption}{Assumptions}
\crefname{claim}{claim}{claims}
\Crefname{claim}{Claim}{Claims}
\crefname{conjecture}{conjecture}{conjectures}
\Crefname{conjecture}{Conjecture}{Conjectures}
\crefname{protocol}{protocol}{protocols}
\Crefname{protocol}{Protocol}{Protocols}
\crefname{question}{open question}{open questions}
\Crefname{question}{Open Question}{Open Questions}

\AtEndEnvironment{proof}{\qed}

\newcommand{\Zq}{\mathbb Z_q}
\newcommand{\bits}{\{0,1\}}
\newcommand{\ket}[1]{\lvert #1\rangle}

\newcommand{\poly}{\operatorname{poly}}
\newcommand{\polylog}{\operatorname{polylog}}
\newcommand{\negl}{\operatorname{negl}}
\renewcommand{\Im}{\operatorname{Im}}
\newcommand{\QNC}{\mathsf{QNC}^{0}}
\newcommand{\QAC}{\mathsf{QAC}^{0}}
\newcommand{\QNCf}{\mathsf{QNC}^{0}_{f}}
\newcommand{\QNCloglog}{\mathsf{QNC}^{0}[\log\log]}
\newcommand{\NC}{\mathsf{NC}^{1}}
\newcommand{\TC}{\mathsf{TC}^{0}}
\newcommand{\Gen}{\mathsf{Gen}}
\newcommand{\Dec}{\mathsf{Dec}}
\newcommand{\Rrc}{\mathsf{Rrc}}
\newcommand{\Adm}{\mathsf{Adm}}
\newcommand{\TrapGen}{\mathsf{TrapGen}}
\newcommand{\TrapInv}{\mathsf{TrapInv}}
\newcommand{\Eq}{\mathsf{Eq}}
\newcommand{\Img}{\mathsf{Im}}
\newcommand{\Adv}{\operatorname{Adv}}
\newcommand{\centerq}{\operatorname{center}_q}
\newcommand{\eps}{\varepsilon}
\newcommand{\bA}{\mathbf A}
\newcommand{\bR}{\mathbf R}
\newcommand{\ba}{\mathbf a}
\newcommand{\bu}{\mathbf u}
\newcommand{\bs}{\mathbf s}
\newcommand{\be}{\mathbf e}
\newcommand{\bg}{\mathbf g}
\newcommand{\bd}{\mathbf d}
\newcommand{\bx}{\mathbf x}
\newcommand{\by}{\mathbf y}
\newcommand{\bK}{\mathbf K}
\newcommand{\bM}{\mathbf M}
\newcommand{\bD}{\mathbf D}
\newcommand{\bp}{\mathbf p}
\newcommand{\br}{\mathbf r}
\newcommand{\bh}{\mathbf h}
\newcommand{\bv}{\mathbf v}
\newcommand{\bw}{\mathbf w}
\newcommand{\bz}{\mathbf z}
\newcommand{\boldeta}{\boldsymbol\eta}
\newcommand{\boldzeta}{\boldsymbol\zeta}

%% file: sections/02_encoded_model_assumptions.tex
\subsection{Proofs of quantumness}
\label{sec:model-assumptions}

We now formalize the proofs of quantumness discussed in
\Cref{sec:introduction}. We use the single-round formulation of
AGGM~\cite{ArabadjievaEtAl2024}, consisting of one challenge and one
response.

\begin{definition}[Single-round proof of quantumness, based on
\cite{BCMVV2021,ArabadjievaEtAl2024}]
\label{def:single-round-poq}
A \emph{single-round proof of quantumness} consists of a classical
generator $\Gen$, a classical verifier, and an honest quantum
prover, all parameterized by a security parameter $\lambda$.
The generator produces a pair
$(k,t)\leftarrow\Gen(1^\lambda)$, where
$k$ is the public challenge and $t$ is the private
verification data. The prover receives $k$ and returns one
classical message, after which the verifier, using $t$,
accepts or rejects.
\begin{enumerate}
 \item \textbf{Completeness.}  The protocol has completeness
  $c(\lambda)$ if the honest quantum prover is accepted with probability
  at least $c(\lambda)$.
 \item \textbf{Soundness.}  The protocol has classical soundness at most
  $s(\lambda)$ if every fixed uniform\footnote{Soundness is stated
  against uniform provers because the LK assumption below, like
  knowledge assumptions generally, is implausible against non-uniform
  adversaries with arbitrary auxiliary
  input~\cite{BitanskyEtAl2014}.} probabilistic polynomial-time
  classical prover is accepted with probability at most $s(\lambda)$ for
  all sufficiently large $\lambda$.
\end{enumerate}
\end{definition}

The honest prover's quantum circuit is generated classically from the public challenge; in
our case it admits two implementations of the same
computation: one over bounded-arity entangling gates, the other over
generalized Toffoli gates of unbounded fan-in.

\subsection{Learning with Errors}
\label{sec:encoded-prime-lwe}

In the Learning with Errors (LWE) problem~\cite{regev}, one is given a
matrix $\bA\leftarrow\Zq^{m\times n}$ and a vector
$\bA\bs+\be\bmod q$, where $\bs\in\Zq^{n}$ is secret and $\be$ is
a short error vector; the decisional variant asks one to distinguish this
pair from $(\bA,\bu)$ with $\bu$ uniform. The (subexponential) LWE assumption states that
no efficient (even subexponential-time) classical or quantum algorithm can solve either
of these two problems.

We do not use LWE as an assumption for our main result, instead relying on the 
stronger carry-predicate assumption. However, we use LWE to provide partial evidence 
for the carry-predicate assumption in \appref{app:encoded-ahcb-evidence}. There, the specific 
LWE parameters used are explained in more detail.

\subsection{Lattice trapdoors}

We use the standard LWE trapdoor generator of Micciancio and
Peikert~\cite{mp12}, in the form used by BCMVV~\cite{BCMVV2021}.

\begin{theorem}[Trapdoor generation, based on \cite{mp12}]
\label{thm:trapgen}
There are polynomial-time algorithms $(\TrapGen,\TrapInv)$ and a
constant $c_T>0$ with the
following properties. For $m\geq c_Tn\log q$,
\[
 (\bA,t_{\bA})\leftarrow\TrapGen(1^n,1^m,q)
\]
produces a matrix statistically close to uniform together with a trapdoor
for bounded-distance inversion. Except on a trapdoor-failure event, if
\[
 \by=\bA\bv+\be\pmod q,
 \qquad \bv\in\Zq^n,\quad \be\in\Zq^m,
 \qquad
 \|\centerq(\be)\|_2<
 r_n:=\frac{q}{C_T\sqrt{n\lceil\log_2 q\rceil}},
\]
then $\TrapInv(t_{\bA},\by)$ returns the unique pair
$(\bv,\be)\in\Zq^n\times\Zq^m$ satisfying this bound, for an absolute
constant $C_T>0$ determined by the generator.
\end{theorem}

We choose the constants in the generator so that its public marginal has
statistical distance at most $2^{-\Omega(n)}$ from uniform and its
trapdoor failure probability is at most $2^{-\Omega(n)}$. These explicit
rates are needed for comparison with the security bounds below, which are
subexponential in $n$.

\subsection{The knowledge-of-lattice-point assumption}
\label{sec:encoded-qlk}

The public input to the LK experiment is the complete parameter tuple
\[
 (q,\bK,n,m,
   Q,W),
\]
where $Q$ and $W$ are the input and noise domain bounds fixed in
\Cref{sec:encoded-parameters}; thus the modulus and domain bounds are available to both the adversary and
extractor. We suppress this metadata in the notation below and call
$\bK$ the lattice part of the instance. For
$\bK\in\mathbb Z_{q}^{m\times(n+1)}$, let
\[
 \mathcal L(\bK)=
 \{\bK\bv+q\bz:
   \bv\in\mathbb Z^{n+1},\
   \bz\in\mathbb Z^{m}\}.
\]
Write $\lambda_1(\mathcal L(\bK))$ and
$\lambda_1^{\infty}(\mathcal L(\bK))$ for the length of a shortest
nonzero lattice vector in the Euclidean and $\ell_\infty$ norms, respectively.
For a generated matrix $\bK$ outside the trapdoor-failure event, the
uniqueness statement in \Cref{thm:trapgen} implies
\begin{equation}
 \lambda_1(\mathcal L(\bK))\geq r_{n+1}.
 \label{eq:encoded-lattice-separation}
\end{equation}
Indeed, a nonzero lattice vector of Euclidean norm below
$r_{n+1}$ would
give two distinct bounded-error representations of its residue modulo
$q$. Consequently
$\lambda_1^{\infty}(\mathcal L(\bK))\geq
r_{n+1}/\sqrt{m}$, which is all we use. We use the online LK-$1/4$ relation
of AGGM~\cite{ArabadjievaEtAl2024}, a variant of the lattice knowledge
assumption of \cite{loftus}: when an algorithm outputs a point within
$\lambda_1^{\infty}(\mathcal L(\bK))/4$ of the lattice, an online extractor
must return the corresponding nearby lattice point.

\begin{assumption}[Subexponential LK-$1/4$, based on
\cite{loftus,ArabadjievaEtAl2024}]
\label{ass:encoded-lk}
There is a constant $c>0$ such that for every classical algorithm of
time at most
$2^{c\sqrt{n}}$ in the LK-$1/4$ experiment, there is an online
stateful
extractor with polynomial overhead and failure probability at most
$2^{-c\sqrt{n}}$, for the trapdoor-matrix ensemble of \Cref{sec:encoded-parameters}.
\end{assumption}

The assumption averages over the trapdoor key, the adversary's random tape,
and the extractor's independent random tape. It is a knowledge assumption
and is separate from LWE. Here ``online stateful extractor'' means a uniform,
straight-line classical algorithm given only the public tuple, the adversary
code and random tape, and its own independent coins; it receives no trapdoor.
Its code can therefore be run unchanged on a syntactically identical equation
key; this is used in the soundness proof of \Cref{thm:encoded-formal}.

\subsection{The AGGM counting argument}

Soundness of our protocol rests on the following counting argument, which
is the single-round argument of AGGM~\cite{ArabadjievaEtAl2024} stated
with explicit error terms.

\begin{lemma}[Counting argument, based on \cite{ArabadjievaEtAl2024}]
\label{lem:aggm-counting}
Consider a single-round proof of quantumness whose generator chooses
uniformly between two public-key modes, an \emph{image} mode and an
\emph{equation} mode, and fix a classical prover. Suppose that:
\begin{enumerate}
 \item the two key modes have distinguishing advantage at most
  $\eps_{\rm hid}$ for every procedure composed from the prover and the
  extraction procedure below;
 \item in image mode there is an extraction procedure, using only public
  data, that recovers a \emph{checked} preimage---one verified against
  the prover's output by a public validity check---whenever the verifier
  accepts, except with probability $\eps_{\rm ext}$; and
 \item in equation mode, the prover together with the extraction
  procedure produces both such a checked preimage and an
  accepting equation with probability at most
  $1/2+\eps_{\rm hc}$.
\end{enumerate}
Then the prover is accepted with probability at most
\[
 \frac34+\frac12(\eps_{\rm hid}+\eps_{\rm ext}+\eps_{\rm hc}).
\]
\end{lemma}

\begin{proof}
Let $\mathcal Q$ be the event that the extraction returns a checked
preimage and let $\mathcal E$ be equation-mode acceptance, both under an
equation key. Since the extraction procedure
uses only public data, key-mode indistinguishability transfers
$\mathcal Q$ between the two key distributions, so image-mode acceptance is at
most $\Pr[\mathcal Q]+\eps_{\rm ext}+\eps_{\rm hid}$. Moreover,
\[
 \Pr[\mathcal Q]+\Pr[\mathcal E]\leq1+\Pr[\mathcal Q\cap\mathcal E]
 \leq\frac32+\eps_{\rm hc}.
\]
Since the verifier chooses the two modes uniformly, the prover's acceptance
probability is at most
\[
 \frac12\bigl(\Pr[\mathcal Q]+\eps_{\rm ext}+\eps_{\rm hid}\bigr)
 +\frac12\Pr[\mathcal E]
 \leq\frac34
 +\frac12(\eps_{\rm hid}+\eps_{\rm ext}+\eps_{\rm hc}),
\]
which is the claimed bound.
\end{proof}

In our instantiation, $\eps_{\rm hid}$ is negligible under the
carry-predicate assumption by
\Cref{lem:encoded-key-hiding}; hence key-mode indistinguishability
contributes only a negligible term to the soundness bound.

\subsection{Randomized encodings}
\label{sec:re-prelims}

A randomized encoding of a function $F$ replaces $F(\bv)$ by the output
of a randomized function $\widehat F(\bv;\rho)$ that reveals $F(\bv)$
but nothing else about $\bv$. We use the following properties.

\begin{definition}[Randomized encoding, based on
\cite{IshaiKushilevitz2000,AIK2006}]
\label{def:randomized-encoding}
Let $F:V\to Y$ be a function. A function
$\widehat F:V\times R\to\widehat Y$ is a \emph{perfect randomized
encoding of $F$ with randomness reconstruction} if the following
conditions hold.
\begin{enumerate}
 \item \textbf{Decoding.}  There is a function $\Dec:\widehat Y\to Y$
  such that
  $\Dec(\widehat F(\bv;\rho))=F(\bv)$ for every $\bv\in V$ and
  $\rho\in R$.
 \item \textbf{Perfect privacy.}  There is a randomized simulator
  $\mathsf{Sim}$ such that, for every $\bv\in V$, the distribution of
  $\mathsf{Sim}(F(\bv))$ equals the distribution of
  $\widehat F(\bv;\rho)$ for uniform $\rho\leftarrow R$.
 \item \textbf{Randomness reconstruction.}  There is a function
  $\Rrc:V\times\widehat Y\to R$ such that, for every $\bv\in V$ and
  every
  $\widehat\by$ in the support of $\widehat F(\bv;\cdot)$,
  $\rho=\Rrc(\bv,\widehat\by)$ is the unique element of $R$ with
  $\widehat F(\bv;\rho)=\widehat\by$.
\end{enumerate}
All three functions are required to be computable in classical polynomial
time.
\end{definition}

Randomized encodings were introduced by Ishai and
Kushilevitz~\cite{IshaiKushilevitz2000} and developed further, with
constant locality, by Applebaum, Ishai, and Kushilevitz~\cite{AIK2006};
randomness reconstruction is the additional property used
in~\cite{Gheorghiu2026}. For shallow implementations we also track
locality: every output symbol should depend on only a constant number of
source symbols, and no source symbol should occur in more than a constant
number of output formulas. These two properties are the encoding's
constant \emph{output locality} and constant \emph{input locality}.

%% file: sections/06_randomized_encoding.tex
\section{\texorpdfstring{The Encoded-LWE Sampling Problem}{The Encoded-LWE Sampling Problem}}
\label{sec:encoded-construction}

In this section, we give a construction with polynomial width and nearly
exponential assumed hardness. Its security is conditional on
\Cref{ass:encoded-lk,ass:carry-ahcb}. The latter is an
adaptive-hardcore-bit assumption for LWE in which the adversary may include
carry bits of the partial sums of the LWE equations in the parity it must
predict; it also implies the
key-mode indistinguishability needed for soundness. The randomized encoding
supplies
the remaining structural properties. The prime-modulus LWE assumption of
\appref{app:encoded-ahcb-evidence} yields the ordinary
adaptive-hardcore-bit property of the decoded maps
$(b,\bx,\boldeta)\mapsto\bA\bx+b\bu+\boldeta$, used as partial evidence for
\Cref{ass:carry-ahcb}; it is not needed for the main theorem.

We fix parameters in \Cref{sec:encoded-parameters}, present the randomized
encoding and its properties in \Cref{sec:qary-encoding}, and state the
carry-predicate assumption in \Cref{sec:encoded-ahcb}. We
then present the protocol in \Cref{sec:encoded-protocol} and prove its
completeness, soundness, and circuit bounds in \Cref{sec:encoded-security}.

\subsection{Parameters}
\label{sec:encoded-parameters}

Let $\lambda$ be the security parameter; all parameters are functions
of it. Set
\begin{equation}
 \begin{aligned}
  L&=4\left\lceil(\log_2(\lambda+2))^2\right\rceil,&
  n&=2\bigl\lceil c_1(\lambda^2+\lambda+L)/2\bigr\rceil,\\
  m&=\lceil c_2(n+1)L\rceil,
 \end{aligned}
 \label{eq:encoded-dimensions}
\end{equation}
where $c_1,c_2$ are sufficiently large for the BCMVV and trapdoor
theorems. Sample a prime
\begin{equation}
 3\cdot2^{L-2}\leq q<2^{L}.
 \label{eq:encoded-prime-modulus}
\end{equation}
The lower end $3\cdot2^{L-2}$, rather than $2^{L-1}$, guarantees that
every nonzero mask parity $\bd\cdot J_{q}(\cdot)$ takes each value on at
least $2^{L-2}\geq q/4$ canonical residues: every such parity is balanced
on the full $L$-bit cube, and the residues omit fewer than $2^{L-2}$ of
its strings. The admissibility analysis of
\appref{app:encoded-good-directions} uses this.
Concretely, the setup tests $c_3\lambda L$ independent uniform odd
candidates in the range of \Cref{eq:encoded-prime-modulus}, for a sufficiently large constant $c_3$,
and, if none is prime, returns a fixed challenge on which the verifier
rejects every response. Prime
density bounds make this abort probability $2^{-\Omega(\lambda)}$, while
keeping strict polynomial worst-case generation time.

We assume LK-$1/4$ and the carry-predicate
assumption are secure for this parameter ensemble against adversaries
of time $2^{c\sqrt{n}}$, with failure or advantage at most
$2^{-c\sqrt{n}}$, for a fixed constant $c>0$; both assumptions
are thus subexponential in the
binary-secret dimension $n$. The LWE
assumption used only for the partial-evidence results is stated in
\appref{app:encoded-ahcb-evidence}. Since
$n=\Theta(\lambda^2)$, the assumed security level is
$2^{\Theta(\lambda)}$ in the security parameter. The
modulus
$q=2^{\Theta(\log^2\lambda)}$ is quasipolynomial in the security
parameter, while its bit length remains polylogarithmic. Together with the
two superpolynomial noise gaps verified in \appref{app:encoded-parameters}, its primality places the
parameters in the BCMVV regime. A uniform
$\mathbb Z_{q}$ register has the exact shallow preparation proved
in \Cref{lem:prime-uniform-preparation} (\Cref{sec:encoded-security}).

Let $C_T$ be the constant of \Cref{thm:trapgen}, and define
\[
 r_G=
 \frac{q}
 {C_T\sqrt{(n+1)L}},\qquad
 P=\frac{r_G}{2m}.
\]
Since $L=\lceil\log_2 q\rceil$, the radius $r_G$ equals the inversion
radius $r_{n+1}$ of \Cref{thm:trapgen} at dimension $n+1$; in
particular $r_G<r_n$. Further define
\begin{equation}
 j_W=\max\!\left(\{1\}\cup
 \left\{j\in\mathbb Z_{\geq1}:
 4m^2C_T^2(n+1)L\,2^{2j}
 \leq q^2\right\}\right),
 \qquad W=2^{j_W}.
 \label{eq:encoded-width-definition}
\end{equation}
This defines a positive integer for every $\lambda$. For all sufficiently
large $\lambda$, the second set in \Cref{eq:encoded-width-definition}
is nonempty, and maximality gives
$P/2<W\leq P$. Set
\[
 Q=2^{L/2},\qquad
 B_V=2^{L/4}.
\]
Let
\begin{align*}
 E
  &=\{z-W/2\bmod q:
          z=0,\ldots,W-1\},\\
 X
  &=\{0,\ldots,Q-1\}^{n},\\
 X_0
  &=\{1,\ldots,Q-1\}^{n},\qquad
 X_1
  =\{0,\ldots,Q-2\}^{n}.
\end{align*}
Thus $E$ is a centered interval of power-of-two size, but it
is not a contiguous interval in the canonical binary ordering. Its uniform
state is prepared exactly by applying Hadamards to an unsigned index
$z\in\{0,\ldots,W-1\}$ and reversibly replacing it by
$\eta=z-W/2\pmod{q}$. This public translation has
bounded-fan-in depth $O(\log L)$ and constant $\QAC$ depth.

The calculation in \appref{app:encoded-parameters}\iftoggle{full}{ (\Cref{lem:parameter-estimates})}{} gives
\begin{equation}
 n=\Theta(\lambda^2),\qquad
 m=\Theta(\lambda^2L),\qquad
 \frac{mB_V}{W}
   +\frac{2n}{Q}
 =\negl(\lambda).
 \label{eq:encoded-scales}
\end{equation}

For the key error we use exactly the BCMVV truncated discrete Gaussian
$D_{\mathbb Z_{q},B_V}$~\cite{BCMVV2021}: its
centered weights are
proportional to $\exp(-\pi z^2/B_V^2)$ on
$|z|\leq B_V$, and zero outside this interval. The classical
key generator samples this finite distribution to statistical
distance at most $2^{-\Omega(n)}$ in polynomial time, using a
sampler whose output always lies in the support $|z|\leq B_V$; this
sampling error is included in the negligible
key-generation error below.
In equation mode, sample
\[
 (\bA,t_\bA)\leftarrow\TrapGen(1^{n},
       1^{m},q),\qquad
 \bs\leftarrow\bits^{n},\qquad
 \be\leftarrow
 D_{\mathbb Z_{q},B_V}^{m},\qquad
 \bu=\bA\bs+\be,
\]
so every coordinate of $\be$ has centered magnitude at most
$B_V$. The generator additionally samples \emph{private shifts}
$\varpi_{i,\tau}\leftarrow\Zq\setminus\{0\}$ independently for
$i\in[m]$ and $\tau\in[N_{\rm sh}]$, where
$N_{\rm sh}=\lceil\sqrt{n}\rceil$; they are used only by the
admissibility predicate of \Cref{def:admissible} and are retained in the
private data $t=(t_\bA,\bs,\be,\boldsymbol\varpi)$. The decoded map, computable from the public key, is
\[
 (b,\bx,\boldeta)\longmapsto
 \by=\bA\bx+b\bu+\boldeta\pmod{q},
 \qquad
 \boldeta\in E^{m}.
\]
In image mode,
$(\bK,t_\bK)\leftarrow\TrapGen(1^{n+1},
1^{m},q)$, and the
$m\times(n+1)$ matrix is parsed as
$\bK=[\bA\mid\bu]$; the decoded map has the same form. In either mode,
the complete public key is
\[
 k=(q,n,m,
       Q,W,\bA,\bu).
\]
The required indistinguishability of these two key modes follows from the
carry-predicate assumption (\Cref{lem:encoded-key-hiding}).

For an equation key, trapdoor inversion gives the two original preimages
related by $\bx_1=\bx_0-\bs$. We retain only the truncated domains
$X_0$ and $X_1$; this loses at most
$2n/Q$ honest probability. Shifting the function
noise by $\be$ loses at most
$mB_V/W$. In image mode,
$W\leq r_G/(2m)$ places every valid output
inside the trapdoor and LK-$1/4$ decoding radii, as verified explicitly in
\Cref{lem:encoded-image-extraction}.

\subsection{The encoding}
\label{sec:qary-encoding}

Here, we construct a $\mathbb Z_q$ analog of the encoding for binary
linear maps used in recent work on constant-depth pseudoentanglement
\cite{Gheorghiu2026}.

Let $\bM\in\mathbb Z_{q}^{m\times(n+1)}$. For
$\bv\in\mathbb Z_{q}^{n+1}$, sample independent uniform masks
\[
 \br=(\br_{j,i})_{j\in[n+1],\,i\in[m]},
 \qquad
 \bh=(\bh_{i,j})_{i\in[m],\,j\in[n]},
\]
and set $\bh_{i,0}=\bh_{i,n+1}=0$. Define
$\mathsf{RE}_{\bM}(\bv;\br,\bh)=(\bw,\bz)$ by
\begin{align}
 \bw_{j,0}&=\bv_j+\br_{j,1},&
 \bw_{j,i}&=\br_{j,i}-\br_{j,i+1}\quad(i\in[m-1]),
 \label{eq:re-w}\\
 \bz_{i,j}&=\bh_{i,j-1}-\bh_{i,j}+\bM_{i,j}\br_{j,i}\quad(j\in[n+1]).
 \label{eq:re-z}
\end{align}
All operations are modulo $q$. The $\bh$-masks are called \emph{path
masks}: the symbols $\bz_{i,1},\ldots,\bz_{i,n+1}$ of row $i$ form a
path, and $\bh_{i,j}$ sits on the edge between the $j$-th and
$(j+1)$-st symbols, subtracted from the first and added to the second,
so that summing along the path cancels every mask and leaves the row
sum.

\begin{lemma}[Properties of the linear-map encoding]
\label{lem:re-properties}
The map in \Cref{eq:re-w,eq:re-z} is a perfect randomized encoding of
$\bv\mapsto \bM\bv$ with randomness reconstruction, in the sense of
\Cref{def:randomized-encoding}. Moreover,
each output symbol depends on at most three source symbols, and every
source symbol occurs in only a constant number of output formulas.
\end{lemma}

\begin{proof}
From $\bw$, compute
\[
 \widetilde{\bv}_{j,i}(\bw)=\bw_{j,0}-\sum_{i'=1}^{i-1}\bw_{j,i'}
           =\bv_j+\br_{j,i}.
\]
The path masks telescope, so
\[
 \sum_{j=1}^{n+1}\bz_{i,j}
 =\sum_{j=1}^{n+1}\bM_{i,j}\br_{j,i}.
\]
Consequently the decoder returns
\[
 \sum_j \bM_{i,j}\widetilde{\bv}_{j,i}(\bw)-\sum_j\bz_{i,j}=(\bM\bv)_i.
\]

For fixed $\bv$, the map $\br\mapsto \bw$ is a bijection. Conditioned on
$\bw$, each row of $\bz$ is uniform subject to the single constraint imposed
by the decoded value. This gives a perfect simulator from $\bM\bv$, as in the
path-mask argument above.

Given $\bv,\bw$, reconstruct
\[
 \br_{j,1}=\bw_{j,0}-\bv_j,\qquad
 \br_{j,i+1}=\br_{j,i}-\bw_{j,i}.
\]
Given $\br,\bz$, reconstruct the path masks from
\[
 \bh_{i,j}=\bh_{i,j-1}+\bM_{i,j}\br_{j,i}-\bz_{i,j}.
\]
The last coordinate checks $\bh_{i,n+1}=0$. This proves randomness
reconstruction. Locality follows directly from
\Cref{eq:re-w,eq:re-z}. For occurrence, $\bv_j$ occurs once, each
$\br_{j,i}$ occurs in at most two $\bw$-formulas and one $\bz$-formula,
and each $\bh_{i,j}$ occurs twice.
\end{proof}

Suppressing the metadata $(q,n,m,Q,W)$ of $k$ in the notation, put
$\bM_k=[\bA\mid \bu]$ and
$\bv=(\bx,b)$. To add function noise
$\boldeta\in E^{m}$, replace
\begin{equation}
 \bz_{i,1}\quad\text{by}\quad \bz_{i,1}-\boldeta_i.
 \label{eq:encoded-noise}
\end{equation}
The decoder then returns $\bM_k\bv+\boldeta=\bA\bx+b\bu+\boldeta$.
Each noise symbol $\boldeta_i$ occurs in exactly one output symbol, so
the locality and occurrence bounds of \Cref{lem:re-properties} persist for
the noisy encoding. We write
\[
 \widehat F_k(b,\bx,\boldeta;\br,\bh)
\]
for the complete encoded output, and parse such an output as
$\widehat{\by}=(\bw,\widehat{\bz})$, where
$\widehat\bz_{i,1}=\bz_{i,1}-\boldeta_i$ and the other coordinates are
unchanged. Given $b,\bx$ and $\widehat\by$, first decode $\by$ and
compute
$\boldeta=\by-\bA\bx-b\bu$. Before invoking the reconstruction algorithm
of \Cref{lem:re-properties}, restore the noiseless symbol by setting
\begin{equation}
 \bz_{i,1}=\widehat\bz_{i,1}+\boldeta_i
 \quad\text{for every }i.
 \label{eq:restore-encoded-noise}
\end{equation}
Randomness reconstruction applied to $(\bw,\bz)$ then recovers the unique
$(\br,\bh)$, including the final-coordinate consistency checks. Thus the
complete coins $(\boldeta,\br,\bh)$ are efficiently recoverable; applying
the noiseless reconstruction directly to
$(\bw,\widehat\bz)$ would be incorrect.

\begin{lemma}[Privacy of the noisy encoding]
\label{lem:noisy-re-privacy}
For uniform independent $\br,\bh$, the distribution
\[
 \widehat F_k(b,\bx,\boldeta;\br,\bh)
\]
depends on
$(b,\bx,\boldeta)$ only through the decoded value
$\bA\bx+b\bu+\boldeta$. Moreover, for a fixed original input and a fixed
encoded output in its support, the encoding coins are unique.
\end{lemma}

\begin{proof}
Fix $\bv=(\bx,b)$. As in \Cref{lem:re-properties}, the map
$\br\mapsto\bw$ is a bijection. Conditioned on $\bw$, every row of
$\widehat\bz$ is uniform subject to the single constraint that its decoder
equals $(\bM_k\bv+\boldeta)_i$: subtracting $\boldeta_i$ from one path
symbol changes only that constraint. This gives the same perfect simulator
from the decoded value for every $(\bv,\boldeta)$ in its fiber.
Uniqueness follows by computing $\boldeta$, restoring
\Cref{eq:restore-encoded-noise}, and using the reconstruction formulas in
\Cref{lem:re-properties}.
\end{proof}

After writing the encoded output, the honest prover clears the
function-noise register through the row-local operation
\begin{equation}
 \boldeta_i\longleftarrow
 \boldeta_i+\widehat{\bz}_{i,1}+\bh_{i,1}
       -\bM_{i,1}\br_{1,i}\pmod{q}.
 \label{eq:eta-erasure}
\end{equation}
\Cref{eq:encoded-noise} makes the new value identically zero on
every honest basis state.

\subsection{The carry-predicate assumption}
\label{sec:encoded-ahcb}

We first fix the validity relation used by both the verifier and the
security analysis, then define carry predicates and admissibility and state
the assumption. We then decompose the verifier's parity into the form the
assumption covers and derive key-mode indistinguishability.

\begin{definition}[Extended preimages and the validity relation]
\label{def:validity-relation}
Let
\[
 \boldzeta=(J_{q}(\bx),
          J_{q}(\br),J_{q}(\bh))
 \in\bits^{N_{\rm ext}},
 \qquad
 N_{\rm ext}=L\bigl(n
   +m(2n+1)\bigr),
\]
denote the canonical binary encoding of the registers other than the branch
bit; a \emph{complete extended preimage} is a pair $(b,\boldzeta)$.
On input $(k,\widehat\by,b,\bx)$, let
$\mathsf{Rec}_k$ decode $\widehat\by$ to $\by$, form the centered
residual
\[
 \boldeta=\by-\bA\bx-b\bu\pmod{q},
\]
reject unless $\boldeta\in E^{m}$, restore the
symbols in \Cref{eq:restore-encoded-noise}, and run randomness
reconstruction. It returns $(\boldeta,\br,\bh)$ only if all reconstruction
checks pass and recomputing
$\widehat F_k(b,\bx,\boldeta;\br,\bh)$ gives exactly
$\widehat\by$. Define
\begin{equation}
 \widehat R_k(\widehat\by;b,\boldzeta)=1
 \label{eq:encoded-valid-relation}
\end{equation}
if and only if $\boldzeta$ parses canonically as
$(J_{q}(\bx),J_{q}(\br),
J_{q}(\bh))$, $\bx\in X_b$, and
$\mathsf{Rec}_k(\widehat\by,b,\bx)=(\boldeta,\br,\bh)$ succeeds.
\end{definition}

\begin{definition}[Paired images]
\label{def:paired-images}
For an equation key with private data
$t=(t_\bA,\bs,\be,\boldsymbol\varpi)$, define
$\mathsf{ExtInv}_t(\widehat\by)$ as follows. Decode to $\by$, apply
$\TrapInv(t_\bA,\by)$ and
$\TrapInv(t_\bA,\by-\bu)$, and obtain candidates
$(\bx_0,\boldeta_0)$ and $(\bx_1,\boldeta_1)$. Check
\[
 \bx_1=\bx_0-\bs,\qquad
 \boldeta_0=\boldeta_1+\be,
\]
$\bx_0\in X_0$, $\bx_1\in X_1$, and the two instances of
\Cref{eq:encoded-valid-relation}; use $\mathsf{Rec}_k$ to fill in the
encoding coins. On success, $\mathsf{ExtInv}_t$ returns
$(\boldzeta_0,\boldzeta_1)$. A \emph{valid paired image} is an encoded
output on which this algorithm succeeds.
\end{definition}

For an equation key, the trapdoor recovers both extended preimages of an
honest encoded output; the algorithm in \Cref{def:paired-images} makes this
explicit. Honest encoded outputs are valid paired images except on the
boundary, shifted-noise-overlap, and trapdoor-failure events quantified
below.

The verifier's equation test checks a parity of
$\boldzeta_0\oplus\boldzeta_1$ along a direction produced by the prover's
final Hadamard measurement. These parities have a useful structure.
Because the secret is binary, the $\bx$- and $\br$-registers of the two
extended preimages differ by shifts in $\{0,\pm1\}$, so their
contribution to any parity is a linear function of $\bs$ with
coefficients computable from the transcript and either extended
preimage. The $\bh$-registers differ
by partial sums of the LWE equations, and their contribution consists of
binary carry bits. The assumption below states that parities of this form
are unpredictable. It is stated directly in terms of the public key, with
no reference to the encoding.

\begin{definition}[Carry predicates]
\label{def:carry-predicates}
Fix a public key $k$. For $\bs'\in\bits^{n}$,
$i\in[m]$, and $j\in[n]$, put
\begin{equation}
 \sigma_{i,j}(\bs')=u_i-\sum_{j'=j+1}^{n}\bA_{i,j'}s'_{j'}
 \pmod{q}.
 \label{eq:partial-sums}
\end{equation}
A \emph{carry predicate} is a tuple $p=(i,j,\xi,\bd)$ with
$i\in[m]$, $j\in[n-1]$,
$\xi\in\mathbb Z_{q}$ (the \emph{base point}), and
$\mathbf 0\neq\bd\in\bits^{L}$. A \emph{predicate list}
$\mathcal P$ is a finite set of carry predicates with pairwise distinct
index pairs $(i,j)$. For a sign $\chi\in\{\pm1\}$, define
$\Phi^{\chi}_{\mathcal P}:\bits^{n}\to\bits$ by
\begin{equation}
 \Phi^{\chi}_{\mathcal P}(\bs')
 =\bigoplus_{(i,j,\xi,\bd)\in\mathcal P}
 \bd\cdot\bigl(J_{q}(\xi)\oplus
 J_{q}\bigl(\xi+\chi\,\sigma_{i,j}(\bs')\bigr)\bigr).
 \label{eq:carry-parity}
\end{equation}
\end{definition}

The map $\sigma_{i,j}$ of \Cref{eq:partial-sums} is defined from the public key alone. At
$\bs'=\bs$ it equals the partial sum
$e_i+\sum_{j'\leq j}\bA_{i,j'}s_{j'}$ of the $i$-th LWE equation,
while $\sigma_{i,n}(\bs')=u_i$ is public, which is why the
terminal index $j=n$ is excluded from predicate lists. The string
$J_{q}(\xi)\oplus J_{q}(\xi+\omega)$ consists of the bit positions that
flip when $\omega$ is added to $\xi$ modulo $q$; away from the modular wrap
it equals $J_{q}(\omega)$ XORed with the vector of carry bits of the
addition, so $\Phi^{\chi}_{\mathcal P}(\bs)$ is an XOR of selected bits of
the partial sums and of their carries, evaluated at base points of the
adversary's choice. The function noise appears nowhere: the error $\be$
enters only through the public $\bu$.
In particular, the adversary can cancel the key error exactly from
the secret-dependent argument of an individual predicate: choosing
$\xi=-\chi u_i$ gives
$\xi+\chi\sigma_{i,j}(\bs)=-\chi\sum_{j'>j}\bA_{i,j'}s_{j'}$.
Thus neither the presence of $e_i$ nor a base point chosen after seeing the
public key is, by itself, a source of unpredictability.

The assumption restricts the adversary's choices through an admissibility
condition. For an integer $\nu$, call a string $\mathbf c\in\bits^\nu$
\emph{balanced} if
\[
 \frac{\nu}{3}\leq\operatorname{wt}(\mathbf c)\leq\frac{2\nu}{3},
\]
where $\operatorname{wt}(\mathbf c)$ denotes the Hamming weight of
$\mathbf c$. Let $\mathcal J=\{3L,\ldots,n-3L\}$, which is nonempty for all
sufficiently large parameters. Columns within $3L$ of either end are not
used to certify row sensitivity: given the public key, the partial sum
$\sigma_{i,j}(\bs)$ is $e_i$ plus at most $3L$ secret-dependent terms for
$j<3L$ and $u_i$ minus at most $3L$ such terms for $j>n-3L$, so predicates
placed there depend on few secret coordinates. The proof of
\Cref{prop:encoded-good-density} uses only that $\mathcal J$ contains all
but $O(L)$ of the columns, and $3L$ is a convenient choice. Let
\[
 \mathcal R=\{\,\pm2^a\bmod q:
             0\leq a\leq L/4\,\}.
\]
We split a bit vector indexed by $[n]$ into two \emph{halves}, its
first and last $n/2$ coordinates, and a bit vector indexed by
$[m]$ into halves after the first
$\lfloor m/2\rfloor$ coordinates.

\begin{definition}[Admissible pairs]
\label{def:admissible}
Fix an equation key $k$ with private data $t=(t_\bA,\bs,\be,\boldsymbol\varpi)$, a sign
$\chi$, a vector $\ba\in\bits^{n}$, and a predicate list
$\mathcal P$. Put
$\Psi(\bs')=\langle\ba,\bs'\rangle\oplus\Phi^{\chi}_{\mathcal P}(\bs')$ and
$\bs^{(j)}=\bs\oplus\boldsymbol\delta_j$, and define
\begin{align}
 S_j&=\Psi(\bs)\oplus \Psi(\bs^{(j)})
 \qquad(j\in[n]),
 \label{eq:adm-secret-profile}\\
 \delta_{i,j}(\rho)&=
 \begin{cases}
  \bd\cdot\bigl(J_{q}(\xi+\chi\sigma_{i,j}(\bs))\oplus
               J_{q}(\xi+\chi\sigma_{i,j}(\bs)+\chi\rho)\bigr)
   &\text{if }(i,j,\xi,\bd)\in\mathcal P,\\
  0&\text{otherwise,}
 \end{cases}
 \label{eq:adm-error-profile}
\end{align}
for $i\in[m]$, $j\in[n-1]$, and $\rho\in\Zq\setminus\{0\}$: the bit
$\delta_{i,j}(\rho)$ is the change in the predicate at $(i,j)$ when the
$i$-th key coordinate $u_i$, equivalently the error coordinate $e_i$, is
shifted by $\rho$. For a nonempty interval
$I\subseteq[n-1]$, a set of consecutive columns, put
$\Delta^{I}_{\rho,i}=\bigoplus_{j\in I}\delta_{i,j}(\rho)$.
For a window $V\in\{[n-1],\mathcal J\}$, the carry phase of row $i$ and
its secret gradient are
\[
 \phi^{V}_i(\bs')=\bigoplus_{\substack{(i,j,\xi,\bd)\in\mathcal P\\ j\in V}}
 \bd\cdot\bigl(J_{q}(\xi)\oplus J_{q}(\xi+\chi\sigma_{i,j}(\bs'))\bigr),
 \qquad
 \Gamma^{V}_{i,k}=\phi^{V}_i(\bs)\oplus\phi^{V}_i(\bs^{(k)})
 \quad(k\in[n]);
\]
$\phi^{V}_i$ contains the carry part only, not the linear term
$\langle\ba,\cdot\rangle$. For a secret half $h\in\{1,2\}$, let
$E^{V}_{i,h}=\{k\text{ in half }h:\ k-1\in V,\ \bA_{i,k}\neq0\}$.
The public key is \emph{regular} if $|E^{V}_{i,h}|\geq\lfloor n/4\rfloor$
for every $i\in[m]$, $V$, and $h$, and every column of $\bA$ has at most
$m/8$ zero entries; regularity depends only on $\bA$.
The bit $S_j$ is the directional derivative of $\Psi$ at $\bs$ in the
direction $\boldsymbol\delta_j$.
The pair $(\ba,\mathcal P)$ satisfies the \emph{sensitivity condition}
for $\chi$ if the key is regular and the following four conditions
hold.
\begin{enumerate}
 \item \textbf{Balance.} Both halves of $\ba$ and both halves of
  $(S_j)_{j\in[n]}$ are balanced.
 \item \textbf{Interval shifts.} For every nonempty interval
  $I\subseteq[n-1]$, both halves of the vector
  $\bigl(\Delta^{I}_{\rho,i}\bigr)_{i\in[m]}$ are balanced, for every
  $\rho\in\mathcal R$, and likewise for the vector
  $\bigl(\Delta^{I}_{\varpi_{i,\tau},i}\bigr)_{i\in[m]}$ for every
  $\tau\in[N_{\rm sh}]$.
 \item \textbf{Row gradients.} For every $i\in[m]$, every window $V$,
  and both halves $h$, the string $(\Gamma^{V}_{i,k})_{k\in E^{V}_{i,h}}$
  has Hamming weight between $|E^{V}_{i,h}|/3$ and
  $2|E^{V}_{i,h}|/3$.
 \item \textbf{Rank.} For every $\rho\in\mathcal R$ and every
  $\tau\in[N_{\rm sh}]$, the $m\times(n-1)$ matrices
  $\bigl(\delta_{i,j}(\rho)\bigr)_{i,j}$ and
  $\bigl(\delta_{i,j}(\varpi_{i,\tau})\bigr)_{i,j}$ have
  $\mathbb F_2$-rank $n-1$ on each half of their rows, and, for every
  window $V$ and secret half $h$, the matrix
  $\bigl(\Gamma^{V}_{i,k}\bigr)_{i,k}$ with columns restricted to
  $\{k\text{ in half }h:k-1\in V\}$ has full column rank on each half
  of its rows.
\end{enumerate}
It satisfies the
\emph{linearity condition} if $\mathcal P=\emptyset$, so that
$\Psi(\bs')=\langle\ba,\bs'\rangle$ is linear, and $\ba$ is nonzero on
the half designated for $\chi$: the first $n/2$
coordinates for $\chi=+1$ and the last $n/2$ for
$\chi=-1$. Write $\Adm(k,t,\chi,\ba,\mathcal P)=1$ if
$(\ba,\mathcal P)$ satisfies either condition.
\end{definition}

The interval family subsumes the single window $\mathcal J$, which is
itself an interval; the shifts $\varpi_{i,\tau}$ are known only to the
verifier; and the linearity condition is unchanged, so the case
$\mathcal P=\emptyset$ of the assumption below, and with it
\Cref{lem:encoded-key-hiding}, is unaffected.
\appref{app:encoded-good-directions} motivates each layer of the
sensitivity condition by an attack that succeeds without it.
The private shifts and rank requirements are tests against the
specific cancellation witnesses analyzed there; their being hidden from
the prover does not turn honest-direction density into hardness for an
adversarially selected direction.

\begin{assumption}[Adaptive hardcore bit with carry predicates]
\label{ass:carry-ahcb}
Let $\Gen_{\Eq}$ denote the equation-key generation of
\Cref{sec:encoded-parameters}. Define
\begin{align*}
 H^{\rm carry}_{k,t}&=\bigl\{(\chi,\ba,\mathcal P,C):
   \Adm(k,t,\chi,\ba,\mathcal P)=1,\
   C=\langle\ba,\bs\rangle\oplus\Phi^{\chi}_{\mathcal P}(\bs)\bigr\},\\
 \overline H^{\rm carry}_{k,t}&=\bigl\{(\chi,\ba,\mathcal P,C):
   \Adm(k,t,\chi,\ba,\mathcal P)=1,\
   C=1\oplus\langle\ba,\bs\rangle\oplus\Phi^{\chi}_{\mathcal P}(\bs)\bigr\}.
\end{align*}
There is a constant $c>0$ such that for every quantum algorithm
$\mathcal A$ of time at most $2^{c\sqrt{n}}$,
\[
  \left|
   \Pr_{(k,t)\leftarrow\Gen_{\Eq}(1^\lambda)}
       [\mathcal A(k)\in H^{\rm carry}_{k,t}]
   -
   \Pr_{(k,t)\leftarrow\Gen_{\Eq}(1^\lambda)}
       [\mathcal A(k)\in\overline{H}^{\rm carry}_{k,t}]
  \right|
  \leq 2^{-c\sqrt{n}}.
\]
In particular, $\mathcal A$ may choose the sign, the linear function, the
predicate list, and $C$ adaptively after seeing the public key; the
private data $t$ is used only to evaluate membership in the two sets.
As for the BCMVV adaptive hardcore bit, the honest quantum prover does
not contradict the assumption: the induced data of a transcript is computed
from $\boldzeta_b$, which the prover's Hadamard measurement consumes in
producing $\bD$ and $C$. The soundness proof of \Cref{thm:encoded-formal}
uses the assumption only against classical algorithms.
\end{assumption}

The admissibility predicate is deliberately secret-dependent: the
adversary sees only $k$, while membership in the two hardcore-bit sets is
evaluated using $t$.  Moreover, the adversary may choose one direction in
one execution after seeing $k$, so the high acceptance probability of a
uniform honest direction does not control its choice.  The security
statement above explicitly assumes hardness even with this private filter
and this one-shot adaptive choice; admissibility alone does not establish
either conclusion.

For $\mathcal P=\emptyset$, the game asks the adversary to predict a
self-chosen linear function of the binary secret, and the linearity
condition is the retained-branch part of the admissibility requirement
$\bd\in G_{k,0,\bx_0}\cap G_{k,1,\bx_1}$ of the adaptive-hardcore-bit
theorem of BCMVV~\cite{BCMVV2021}; the direction set of
\Cref{eq:direction-set} imposes it at both branches (see
\appref{app:encoded-ahcb-evidence}). That case of the assumption is
supported at the usual negligible-advantage level by the
subexponential prime-modulus LWE assumption, as shown in
\appref{app:encoded-ahcb-evidence}.  That reduction does not supply the
specific $2^{-c\sqrt n}$ advantage bound in \Cref{ass:carry-ahcb}; this
quantitative rate remains assumed even for empty predicate lists.  The
additional qualitative content of the assumption is that adaptively chosen
carry bits of the partial sums do not help. The carries make
$\Phi^{\chi}_{\mathcal P}$ nonlinear in $\bs$, and this nonlinearity is
what separates the assumption from the BCMVV theorem.
The balance, interval, gradient, private-shift, and rank conditions
are designed to obstruct the concrete carry-cancellation and padding
attacks analyzed in \appref{app:encoded-good-directions}, among them a
boundary-padding predictor against the weaker test using only the window
$\mathcal J$ and public shifts.  The appendix shows that those particular
witnesses fail a check; it does not show that the checks exclude every
exact or correlated cancellation.  Passing them is not proved to imply
hardness, and we do not claim that standard LWE implies this assumption.

Because the underlying modulus is prime and the error scales follow the
BCMVV hierarchy, the prime-modulus LWE assumption of
\appref{app:encoded-ahcb-evidence} supplies key-mode indistinguishability
and the ordinary adaptive-hardcore-bit property for the decoded maps
directly; we prove this, together with three restricted consequences,
there. They rule out a classical
predictor that commits to its base points and can then be rerun from one
saved state on independent uniform challenges\iftoggle{full}{ (\Cref{prop:encoded-uniform-challenge})}{}; a classical sampler that can
be repeatedly invoked from the same state and satisfies a pointwise
conditional min-entropy bound and overwhelming equation correctness\iftoggle{full}{ (\Cref{prop:encoded-high-entropy-sampler})}{}; and
every adversary restricted to empty predicate lists, in this last case with
full adaptivity\iftoggle{full}{ (\Cref{prop:linear-predicates})}{}. None of these restricted statements establishes the fully
adaptive game above. \Cref{ass:carry-ahcb} remains an independent
assumption.

We now connect the assumption to the protocol. Parse a direction as
\[
 \bD=\bigl((\bD^x_j)_{j\in[n]},
                  (\bD^r_{j,i})_{j\in[n+1],i\in[m]},
                  (\bD^h_{i,j})_{i\in[m],j\in[n]}
       \bigr),
\]
where every displayed block is an $L$-bit string. Fix a valid
paired image and a branch $b$, write $\chi_b=(-1)^b$, and parse
$\boldzeta_b=(J_{q}(\bx_b),J_{q}(\br^{(b)}),
J_{q}(\bh^{(b)}))$. The \emph{induced data} of the transcript
at branch $b$ is $(\ba_b(\bD),\beta_b(\bD),\mathcal P_b(\bD))$, where
\begin{align}
 \ba_{b}(\bD)_j
 &=\bD^x_j\cdot\bigl(J_{q}(x_{b,j})\oplus
   J_{q}(x_{b,j}-\chi_b)\bigr)
  \oplus\bigoplus_{i}\bD^r_{j,i}\cdot
   \bigl(J_{q}(r^{(b)}_{j,i})\oplus
   J_{q}(r^{(b)}_{j,i}+\chi_b)\bigr),
 \label{eq:induced-coefficient}\\
 \beta_b(\bD)
 &=\bigoplus_i \bD^r_{n+1,i}\cdot
   \bigl(J_{q}(r^{(b)}_{n+1,i})\oplus
   J_{q}(r^{(b)}_{n+1,i}-\chi_b)\bigr)
  \oplus\bigoplus_i \bD^h_{i,n}\cdot
   \bigl(J_{q}(h^{(b)}_{i,n})\oplus
   J_{q}(h^{(b)}_{i,n}+\chi_bu_i)\bigr),
 \label{eq:induced-offset}
\end{align}
and
$\mathcal P_b(\bD)
 =\{(i,j,h^{(b)}_{i,j},\bD^h_{i,j}):
   \bD^h_{i,j}\neq\mathbf 0,\ j\leq n-1\}$.
All three are computable in deterministic polynomial time from
$(k,b,\boldzeta_b,\bD)$.

\begin{lemma}[Decomposition]
\label{lem:decomposition}
For every valid paired image, every branch $b\in\bits$, and every
direction $\bD$,
\[
 \bD\cdot(\boldzeta_0\oplus\boldzeta_1)
 =\langle\ba_b(\bD),\bs\rangle
 \oplus\beta_b(\bD)
 \oplus\Phi^{\chi_b}_{\mathcal P_b(\bD)}(\bs).
\]
\end{lemma}

The proof is in \appref{app:encoded-good-directions}, together with a
converse\iftoggle{full}{ (\Cref{lem:realizability})}{}: every pair $(\ba,\mathcal P)$ arises, with
$\beta_b(\bD)=0$, as the induced data of a transcript that the
adversary can generate itself. The assumption therefore grants the
adversary at least the freedom that the protocol exposes, and every pair it
ranges over is one that the adversary can realize at some branch.

The verifier accepts directions whose induced data is admissible at both
branches:
\begin{equation}
 \widehat G_{k,t,\boldzeta_0,\boldzeta_1}
 =\bigl\{\bD:\ \Adm\bigl(k,t,\chi_b,\ba_b(\bD),\mathcal P_b(\bD)\bigr)=1
 \text{ for both }b\in\bits\bigr\}.
 \label{eq:direction-set}
\end{equation}

\begin{proposition}[Checkability and density]
\label{prop:encoded-good-density}
The predicate $\Adm$ is computable in uniform deterministic polynomial time given
$(k,t)$. Except on a key event of probability $2^{-\Omega(n)}$, the
failure of the regularity condition of \Cref{def:admissible}, pointwise
for every valid paired image,
\[
 \Pr_{\bD\leftarrow\bits^{N_{\rm ext}}}
  [\bD\notin
   \widehat G_{k,t,\boldzeta_0,\boldzeta_1}]
 \leq 2^{-\Omega(n)}.
\]
\end{proposition}

The direction set is evaluated only by the classical verifier, so it does
not affect the honest prover's depth. We prove
\Cref{prop:encoded-good-density}, together with an analysis of the
elementary attacks motivating \Cref{def:admissible}, in
\appref{app:encoded-good-directions}.

The counting argument also needs key-mode indistinguishability; we show
that it already follows from \Cref{ass:carry-ahcb}.

\begin{lemma}[The carry-predicate assumption implies key-mode indistinguishability]
\label{lem:encoded-key-hiding}
Under \Cref{ass:carry-ahcb}, no
uniform quantum algorithm of time
$2^{o(\sqrt{n})}$ distinguishes the public equation and image
keys in \Cref{sec:encoded-parameters} with non-negligible advantage.
\end{lemma}

\begin{proof}
Suppose that $\mathcal D$ distinguishes the two public-key modes, and let
$\Delta_{\rm real}$ be its equation-mode acceptance probability minus its
image-mode acceptance probability. Let
$\eps_{\rm stat}=\negl(\lambda)$ collect the two public-marginal
distances of the trapdoor generators and the finite-precision error-sampling
distance. Replacing these distributions by their idealized versions gives
\[
 K_{\Eq}=[\bA\mid\bA\bs+\be]
 \qquad\text{and}\qquad
 K_{\Img}=[\bA\mid\bu],
\]
where $\bA$ and $\bu$ are uniform and
$\bs\leftarrow\bits^{n}$ and
$\be\leftarrow
D_{\mathbb Z_{q},B_V}^{m}$; all suppressed
public metadata is unchanged. Write
$p_{\Eq}$ and $p_{\Img}$ for the probabilities that $\mathcal D$
outputs one in these two experiments. Then
\[
 \bigl|\Delta_{\rm real}-(p_{\Eq}-p_{\Img})\bigr|
 \leq\eps_{\rm stat}.
\]

We construct an adversary $\mathcal B$ for \Cref{ass:carry-ahcb}. Fix
$j_1\in[n/2]$ and
$j_2\in\{n/2+1,\ldots,n\}$, and let
$\mathcal T=\{j_1,j_2\}$. Given an equation key
$(\bA,\bu=\bA\bs+\be)$, sample
$\bg=(g_{j_1},g_{j_2}),\bg'=(g'_{j_1},g'_{j_2})\leftarrow\bits^2$ and
$\bR\leftarrow\mathbb Z_{q}^{m\times2}$.
Replace the columns $\bA_{\mathcal T}$ by $\bR$, obtaining $\bA'$, and set
\[
 \bu'=\bu-\bA_{\mathcal T}\bg+\bR\bg'.
\]
If $\bg=\bs_{\mathcal T}$, then $[\bA'\mid\bu']$ is distributed as
$K_{\Eq}$ when $\bA$ and $\be$ are ideal, and within $\eps_{\rm stat}$ of
it otherwise, with the two replacement secret coordinates given by
$\bg'$. If $\bg\neq\bs_{\mathcal T}$, some coordinate of
$\bs_{\mathcal T}-\bg$ is $1$ or $-1$; hence the discarded uniform columns in
$\bA_{\mathcal T}(\bs_{\mathcal T}-\bg)$ make $\bu'$ uniform and independent of $\bA'$.
In that case $[\bA'\mid\bu']$ is distributed as $K_{\Img}$ when $\bA$ is
ideal, and within $\eps_{\rm stat}$ of it otherwise.
Run $\mathcal D$ on this rerandomized key, let $Y$ be its output, set
\[
 C=g_{j_1}\oplus g_{j_2}\oplus Y,
\]
and output $(+1,\ \boldsymbol\delta_{j_1}\oplus\boldsymbol\delta_{j_2},\ \emptyset,\ C)$.

This output satisfies the linearity condition of \Cref{def:admissible} for every key:
the predicate list is empty and the linear function is nonzero at the
first-half coordinate $j_1$. Membership in $H^{\rm carry}_{k,t}$ or
$\overline H^{\rm carry}_{k,t}$ is therefore exactly the test of $C$
against
$\langle\boldsymbol\delta_{j_1}\oplus\boldsymbol\delta_{j_2},\bs\rangle=s_{j_1}\oplus s_{j_2}$.
By data processing, replacing the real public matrix
marginal and implemented error sampler by their ideal distributions changes
the resulting expectation by at most $O(\eps_{\rm stat})$.

Let $\mu_{\Eq}=1-2p_{\Eq}$ and
$\mu_{\Img}=1-2p_{\Img}$. Among the four choices of $\bg$, one equals
$\bs_{\mathcal T}$; of the other three, two have the opposite parity from $\bs_{\mathcal T}$
and one has the same parity. Therefore
\begin{align*}
 \mathbb E\bigl[(-1)^{C\oplus(s_{j_1}\oplus s_{j_2})}\bigr]
 &=\frac14(\mu_{\Eq}-\mu_{\Img}-\mu_{\Img}+\mu_{\Img})\\
 &=\frac{p_{\Img}-p_{\Eq}}2.
\end{align*}
Up to the statistical error above, the advantage of $\mathcal B$ is
therefore, in absolute value, half the key-distinguishing gap:
\[
 \Adv(\mathcal B)
 \geq \frac12|\Delta_{\rm real}|
       -O(\eps_{\rm stat}).
\]
The reduction invokes $\mathcal D$ once and otherwise has polynomial
overhead. If $\mathcal D$ runs in time
$2^{o(\sqrt{n})}$, then for all sufficiently large parameters
the
reduction runs in time at most $2^{c\sqrt{n}}$, where $c$
is the constant in \Cref{ass:carry-ahcb}, contradicting the assumption
whenever the original gap is non-negligible.
\end{proof}

\subsection{The protocol}
\label{sec:encoded-protocol}

Our single-round proof of quantumness is \Cref{prot:encoded} below.

\begin{protocolframe}
\begin{protocol}[Single-round proof of quantumness from the encoded LWE
family]
\label{prot:encoded}
\leavevmode
\smallskip

\noindent\textbf{Input:} a security parameter $\lambda$, in unary.

\begin{enumerate}
 \item The verifier generates the parameters of
  \Cref{sec:encoded-parameters}, samples a mode
  $\mathsf{mode}\leftarrow\{\Eq,\Img\}$ uniformly, and generates the corresponding
  equation or image key. For an equation key it retains
  $t=(t_\bA,\bs,\be,\boldsymbol\varpi)$; for an image key it retains $t=t_\bK$. It
  sends the public key
  $k=(q,n,m,
        Q,W,\bA,\bu)$
  to the prover.
 \item The prover prepares a uniform superposition over
  \[
   b\in\bits,\quad \bx\in X,\quad
   \boldeta\in E^{m},\quad \br,\quad \bh
  \]
  as follows:
  \begin{enumerate}
   \item the $b$-register and the low $L/2$ qubits of each $L$-qubit
    $x_j$-register are Hadamard states, because their side lengths are
    powers of two; the top $L/2$ qubits of each $x_j$-register hold the
    leading zeros of $J_{q}(x_j)$ and remain $\ket{0}$ until the final
    Hadamards;
   \item each function-noise register is prepared as an unsigned
    $z_i\in\{0,\ldots,W-1\}$ by Hadamards and translated in
    place to $\boldeta_i=z_i-W/2\pmod{q}$;
   \item every coordinate of $\br,\bh$ is prepared independently in the
    exact state $\ket{U_{q}}$ using
    \Cref{lem:prime-uniform-preparation}.
  \end{enumerate}
  It computes $\widehat F_k(b,\bx,\boldeta;\br,\bh)$ into an output
  register, applies the operation in \Cref{eq:eta-erasure} to clear
  $\boldeta$, and clears all arithmetic scratch. It then applies
  Hadamards to the registers containing $(b,\bx,\br,\bh)$ and measures
  all registers; write $C\in\bits$ for the outcome on the
  $b$-register and $\bD\in\bits^{N_{\rm ext}}$ for the outcome on the
  registers containing $(\bx,\br,\bh)$. It sends the verifier the encoded output
  $\widehat{\by}$, the direction $\bD$, the bit
  $C$, and the measured scratch registers.
 \item The verifier rejects malformed responses and nonzero scratch
  registers, and decodes $\widehat{\by}$ to $\by$.
  \begin{itemize}
   \item \textbf{if} $\mathsf{mode}=\Eq$: it runs
    $\mathsf{ExtInv}_t(\widehat\by)$ (\Cref{def:paired-images}) and
    rejects on failure; otherwise it obtains $(\boldzeta_0,\boldzeta_1)$
    and accepts if and only if
    $\bD\in
    \widehat G_{k,t,\boldzeta_0,\boldzeta_1}$ (\Cref{eq:direction-set}) and
    \[
      C=\bD\cdot(\boldzeta_0\oplus\boldzeta_1).
    \]
   \item \textbf{if} $\mathsf{mode}=\Img$: it runs $\TrapInv(t,\by)$ to obtain a
    candidate $(\bv,\boldeta)$, parses $\bv=(\bx,b)$, rejects unless
    $b\in\bits$, and runs
    $\mathsf{Rec}_k(\widehat\by,b,\bx)$. It accepts if and only if
    these steps produce a string $\boldzeta$ satisfying
    $\widehat R_k(\widehat\by;b,\boldzeta)=1$
    (\Cref{def:validity-relation}).
  \end{itemize}
\end{enumerate}
\end{protocol}
\end{protocolframe}

All of the honest prover's measurements are performed at the end. Measuring
the encoded output before the Hadamard registers is used only as an
equivalent description in the completeness proof; the projectors commute.

\subsection{Complexity and security}
\label{sec:encoded-security}

We first give the exact preparation of the uniform $\mathbb Z_q$ state
(\Cref{lem:prime-uniform-preparation}) and bound the depth and width of the
honest prover (\Cref{lem:encoded-circuit}), then compile it into standard
$\QAC$ (\Cref{lem:encoded-qac}), show that an accepting image transcript
yields an LK query (\Cref{lem:encoded-image-extraction}), and finally prove
\Cref{thm:encoded-formal} by the counting argument of
\Cref{lem:aggm-counting}.

\begin{lemma}[Exact uniform prime-field state]
\label{lem:prime-uniform-preparation}
For every public prime satisfying
\Cref{eq:encoded-prime-modulus}, the state
\[
 \ket{U_{q}}
 =\frac{1}{\sqrt{q}}
   \sum_{z=0}^{q-1}\ket{z}
\]
has an exact, bounded-arity preparation of depth
$O(\log L)$ and an exact constant-depth $\QAC$ preparation.
Polynomially many independent copies can be prepared in parallel with the
same depth.
\end{lemma}

\begin{proof}
Apply Hadamards to obtain the uniform state on
$\{0,\ldots,2^{L}-1\}$, and mark a basis string precisely when
$z<q$. The marked weight is
$q/2^{L}\in[1/2,1)$. Exact amplitude amplification
with phase matching therefore maps this state to its normalized marked
projection using a constant number of generalized Grover steps
\cite{Hoyer2000}. Symmetry within the marked subspace makes that projection
exactly $\ket{U_{q}}$.

Each marked-subspace phase is implemented by computing the public comparison
$[z<q]$, applying the prescribed one-qubit phase to its flag,
and uncomputing. A phase about the initial Hadamard state is obtained by
conjugating a phase on $\ket{0^{L}}$ by Hadamards. The
phase angles are efficiently computable functions of the known rational
$q/2^{L}$, and arbitrary one-qubit rotations are
allowed in our circuit model. Here exactness is relative to that
declared gate model: the phase angles have succinct, effectively computable
descriptions determined by $q/2^L$.  Over a fixed finite gate set their
synthesis would generally be approximate, giving negligible preparation
error rather than literal zero error. Comparisons and the two multi-controlled
phases have clean bounded-arity depth $O(\log L)$. For the
$\QAC$ implementation, comparison is in uniform
$\TC$~\cite{HesseAllenderBarrington2002}; its standard
$\QNCf$ realization uses fan-out only within a
$\poly(L)$-size block, and the exact polylogarithmic-fan-out
compiler of Grier--Morris--Wu converts it to constant-depth
$\QAC$~\cite[Corollary~10]{GrierMorrisWu2026}. A multi-controlled phase is obtained by
computing its conjunction into one flag, applying the one-qubit phase, and
uncomputing. All work registers return to zero, and disjoint copies of the
construction run in parallel.
\end{proof}

\begin{lemma}[Encoded circuit complexity]
\label{lem:encoded-circuit}
The honest prover in \Cref{prot:encoded} has an exact bounded-arity
implementation of depth
\[
 O(\log L)=O(\log\log\lambda).
\]
The source and output registers contain
\[
 \Theta(mnL)
 =\Theta(\lambda^4L^2)
 =\Theta(\lambda^4\log^4\lambda)
\]
qubits. A direct parallel reversible implementation has width
\[
 \widetilde O(mnL^2)
 =\widetilde O(\lambda^4L^3)
 =\widetilde O(\lambda^4\log^6\lambda).
\]
\end{lemma}

\begin{proof}
Every symbol in \Cref{eq:re-w,eq:re-z,eq:encoded-noise} uses a constant
number of additions, subtractions, and multiplications by public
$L$-bit constants. Standard bounded-fan-in addition and
public-constant multiplication have depth $O(\log L)$ and use
$\widetilde O(L^2)$ reversible workspace per encoded
symbol. Distinct encoded symbols use distinct mask inputs, except for the
constant number of occurrences displayed in the formulas, so no dense
fan-out tree is needed. All symbols are evaluated in parallel and their
scratch is uncomputed before the final Hadamards. Counting the $\bw,\bz,\br,\bh$
symbols gives $\Theta(mn)$ symbols and proves the
claims. The operation in \Cref{eq:eta-erasure} is one further constant
composition of the same local arithmetic operations and has the same depth
bound. All prime-field mask registers are prepared in parallel by
\Cref{lem:prime-uniform-preparation}, adding $O(\log L)$ depth
and only $\widetilde O(L)$ clean workspace per mask.
\end{proof}

We next compile the same prover into standard $\QAC$ by replacing every
fan-out gate of polylogarithmic arity with the exact constant-depth circuit
of Grier, Morris, and Wu.

\begin{lemma}[Exact compilation into standard $\QAC$]
\label{lem:encoded-qac}
The honest prover in \Cref{prot:encoded} has an exact,
polynomial-size, constant-depth $\QAC$ implementation. Every circuit layer
is a tensor product of gates on pairwise disjoint qubit sets.
\end{lemma}

\begin{proof}
First consider one output-symbol formula in
\Cref{eq:re-w,eq:re-z,eq:encoded-noise}. It is a constant composition of
addition, subtraction, and multiplication by a public constant on
$L$-bit strings. These operations have uniform
polynomial-size $\TC$ circuits~\cite{HesseAllenderBarrington2002}.
Their clean reversible implementations are therefore constant-depth
$\QNCf$ circuits: exact unbounded fan-out realizes the required threshold
and arithmetic operations~\cite{HoyerSpalek2005,TakahashiTani2016}.
The local circuit has size $\poly(L)$, so every fan-out gate in
it has arity at most
\[
  \poly(L)=\log^{O(1)}\lambda.
\]

Grier, Morris, and Wu prove that, for every fixed $d$, fan-out on
$\log^d\lambda$ targets has an exact polynomial-size, constant-depth
standard-$\QAC$ implementation~\cite[Corollary~10]{GrierMorrisWu2026}.
Choose $d$ larger than the fixed exponent in the preceding display and
replace every local fan-out gate by that circuit. A layer of the original
$\QNCf$ circuit contains disjoint fan-out gates; their replacement circuits
use private ancillas and hence run in parallel as a constant number of
$\QAC$ layers.
Exact uncomputation returns those ancillas to zero.

It remains to run all output-symbol blocks at once. The occurrence bound in
\Cref{lem:re-properties} lets us make the constant number of copies of each
source symbol required by distinct blocks using a constant number of
bounded-arity CNOT layers. The blocks then have disjoint supports and can
run in parallel. Reversing the copy layers after compute--copy--uncompute
clears their ancillas. The row-local operations in
\Cref{eq:eta-erasure} use the same compilation and execute in parallel.
There are polynomially many blocks, each of
polynomial size in $\lambda$, so the total size is polynomial and the
depth is constant. State preparation, the final Hadamards, and measurement
add only constant depth: the translation used to prepare
$E$ is one more public-constant arithmetic block of the same
kind, and the uniform prime-field mask preparation is constant-depth
$\QAC$ by \Cref{lem:prime-uniform-preparation}.
\end{proof}

The next lemma turns an accepting image-mode transcript into an
LK-$1/4$ query and shows that the extracted lattice point determines the
unique valid extended preimage; it supplies the extraction procedure of
\Cref{lem:aggm-counting}.

\begin{lemma}[Image-mode extraction]
\label{lem:encoded-image-extraction}
Except with negligible probability over an image key, an accepting image
transcript gives an LK-$1/4$ query whose extracted lattice point determines
the unique valid extended preimage in polynomial time using only public
data.
\end{lemma}

\begin{proof}
Except with probability at most
$q^{n+1-m}$ for a uniform matrix (plus the trapdoor
generator's statistical error), $\bK$ has full column rank over
$\mathbb F_{q}$. Gaussian elimination over this field then
computes a public left inverse
\begin{equation}
 \mathbf B\bK=I_{n+1}\pmod{q}.
 \label{eq:encoded-left-inverse}
\end{equation}

If the image test accepts, it has found
$\bv=(\bx,b)$ and $\boldeta\in E^{m}$ with
$\by=\bK\bv+\boldeta$. Since every centered coordinate of
$\boldeta$ has magnitude at most $W/2$,
\[
 \|\centerq(\boldeta)\|_2
 \leq\frac{\sqrt{m}\,W}{2}
 \leq\frac{r_G}{4\sqrt{m}}
 <r_G.
\]
Moreover, \Cref{eq:encoded-lattice-separation} gives
\[
 \lambda_1^{\infty}(\mathcal L(\bK))
 \geq\frac{r_G}{\sqrt{m}},
 \qquad
 \|\centerq(\boldeta)\|_\infty
 \leq\frac{r_G}{4m}
 \leq\frac14\lambda_1^{\infty}(\mathcal L(\bK)).
\]
We apply \Cref{ass:encoded-lk} to the classical algorithm that runs the
prover and outputs the decoded point $\by$; decoding is public and
polynomial time, so this algorithm obeys the same time bound up to a
polynomial factor, and its online extractor is the one used here and in
the soundness proof. Thus the online LK extractor returns the nearby lattice point
$\bp\in\mathcal L(\bK)$. It need not return its coefficients: reducing
modulo $q$ and using \Cref{eq:encoded-left-inverse} recovers
\[
 \bv=\mathbf B\bp\pmod{q}.
\]
The extraction procedure then checks that $b\in\bits$ and
$\bx\in X_b$, computes
$\boldeta=\by-\bp\pmod{q}$, and runs
$\mathsf{Rec}_k$ on the encoded image. The verifier's relation
check and uniqueness in \Cref{lem:noisy-re-privacy} then give exactly one
valid $\boldzeta$.
\end{proof}

We can now state and prove the main result.

\begin{theorem}[Main result]
\label{thm:encoded-formal}
Assume
\Cref{ass:encoded-lk,ass:carry-ahcb} for the
parameter ensemble in
\Cref{eq:encoded-dimensions,eq:encoded-prime-modulus}. Then
\Cref{prot:encoded} is a single-round proof of quantumness
(\Cref{def:single-round-poq}) with the following properties.
\begin{enumerate}
 \item \textbf{Completeness.}  The honest quantum prover is accepted with
  probability $1-\negl(\lambda)$. It admits both an exact
  polynomial-size, constant-depth $\QAC$ implementation and an exact
  bounded-arity implementation of log-logarithmic depth $\QNCloglog$. Both implementations use only final
  measurements.
 \item \textbf{Soundness.}  Every fixed uniform probabilistic
  polynomial-time classical prover is accepted with probability at most
  $3/4+\negl(\lambda)$. Moreover, every fixed
  uniform classical prover of time $2^{o(\lambda)}$ is accepted
  with probability at most $3/4+\negl(\lambda)$ for all sufficiently
  large $\lambda$.
\end{enumerate}
In particular, the honest and classical acceptance
probabilities are separated by a constant (for example, by at least
$1/8$) for all sufficiently large $\lambda$.
\end{theorem}

\begin{proof}
\textbf{Completeness.}  The two implementations are given by
\Cref{lem:encoded-circuit,lem:encoded-qac}. For the acceptance
probability, first restrict to
$\bx_0\in X_0$,
$\bx_1=\bx_0-\bs\in X_1$, and pairs
$\boldeta_0=\boldeta_1+\be$ lying in
$E^{m}$ on both branches. The excluded boundary
and shifted-noise mass is at most
$2n/Q+
mB_V/W$, as in
\Cref{eq:encoded-scales}. The two original preimages then have the same
decoded image. By \Cref{lem:noisy-re-privacy}, they induce identical
distributions on encoded outputs, and randomness reconstruction gives one
coin string per original preimage and encoded output. Moreover,
outside the trapdoor-failure event of $\bA$, for each $b\in\bits$ there is
at most one $\bx$ with $\by-b\bu-\bA\bx\in E^{m}$, by the uniqueness
statement of \Cref{thm:trapgen} at dimension $n$: every
$\boldeta\in E^{m}$ has
$\|\centerq(\boldeta)\|_2\leq\sqrt{m}\,W/2<r_G<r_n$. Hence $(0,\bx_0)$
and $(1,\bx_1)$ are the only sources with noise in $E^{m}$ that map to
this output. Conditioned on such
an output, the surviving registers are therefore the equal superposition
\[
 \frac{1}{\sqrt2}
 \bigl(\ket{0,\boldzeta_0}+
       \ket{1,\boldzeta_1}\bigr).
\]
The final Hadamards act on all $L$ qubits of every $x_j$-register,
so they give a uniform $\bD\in\bits^{N_{\rm ext}}$ and a branch
outcome satisfying
$C=\bD\cdot(\boldzeta_0\oplus\boldzeta_1)$ exactly.
By \Cref{prop:encoded-good-density}, the verifier rejects at most a
$2^{-\Omega(n)}$ fraction of these directions, outside the
regularity-failure key event, which is absorbed below. In image mode,
restrict to $\bx\in X_b$, which excludes mass at most $n/Q$. Every
retained honest basis state then passes
\Cref{eq:encoded-valid-relation} by construction, and trapdoor inversion
recovers it within the radius proved in
\Cref{lem:encoded-image-extraction}. The excluded branch-boundary mass,
shifted-noise mass, and prime-sampling abort together contribute
$\negl(\lambda)$ by \Cref{eq:encoded-scales} and the $2^{-\Omega(\lambda)}$ prime-sampling abort bound of \Cref{sec:encoded-parameters}; matrix-regularity, trapdoor,
key-sampling, and image-mode full-rank failures are smaller. This proves
completeness $1-\negl(\lambda)$.

\textbf{Soundness.}  Fix a classical prover, and let $\mathcal I$ be
image-mode acceptance and $\mathcal E$ be
equation-mode acceptance. In the image experiment, run the online
extractor from \Cref{lem:encoded-image-extraction}, and let $\mathcal Q$
be the event that it returns a tuple
$(\widehat\by,b,\boldzeta_b,\bD,C)$ satisfying
$\widehat R_k(\widehat\by;b,\boldzeta_b)=1$. The lemma and
\Cref{ass:encoded-lk} give
\[
 \Pr_{\Img}[\mathcal I]\leq\Pr_{\Img}[\mathcal Q]+\eps_{\rm ext}
\]
for $\eps_{\rm ext}=2^{-\Omega(\sqrt{n})}$. The extraction
procedure, including
the public-left-inverse computation, uses only the public key and the online
extractor. The equation/image key indistinguishability established in
\Cref{lem:encoded-key-hiding} therefore transfers the event $\mathcal Q$:
\[
 \bigl|\Pr_{\Img}[\mathcal Q]-\Pr_{\Eq}[\mathcal Q]\bigr|
 \leq\eps_{\rm key},
\]
where $\eps_{\rm key}=\negl(\lambda)$.

In the equation experiment, let $\mathcal U_{\rm uniq}$ be the good-key event on
which trapdoor inversion is unique throughout the noise radius.
Its complement has probability
$\delta_{\rm uniq}=2^{-\Omega(n)}$. On
$\mathcal Q\cap\mathcal E\cap \mathcal U_{\rm uniq}$, the extracted $\boldzeta_b$ equals the
branch-$b$ string returned by $\mathsf{ExtInv}_t$:
\Cref{lem:noisy-re-privacy} rules out a second valid mask string for that
branch.
Consider the algorithm that computes the induced data of the extracted
transcript and outputs
\[
 \bigl(\chi_b,\ \ba_b(\bD),\ \mathcal P_b(\bD),\
 C\oplus\beta_b(\bD)\bigr).
\]
Whenever $\mathcal Q\cap\mathcal E\cap \mathcal U_{\rm uniq}$ occurs, this output lies in
$H^{\rm carry}_{k,t}$: the verifier's direction test includes the
branch-$b$ admissibility check, and the parity is correct by
\Cref{lem:decomposition}. The same output lies in
$\overline H^{\rm carry}_{k,t}$ when the direction test passes but the parity
is flipped; on extraction failure the algorithm outputs a fixed tuple with
$\ba=\mathbf 0$ and an empty predicate list, which is inadmissible. The
two sets are disjoint, so \Cref{ass:carry-ahcb} implies
\[
 \Pr_{\Eq}[\mathcal Q\cap\mathcal E]
 \leq \Pr[H^{\rm carry}_{k,t}]+\delta_{\rm uniq}
 \leq\frac12+\frac12\eps_{\rm ahcb}+\delta_{\rm uniq},
\]
where $\eps_{\rm ahcb}=2^{-\Omega(\sqrt{n})}$.
Applying \Cref{lem:aggm-counting} with
$\eps_{\rm hid}=\eps_{\rm key}$, extraction failure $\eps_{\rm ext}$,
and hardcore-bit advantage
$\eps_{\rm hc}=\eps_{\rm ahcb}/2+\delta_{\rm uniq}$ gives
\[
 s(\lambda)
 \leq\frac34+
 O(\eps_{\rm ext}+\eps_{\rm key}+\eps_{\rm ahcb}+\delta_{\rm uniq})
 =\frac34+\negl(\lambda).
\]
The last equality uses key-mode indistinguishability from
\Cref{lem:encoded-key-hiding}; the rank, trapdoor, key-sampling, LK, and
hardcore-bit errors are negligible as well.

For the time-bounded claim, fix a uniform prover of time
$2^{o(\lambda)}$. Since $n=\Theta(\lambda^2)$, we have
$\sqrt{n}=\Theta(\lambda)$, so the prover runs in time
$2^{o(\sqrt{n})}$, and the polynomial overhead in the
key-hiding, extractor, and hardcore-bit reductions stays below the time
bounds of \Cref{ass:encoded-lk} and \Cref{ass:carry-ahcb} for all
sufficiently large $\lambda$. The preceding bound applies unchanged.
\end{proof}

%% file: sections/appendix_encoded_good_directions.tex
\section{\texorpdfstring{The Admissibility Predicate}{The Admissibility Predicate}}
\label{app:encoded-good-directions}

This appendix proves \Cref{lem:decomposition} and
\Cref{prop:encoded-good-density}, and explains why the balance tests in
\Cref{def:admissible} rule out several elementary classical strategies. Passing the tests is not known to imply
a hardcore bit, and the arguments below do not derive
\Cref{ass:carry-ahcb} from LWE.
Rather, they show that the assumption is being made for an explicit
predicate that accepts almost all of the honest challenges while excluding
the carry-cancellation and padding strategies described below. The BCMVV result uses a similar admissible set.

\subsection{Preimage differences and the decomposition}

We first make explicit how the two extended preimages of a valid paired
image differ. The decomposition of \Cref{lem:decomposition} then follows by
splitting the parity over the register blocks, and the realizability lemma at
the end of the subsection shows that every pair $(\ba,\mathcal P)$ arises
from a transcript that the adversary can generate itself.

\begin{lemma}[Preimage differences]
\label{lem:preimage-differences}
Fix a valid paired image, a branch $b$, and write $\chi_b=(-1)^b$.
Then, modulo $q$,
\[
 \bx_{1-b}=\bx_b-\chi_b\bs,
 \qquad
 r^{(1-b)}_{j,i}=r^{(b)}_{j,i}+\chi_bs_j
 \quad(j\leq n),
 \qquad
 r^{(1-b)}_{n+1,i}
  =r^{(b)}_{n+1,i}-\chi_b,
\]
and
\[
 h^{(1-b)}_{i,j}
 =h^{(b)}_{i,j}+\chi_b\left(
     e_i+\sum_{j'\leq j}\bA_{i,j'}s_{j'}\right)
 =h^{(b)}_{i,j}+\chi_b\,\sigma_{i,j}(\bs).
\]
\end{lemma}

\begin{proof}
The relation between the decoded preimages is part of
\Cref{def:paired-images}. Equality of the $\bw$-part of the two
encodings then gives
\[
 r^{(1-b)}_{j,i}=r^{(b)}_{j,i}+\chi_bs_j
 \quad(j\leq n),
 \qquad
 r^{(1-b)}_{n+1,i}
  =r^{(b)}_{n+1,i}-\chi_b,
\]
because passing to the other branch changes $\bv_j$ by $-\chi_bs_j$
for $j\leq n$ and the branch coordinate
$\bv_{n+1}$ by $\chi_b$, and
$\br_{j,1}=\bw_{j,0}-\bv_j$ negates these changes and propagates them
down each column. Equality
of the $\widehat\bz$-part and
$\boldeta_0-\boldeta_1=\be$ similarly gives the $h$-relation, by
telescoping \Cref{eq:re-z} along row $i$. At the last stored path mask,
$j=n$, the parenthesized value is
$(\bA\bs+\be)_i=u_i$; the branch column then contributes
$\bM_{i,n+1}\bigl(r^{(1-b)}_{n+1,i}-r^{(b)}_{n+1,i}\bigr)=-\chi_bu_i$, so the
telescoped difference returns to zero at $j=n+1$, as required by the boundary
condition $\bh_{i,n+1}=0$ on both branches.
\end{proof}

\begin{proof}[of \Cref{lem:decomposition}]
Fix a branch $b$ and abbreviate $\chi=\chi_b$ and
$J=J_{q}$. By \Cref{lem:preimage-differences}, the parity
splits over the register blocks as
\begin{align*}
 \bD\cdot(\boldzeta_0\oplus\boldzeta_1)
 =\ &\bigoplus_{j}\bD^x_j\cdot
   \bigl(J(x_{b,j})\oplus J(x_{b,j}-\chi s_j)\bigr)\\
 {}\oplus\ &\bigoplus_{j\leq n}\bigoplus_i
   \bD^r_{j,i}\cdot
   \bigl(J(r^{(b)}_{j,i})\oplus J(r^{(b)}_{j,i}+\chi s_j)\bigr)
 \ \oplus\ \bigoplus_i \bD^r_{n+1,i}\cdot
   \bigl(J(r^{(b)}_{n+1,i})\oplus
   J(r^{(b)}_{n+1,i}-\chi)\bigr)\\
 {}\oplus\ &\bigoplus_{i,j}\bD^h_{i,j}\cdot
   \bigl(J(h^{(b)}_{i,j})\oplus
   J(h^{(b)}_{i,j}+\chi\,\sigma_{i,j}(\bs))\bigr).
\end{align*}
Since $s_j\in\bits$, the $j$-th $x$-term equals
$s_j\cdot\bD^x_j\cdot(J(x_{b,j})\oplus J(x_{b,j}-\chi))$, and the
$(j,i)$-th $r$-term equals
$s_j\cdot\bD^r_{j,i}\cdot(J(r^{(b)}_{j,i})\oplus
J(r^{(b)}_{j,i}+\chi))$; summing over $j$ and $i$ gives
$\langle\ba_b(\bD),\bs\rangle$ by \Cref{eq:induced-coefficient}. The
branch column and the terminal $h$-column, where
$\sigma_{i,n}(\bs)=u_i$, do not depend on $\bs$ and sum to
$\beta_b(\bD)$ by \Cref{eq:induced-offset}. The remaining $h$-terms
are exactly $\Phi^{\chi}_{\mathcal P_b(\bD)}(\bs)$ by
\Cref{eq:carry-parity}.
\end{proof}

\begin{lemma}[Realizability]
\label{lem:realizability}
For every sign $\chi$, every $\ba\in\bits^{n}$, and every predicate
list $\mathcal P$, there is a
transcript with a direction, computable in polynomial time, whose induced
data at the branch $b$ with $(-1)^b=\chi$ is exactly
$(\ba,0,\mathcal P)$, and whose encoded output is a valid paired image
for every $(\bs,\be)$ in the support of the key generation, except on
the trapdoor-failure event.
\end{lemma}

\begin{proof}
Take $b$ with $(-1)^b=\chi$, $\bx=\mathbf 1$ if $b=0$ and
$\bx=(Q-2)\mathbf 1$ if $b=1$, $\boldeta=\mathbf 0$,
$\br=\mathbf 0$, $\bh_{i,j}=\xi$ for each
$(i,j,\xi,\bd)\in\mathcal P$ and $\bh_{i,j}=0$ otherwise, and encode
honestly. Outside the trapdoor-failure event, the checks of
\Cref{def:paired-images} succeed for every $(\bs,\be)$ in the support: the
two relations hold by construction, $\bx$ and $\bx-\chi\bs$ lie in the two
branch domains, and $\boldeta=\mathbf 0$ and $\boldeta-\chi\be$ lie in
$E^{m}$ because $B_V<W/2$. Choose $\bD^h_{i,j}=\bd$ on the
listed blocks and zero elsewhere, in particular on the terminal column;
$\bD^r=\mathbf 0$; and, for each $j$, a block $\bD^x_j$ with
$\bD^x_j\cdot(J_{q}(x_j)\oplus
J_{q}(x_j-\chi))=a_j$, which is solvable for both bit values
because the word in parentheses is nonzero. With this support $\beta_b(\bD)=0$, and
the induced data is $(\ba,0,\mathcal P)$.
\end{proof}

\subsection{Polynomial-time checkability and pointwise density}

The bits $S_j$, $\delta_{i,j}(\rho)$, and $\Gamma^{V}_{i,k}$ of
\Cref{def:admissible} are computed from the secret $\bs$, the private
shifts $\boldsymbol\varpi$, the public matrix and key vector, and the pair
under test. There are $O\bigl(mn(L+N_{\rm sh})\bigr)$ of them, each
computed by polynomially many modular operations, canonical encodings, and
binary inner products; every interval vector is a difference of two prefix
XORs of the $\delta_{i,j}(\rho)$, the regularity condition depends only on
$\bA$, and the rank conditions are $O(L+N_{\rm sh})$ Gaussian
eliminations on binary matrices with at most $m$ rows and $n$ columns.
Hence even the direct implementation of $\Adm$ is deterministic polynomial
time.

We also verify the regular-key event used below.  First replace the
public marginal of $\TrapGen$ by a uniform matrix; this costs
$2^{-\Omega(n)}$ in statistical distance.  For a fixed row, window
$V\in\{[n-1],\mathcal J\}$, and secret half $h$, the set of candidate
indices $\{k\text{ in half }h:k-1\in V\}$ has size at least
$n/2-3L$ and hence at least $3n/8$ for all sufficiently large parameters.
If $|E^V_{i,h}|<\lfloor n/4\rfloor$, at least $n/8-O(1)$ of the
corresponding independent uniform entries of $\bA$ are zero.  Therefore,
for a fixed $(i,V,h)$, this event has probability at most
\[
 2^{n/2}q^{-n/8+O(1)}=2^{-\Omega(nL)}.
\]
Likewise, for a fixed column, the probability that more than $m/8$ of its
entries are zero is at most
$2^m q^{-m/8}=2^{-\Omega(mL)}$.  A union bound over all rows, windows,
halves, and columns shows that a uniform matrix is irregular with
probability $2^{-\Omega(nL)}$.  Restoring the trapdoor-generator marginal,
the regularity-failure probability is thus $2^{-\Omega(n)}$, as claimed in
\Cref{prop:encoded-good-density}.

It remains to prove density. If $U_\nu$ is uniform in $\bits^\nu$, then a
Hoeffding bound (deviation $\nu/6$ from the mean $\nu/2$) gives
\begin{equation}
 \Pr[U_\nu\text{ is not balanced}]
 \leq 2\exp(-\nu/18).
 \label{eq:balanced-tail}
\end{equation}
Fix a valid paired image and a branch $b$, and consider the induced data
of a uniform direction $\bD$. We show that every vector tested by the
sensitivity condition at $(\ba_b(\bD),\mathcal P_b(\bD))$ is individually
uniform. Joint independence between these vectors is neither claimed nor
needed.

Fix $j$. Conditional on all direction blocks other than $\bD^x_j$,
the bit $\ba_b(\bD)_j$ contains the term
\[
 \bD^x_j\cdot
 \bigl(J_{q}(x_{b,j})
       \oplus J_{q}(x_{b,j}-\chi_b)\bigr).
\]
The word in parentheses is nonzero because $x_{b,j}-\chi_b\neq x_{b,j}$
in $\mathbb Z_{q}$ and canonical encoding is injective. Distinct coordinates use disjoint
$\bD^x_j$-blocks. Consequently
$\ba_b(\bD)$ is exactly uniform.

The directional derivatives $S_j$ of \Cref{eq:adm-secret-profile} satisfy
$S_j=\ba_b(\bD)_j\oplus
\bigl(\Phi^{\chi_b}_{\mathcal P_b(\bD)}(\bs)\oplus
\Phi^{\chi_b}_{\mathcal P_b(\bD)}(\bs^{(j)})\bigr)$, and the
$\Phi$-difference does not involve $\bD^x$. The same block
$\bD^x_j$ therefore makes the complete vector
$(S_j)_{j\in[n]}$ uniform after conditioning on all
$r$- and $h$-blocks.

Next fix a row $i$, a nonzero shift $\rho\in\mathcal R_i=\mathcal
R\cup\{\varpi_{i,1},\ldots,\varpi_{i,N_{\rm sh}}\}$, and a nonempty
interval $I\subseteq[n-1]$ with largest element $j_I$. Writing
$w=h^{(b)}_{i,j_I}+\chi_b\sigma_{i,j_I}(\bs)$, the bit
$\Delta^{I}_{\rho,i}$ contains the term
\[
 \bD^h_{i,j_I}\cdot\bigl(
 J_{q}(w)
 \oplus J_{q}(w+\chi_b\rho)\bigr),
\]
including the case $\bD^h_{i,j_I}=\mathbf 0$, in which the pair
$(i,j_I)$ is simply not listed in $\mathcal P_b(\bD)$ and contributes
zero. The coefficient word is nonzero because $\rho\not\equiv0\pmod q$.
Conditioning on every block except $\bD^h_{i,j_I}$ therefore makes
$\Delta^{I}_{\rho,i}$ uniform, and different rows use disjoint
$\bD^h$-blocks, so each tested interval vector is exactly uniform.

For the gradients, fix $i$, a window $V$, and a half $h$, and order
$E^{V}_{i,h}$ increasingly. Expanding $\sigma_{i,j}$,
\[
 \Gamma^{V}_{i,k}
 =\bigoplus_{\substack{j\in V\\ j<k}}\bD^h_{i,j}\cdot
  \bigl(J_{q}(w_{i,j})\oplus
        J_{q}(w_{i,j}+\chi_b(2s_k-1)\bA_{i,k})\bigr),
 \qquad w_{i,j}=h^{(b)}_{i,j}+\chi_b\sigma_{i,j}(\bs),
\]
and the coefficient at $j=k-1$ is nonzero whenever $\bA_{i,k}\neq0$.
The block $\bD^h_{i,k-1}$ occurs in no $\Gamma^{V}_{i,k'}$ with
$k'<k$, so the map from the direction to the string
$(\Gamma^{V}_{i,k})_{k\in E^{V}_{i,h}}$ is block triangular with nonzero
pivots, and that string is exactly uniform.

For the rank conditions, each entry $\delta_{i,j}(\rho)$ of a derivative
matrix is the inner product of the fresh uniform block $\bD^h_{i,j}$ with
a nonzero word, so the entries are independent fair bits and a fixed row
half fails to have rank $n-1$ with probability at most
$2^{\,n-1-m_0}$. For a gradient matrix, fix a nonzero combination
$\mathbf c$ of the columns permitted for the half $h$. On a regular key, at
most $m/8$ rows have $\bA_{i,k}=0$ at every $k$ in the support of
$\mathbf c$; every other row $i$ contains, at the largest supported
$k\in E^{V}_{i,h}$, the fresh pivot $\bD^h_{i,k-1}$, so the bit
$\bigoplus_k c_k\Gamma^{V}_{i,k}$ is uniform. Hence the combination
vanishes on a row half with probability at most $2^{\,m/8-m_0}$, and a
union bound over the at most $2^{n/2}$ combinations gives
$2^{\,n/2+m/8-m_0}$.

Let $\nu=n/2$, put $m_0=\lfloor m/2\rfloor$ and
$m_1=\lceil m/2\rceil$, and let $\mathcal I$ denote the family of
nonempty intervals, $|\mathcal I|=\binom{n-1}{2}+n-1\leq n^2$. Applying
\Cref{eq:balanced-tail}, at each of the two branches, to the two halves
of $\ba_b(\bD)$ and of $(S_j)_{j\in[n]}$, to the two halves of every
tested interval vector, and to every eligible gradient string, whose
length is at least $\lfloor n/4\rfloor$ on a regular key, and adding the
rank failures, gives
\begin{align}
 \Pr[\bD\notin\widehat G_{k,t,\boldzeta_0,\boldzeta_1}]
 &\leq 16\exp(-\nu/18)
 +4|\mathcal I|\bigl(2\lfloor L/4\rfloor+2+N_{\rm sh}\bigr)
 \left(\exp(-m_0/18)+\exp(-m_1/18)\right)
 \nonumber\\
 &\quad{}+16m\exp\bigl(-\lfloor n/4\rfloor/18\bigr)
 +8\bigl(2\lfloor L/4\rfloor+2+N_{\rm sh}\bigr)2^{\,n-1-m_0}
 +16\cdot2^{\,n/2+m/8-m_0}\nonumber\\
 &=2^{-\Omega(n)}.
 \label{eq:good-density-proof}
\end{align}
The first term covers the two halves of $\ba$ and of $S$ at both
branches with the two tails of \Cref{eq:balanced-tail}; the second covers
the interval tests over all intervals, tested shifts, halves, and
branches; the third covers the $m\times2\times2$ gradient strings at both
branches; and the last two cover the rank tests. The total is
$2^{-\Omega(n)}$ because $m=\Theta(nL)$ and
$N_{\rm sh}=\lceil\sqrt n\rceil$. A direct union bound
is valid despite correlations among the tests. On a regular key, failing
the sensitivity condition at some branch is the only way to leave
$\widehat G_{k,t,\boldzeta_0,\boldzeta_1}$, so this proves
\Cref{prop:encoded-good-density} pointwise for every
valid paired image. The honest prover still samples $\bD$ uniformly by
its final Hadamards; it does not conditionally sample from the direction
set, so the predicate adds no quantum depth.

\subsection{Elementary attacks excluded by the balance tests}

Consider the adversary's predicate
\[
 \Psi(\bs')
 =\langle\ba,\bs'\rangle\oplus\Phi^{\chi}_{\mathcal P}(\bs').
\]
Its value at $\bs$ is the target parity, and $S_j$ is by definition the
directional derivative of $\Psi$ at $\bs$ in the direction
$\boldsymbol\delta_j$:
\[
 S_j
 =\Psi(\bs)
  \oplus \Psi(\bs\oplus\boldsymbol\delta_j).
\]
A carry cancellation is a choice of $(\ba,\mathcal P)$ in which the carry
term $\Phi^{\chi}_{\mathcal P}(\bs')$ agrees with $\langle\ba,\bs'\rangle$ up
to a constant, so that $\Psi(\bs)$ is known without knowledge of $\bs$. Any
cancellation that holds at $\bs$ and at all of its Hamming neighbors, so that
$\Psi$ is locally constant at $\bs$, has all directional derivatives zero and
is inadmissible. More generally, a pair whose
predicate has a nonzero derivative in fewer than one third of the
coordinates of either secret half
fails the corresponding balance test. This includes any carry
cancellation involving only a bounded number of secret coordinates, as
well as XOR padding by other locally constant functions.

Once carry predicates are used, the balance test on $\ba$ separately
prevents an attack from setting the
linear coefficient to zero or concentrating it on a few selected
coordinates. The terminal column cannot repair a failed test: its partial
sum is the public $u_i$, so it is excluded from predicate lists outright
and contributes only a known bit to the protocol's parities.

The interval tests address a second form of padding. A predicate list whose
predicates in an interval $I$ meet fewer than one third of the rows in
either half gives vectors $(\Delta^{I}_{\rho,i})_{i\in[m]}$ of weight below the
balanced range at every tested shift; taking $I=[n-1]$, a list
supported on one or a few rows is inadmissible, and no column of
$[n-1]$ escapes inspection. An empty list makes every interval
vector zero and every derivative matrix vanish; the sensitivity condition
fails, and the linearity condition forces
$\mathcal P=\emptyset$.

The interval family is not optional. Consider the variant of
\Cref{def:admissible} that tests only the single window $\mathcal J$ at
the public shifts, which is the definition of an earlier version of this
paper, and the boundary pair $j_0=n-3L\in\mathcal J$,
$j_1=j_0+1\notin\mathcal J$. Placing, in each row independently, a random
equal-mask, equal-base pair of predicates at $j_0,j_1$ cancels the pair's
contribution to the target on the slice $s_{j_1}=0$, of probability
$1/2$, because $\sigma_{i,j_0}(\bs)=\sigma_{i,j_1}(\bs)$ there, while
only the member at $j_0$ enters the tested vector, which is therefore
uniform across rows and passes. Combined with adjacent
least-significant-bit predicates at base points $\xi=-u_i$, whose
secret-dependent argument $\xi+\sigma_{i,j}(\bs)=-\sum_{k>j}
\bA_{i,k}s_k$ is free of the error $e_i$ and, for an odd centered
coefficient $\bA_{i,j}$, reduces to $s_j$ up to a modular-wrap indicator
of small probability, this yields a polynomial-time predictor with
constant bias against that variant. Under the interval family, the
interval $I=[j_0,j_1]$ sees both members, which cancel identically after
every shift, so all padding rows contribute zero and the vector is
unbalanced; and the row-gradient tests reject the terminal variant of the
same attack outright, since a predicate at $j=n-1$ depends only on
$s_n$ and its full-window gradient is supported on a single coordinate.

The private shifts broaden the audit beyond the fixed public
shifts, but they do not by themselves give a detection theorem.  For a
fixed evaluated word $z$, put
\[
 V_z=\operatorname{span}_{\mathbb F_2}
 \{J_q(z)\oplus J_q(z+\chi\rho):\rho\in\mathcal R\}.
\]
The bound $\dim V_z\leq|\mathcal R|<L$ shows only that
$V_z^\perp$ is nontrivial.  A public-audit-preserving perturbation
$\delta\bd\in V_z^\perp$ that also flips the target exists precisely when
$J_q(\xi)\oplus J_q(z)\notin V_z$; the dimension bound alone does not imply
this additional condition.  For any fixed nonzero $\delta\bd$, a uniform
nonzero private shift $\varpi$ changes the corresponding derivative bit
with probability at least $1/4$, by the residue balance guaranteed by
\Cref{eq:encoded-prime-modulus}.  However, the verifier tests only aggregate
balance and rank, rather than comparing that bit with a reference value, so
a changed bit need not cause rejection.  Indeed, if a direction set has
density $1-\epsilon$, then for every fixed direction perturbation
$\Delta$,
\[
 \Pr_{\bD}[\bD\text{ and }\bD\oplus\Delta\text{ are both accepted}]
 \geq 1-2\epsilon.
\]
Thus the private shifts make the audited profiles less amenable to
programming against a small public set of shifts, but excluding a
target-correlated perturbation remains part of
\Cref{ass:carry-ahcb}; no $(3/4)^{N_{\rm sh}}$ rejection bound follows from
the present tests.

The rank conditions address a separate elementary channel left by
contiguous-interval tests.  Column vectors
$\bv,\bv\oplus\bw,\bv$ of single-column derivative bits, with
$\bv,\bw$, and $\bv\oplus\bw$ balanced, pass every contiguous-interval test although the
first and third columns cancel exactly.  Full column rank rules out such
zero combinations in each audited derivative or gradient matrix at the
tested shifts.  It does not, on its own, rule out target correlations or
cancellations not witnessed by those matrices; those possibilities also
remain within the unproved content of the assumption.

The interval tests also exclude cancellation of the error within a row. Two
predicates in row $i$ at columns $j_1<j_2\in\mathcal J$ with the same word
$\bd$, say the low-order bit at base point $\xi=0$, have XOR
$\operatorname{lsb}(\varsigma)\oplus[\sigma_{i,j_1}(\bs)+\varsigma\geq q]$
with
$\varsigma=\sigma_{i,j_2}(\bs)-\sigma_{i,j_1}(\bs)
=\sum_{j_1<j'\leq j_2}\bA_{i,j'}s_{j'}\bmod q$, so $e_i$ enters only through
the wrap indicator, and
$\Delta^{\mathcal J}_{\rho,i}=[\sigma_{i,j_1}(\bs)\geq q-\rho]\oplus[\sigma_{i,j_2}(\bs)\geq q-\rho]$
is nonzero with probability $O(\rho/q)$; rows of this form give weight far
below $m/3$ at every tested shift. Conversely,
$\Delta^{\mathcal J}_{\rho,i}=1$ certifies only that the selected bits in
row $i$ change parity when $u_i$ is shifted by $\rho$.  For the low-bit word
$\bd=\{0,\ldots,\lfloor L/4\rfloor+1\}$ at $\xi=0$, the change probability
depends on the scale: in the ideal uniform low-bit cycle it is $1/2$ at
$a=L/4$, $3/4$ at $a=L/4-1$, and approaches $2/3$ at lower scales (with only
the negligible residue-boundary correction here).  Placing this word on
three quarters of the rows of each half therefore gives expected relative
weights ranging from $3/8$ to $9/16$, inside the balanced interval at every
tested scale.  The tests force a constant fraction of rows to depend on
$e_i$ at each scale, but they do not make that dependence unpredictable.

Adversarial base points expose an additional limitation that is
independent of these audits.  For any public $w_{i,j}\in\mathbb Z_q$, the
allowed choice
\[
 \xi_{i,j}=w_{i,j}-\chi u_i
\]
gives the exact identity
\[
 \xi_{i,j}+\chi\sigma_{i,j}(\bs)
 =w_{i,j}-\chi\sum_{k>j}\bA_{i,k}s_k\pmod q.
\]
Thus the key error $e_i$ cancels completely, leaving a noiseless binary
readout of a public-coefficient subset sum.  Such adaptively aligned base
points are not forbidden by admissibility; for these fixed bases, uniformly
random induced directions still satisfy the tests with the pointwise
probability proved above.
This is not by itself a predictor, because the suffix subset sum remains
hidden, but it shows that neither the key noise nor the shift audits can be
the source of hardness in the fully adaptive game.  The assumption must in
particular cover dense, correlated compositions of these noiseless
subset-sum carry readouts.

Shifts by positive and negative powers of two, together with the
private shifts $\varpi_{i,\tau}$, probe carry behavior at
every scale from one through the key-error scale $B_V$. A padding
strategy that passes the sensitivity condition must therefore spread
sensitivity of the selected bits over a constant fraction of the LWE
rows at every one of these
scales, in addition to spreading its secret-coordinate sensitivity over both
secret halves. The preceding arguments show only that terminal or locally
constant padding cannot repair these failures. Whether correlated
many-row padding can retain a bias is left to \Cref{ass:carry-ahcb}.

The vector $(S_j)_{j\in[n]}$ equals $(\Psi(\bs^{(j)}))_{j\in[n]}$ or its
bitwise complement according to the value of the target parity $\Psi(\bs)$.
Since the balance interval is invariant under complementation, the test on
$(S_j)_{j\in[n]}$ depends only on the values of $\Psi$ at the Hamming
neighbors of $\bs$ and not on the target parity itself. Nevertheless, the tests certify only
local sensitivity at the secret $\bs$. A many-row predicate could
in principle have balanced derivatives and remain predictable on the LWE
distribution, and the tests, public and private together, do not enumerate every possible
shift. Excluding such correlated predicates remains part of the
unproved content of \Cref{ass:carry-ahcb}.

%% file: sections/appendix_encoded_ahcb_evidence.tex
\section{\texorpdfstring{Partial Evidence for the Carry-Predicate Assumption}{Partial Evidence for the Carry-Predicate Assumption}}
\label{app:encoded-ahcb-evidence}

\Cref{ass:carry-ahcb} allows an adversary to choose the linear function,
the predicate list, and its base points jointly, after seeing the public
key. We do not know how to reduce this fully adaptive game to the
ordinary BCMVV adaptive-hardcore-bit property. In this appendix, we prove
three restricted results. First, key-mode indistinguishability rules out a
predictor that commits to its base points and can then be rerun from the
same saved state on independent uniform challenges. Second, the
ordinary adaptive-hardcore-bit property rules out
a sampler that can be repeatedly invoked from one saved state, produces
overwhelmingly correct equations, and chooses challenges of sufficiently
high conditional min-entropy. Third, every adversary restricted to empty
predicate lists has negligible advantage, with full adaptivity; the
case $\mathcal P=\emptyset$ of \Cref{ass:carry-ahcb}, in which
admissibility reduces to the linearity condition of \Cref{def:admissible},
therefore follows from \Cref{ass:encoded-prime-lwe} at the security level
of \Cref{prop:encoded-bcmvv}.
This establishes negligible advantage for the empty-list
case; it does not establish the stronger $2^{-c\sqrt n}$ advantage bound
postulated in \Cref{ass:carry-ahcb}.

Recall from \Cref{sec:encoded-parameters} the security parameter
$\lambda$, the prime modulus $q$, its bit length
$L=\lceil\log_2 q\rceil=\Theta(\log^2\lambda)$, the dimension $n$ of the
binary LWE secret, and the number $m$ of LWE samples. The
reductions rest on \Cref{prop:encoded-bcmvv}, which supplies key-mode
indistinguishability and the ordinary adaptive-hardcore-bit bound for
the decoded family under
\Cref{ass:encoded-prime-lwe}. Here and below, the
\emph{ambient decoded family} denotes the maps
$(\bx,\boldeta)\mapsto\bA\bx+b\bu+\boldeta$ on
$\mathbb Z_q^n\times E^m$, indexed by $b\in\bits$.  For the ordinary
hardcore-bit statement we use the ambient domain $\mathbb Z_q^n$ of the
BCMVV relation.  This family extends the maps used by the protocol, which
restricts branch $b$ to $X_b$; that protocol restriction is not imposed
inside the ordinary BCMVV relation. Its ordinary
adaptive-hardcore-bit game is the BCMVV game with the sets
$H^{\rm dec}_{k,t}$ and $\overline H^{\rm dec}_{k,t}$ of
\Cref{app:decoded-relation}. Conditions~A.1--A.4 of BCMVV hold at our
parameters as stated, without relaxation (verified in
\Cref{app:lwe-assumption}). The reductions use the canonical map $J_{q}$
of \Cref{sec:preliminaries} together with bitwise XOR and linear algebra
over $\mathbb F_2$ on its output.

\subsection{The LWE assumption}
\label{app:lwe-assumption}

Set
\[
 \ell=\left\lceil\frac{\lambda^2}{L}\right\rceil,
 \qquad
 B_L=\left\lceil2\sqrt\ell\right\rceil,
\]
so that $\ell=\Theta(\lambda^2/\log^2\lambda)$ and, by
\Cref{eq:encoded-dimensions}, $n\geq c_1(\ell L+\lambda)$.

As in~\cite{BCMVV2021}, the construction of
\Cref{sec:encoded-construction} uses an LWE instance whose secret is
\emph{binary}. This simply means that the entries of the secret are $0$
and $1$, though it is still viewed as a vector over $\mathbb Z_{q}$; its
length is the parameter $n$. The assumption below instead has a secret
uniform over $\mathbb Z_{q}^{\ell}$: by the lossy-mode argument of BCMVV
(their Theorem~2.8 and the proof of their Lemma~4.4), LWE with a uniformly
random binary secret of length $n=\Omega(\ell\log q+\lambda)$ is at least
as hard as LWE at dimension $\ell$ with a uniform secret, which is why the
assumption is stated at dimension $\ell$.
We use the truncated
discrete Gaussian $D_{\mathbb Z_{q},B_L}$ of
BCMVV~\cite{BCMVV2021}: its weight at a centered representative $z$ is
proportional to
$\exp(-\pi z^2/B_L^2)$ for
$|z|\leq B_L$, and is zero otherwise.

\begin{assumption}[Subexponential prime-modulus LWE, based on
\cite{regev,BCMVV2021}]
\label{ass:encoded-prime-lwe}
There is a fixed constant $1/2<\delta_{\rm LWE}<1$ such that, for every
polynomially bounded sample count $m'$, no uniform quantum algorithm of
time at most $2^{\ell ^{\delta_{\rm LWE}}}$ distinguishes
\[
 (\bA,\bA\bs+\be)
 \quad\text{from}\quad
 (\bA,\bu)
\]
with advantage that is non-negligible as a function of $\lambda$, where
$\bA\leftarrow
 \mathbb Z_{q}^{m'\times \ell }$,
$\bs\leftarrow\mathbb Z_{q}^{\ell }$,
$\be\leftarrow
 D_{\mathbb Z_{q},B_L}^{m'}$, and
$\bu\leftarrow\mathbb Z_{q}^{m'}$.
\end{assumption}

We apply it with $m'=m$.
This is a parameter-specific, fine-grained LWE assumption; worst-case LWE
reductions do not by themselves establish its precise time--advantage
curve. The modulus has quasipolynomial magnitude, and the error
satisfies the usual
$B_L\geq2\sqrt{\ell }$ threshold. Regev's
quantum worst-case-to-average-case reduction therefore applies and relates
this ensemble to worst-case lattice problems at approximation factor
\[
 \widetilde O\!\left(
   \frac{\ell q}{B_L}
 \right)
 =2^{O(\log^2\ell )}
\]
\cite{regev}. This places the assumption in the standard LWE regime at a
quasipolynomial approximation factor, although the reduction does not prove
the asserted subexponential running-time hardness. The
standard dual-lattice attack estimate has exponent
\[
 \widetilde\Theta\!\left(
  \frac{\ell }
       {\log(q/B_L)}
 \right)
 =\widetilde\Theta(\lambda^{2}),
\]
which is larger than
$\ell ^{\delta_{\rm LWE}}$~\cite[Remark~4.2]{BCMVV2021}.
Thus this attack does not contradict the assumption, but the comparison is a
parameter check rather than a proof of hardness.

The recent quasipolynomial algorithm for extrapolated dihedral coset
problems (EDCP) of Bai et al.~\cite{BaiEtAl2025EDCP} is also consistent with this
assumption. Their algorithm is specific to power-of-two moduli, whereas
$q$ is prime. Moreover, at our quasipolynomial modulus its
state requirement is
$2^{\Omega(\log \ell \log q)}
 =2^{\Omega(\log^3\ell )}$, whereas the applicable LWE-to-EDCP
reduction supplies at most $2^{O(\log^2\ell )}$ states. It
therefore does not yield an attack on the LWE ensemble used here.

We now check the hypotheses of the original BCMVV theorem for our
construction, rather than use a quantitative relaxation of them.
Recall also the equation-key Gaussian
width $B_V$ from \Cref{sec:encoded-parameters}, and let
\[
 B_P
 =\frac{q}
   {2C_T\sqrt{mn\log_2q}}
\]
be the BCMVV function-noise scale. The dimensions satisfy
\[
 n
 =\Omega(\ell \log q+\lambda),
 \qquad
 m=\Omega(n\log q).
\]
Moreover, since
$L=\Theta(\log^2\lambda)$,
$B_L=\poly(\lambda)$, and
$B_V=2^{L/4}$,
\begin{equation}
 \frac{B_V}{B_L}
 =2^{\Theta(\log^2\lambda)},
 \qquad
 \frac{B_P}{B_V}
 =2^{\Theta(\log^2\lambda)}.
 \label{eq:encoded-bcmvv-superpoly-gaps}
\end{equation}
Both ratios are superpolynomial in $\lambda$. Thus Conditions~A.1--A.4
of BCMVV~\cite{BCMVV2021}, namely the two dimension conditions, the
definition of $B_P$, and the superpolynomial noise gaps, hold at our
parameters, and $q$ is prime as their construction requires. In particular, the
statistical losses from the overlap of the two function distributions
(their Lemma~2.4) and from the error-shifting step in the lossy-matrix
hybrid of their Lemma~4.4 are negligible.

\begin{proposition}[Prime-modulus BCMVV consequences]
\label{prop:encoded-bcmvv}
Under \Cref{ass:encoded-prime-lwe}, both of the following advantages are
negligible for every uniform quantum polynomial-time algorithm:
\begin{enumerate}
 \item the advantage in distinguishing the equation key
 $(\bA,\bA\bs+\be)$, with
 $\bs\leftarrow\bits^{n}$, from the image-key distribution;
 \item the prediction advantage in the ordinary BCMVV adaptive-hardcore-bit game
 for the ambient decoded prime-modulus family, made precise in
 \Cref{app:decoded-relation}.
\end{enumerate}
The same conclusion holds for adversaries of time
$2^{\ell ^{\delta'}}$ for every fixed
$0<\delta'<\delta_{\rm LWE}$. The reductions have polynomial running-time
overhead.
\end{proposition}

\begin{proof}
The preceding parameter check verifies Conditions~A.1--A.4 of
BCMVV~\cite{BCMVV2021}. Key-mode indistinguishability follows from the
hybrid in the proof of \cite[Lemma~4.4]{BCMVV2021}, which replaces the
public matrix by a lossy matrix (their Theorem~2.8) and shifts the
error; the same argument appears in Mahadev~\cite{Mahadev2018}. The
adaptive-hardcore-bit lemma, \cite[Lemma~4.3]{BCMVV2021}, proved from
the leakage bound of their Lemma~4.4, gives the ordinary
adaptive-hardcore-bit property. This property is applied
to the ambient relation of \Cref{app:decoded-relation}, without a hidden
test of membership in the protocol's truncated partner domain. (The
ordinary game contains no function-noise sample, so its statement
applies to our family even though our function noise is uniform rather
than the BCMVV Gaussian at scale $B_P$.) All
statistical hybrids are negligible: the noise-overlap and error-shifting
losses by \Cref{eq:encoded-bcmvv-superpoly-gaps}, and the finite-precision
key-sampling error by the sampler of \Cref{sec:encoded-parameters}, whose
statistical distance from $D_{\mathbb Z_{q},B_V}$ is $2^{-\Omega(n)}$;
every computational hybrid reduces with polynomial overhead to
\Cref{ass:encoded-prime-lwe}. For the
subexponential statement, such overhead is absorbed by replacing
$\delta_{\rm LWE}$ with any fixed smaller exponent $\delta'$.
\end{proof}

This proposition concerns only the decoded family. It does not establish
\Cref{ass:carry-ahcb}, which also admits carry bits of the partial sums of
the LWE equations. We use it as a standard-LWE basis for
the partial evidence in this appendix, rather than as a hypothesis of
the main result.

\subsection{The decoded relation}
\label{app:decoded-relation}

We now make the ordinary adaptive-hardcore-bit game for the decoded
family precise; it is the target of the reductions below. Put
\[
 N_X=nL,
\]
and, for $b\in\bits$, $\bx\in\mathbb Z_{q}^{n}$,
and $\bd\in\bits^{N_X}$ parsed as $n$ blocks of
$L$ bits, define the induced coefficient
\[
 I_{b,\bx}(\bd)_j
 =\bd_j\cdot
  \bigl(J_{q}(x_j)\oplus
  J_{q}(x_j-(-1)^b)\bigr),
 \qquad j\in[n].
\]
Following BCMVV~\cite{BCMVV2021}, the direction set $G_{k,b,\bx}$
consists of the $\bd$ for which $I_{b,\bx}(\bd)$ is nonzero on the
branch-$b$ half of $[n]$: the first $n/2$
coordinates for $b=0$ and the last $n/2$ for $b=1$.
Because the secret is binary, for the intended pair with
$\bx_{1-b}=\bx_b-(-1)^b\bs$ we have the pointwise identity
\begin{equation}
 \bd\cdot\bigl(J_{q}(\bx_0)\oplus
 J_{q}(\bx_1)\bigr)
 =\langle I_{b,\bx_b}(\bd),\bs\rangle.
 \label{eq:decoded-linearization}
\end{equation}
In the ordinary BCMVV adaptive-hardcore-bit game for the decoded
family, the adversary outputs a branch $b$, a preimage
$\bx_b\in\mathbb Z_q^n$, a
direction $\bd\in\bits^{N_X}$, and a bit; write
$\bx_{1-b}=\bx_b-(-1)^b\bs$ for the partner preimage.
Define $H^{\rm dec}_{k,t}$ to contain the tuples
\[
 \left(b,\bx_b,\bd,
   \bd\cdot
   \bigl(J_{q}(\bx_0)
         \oplus J_{q}(\bx_1)\bigr)\right)
\]
with $\bx_b\in\mathbb Z_q^n$ and
$\bd\in G_{k,0,\bx_0}\cap G_{k,1,\bx_1}$, and let
$\overline H^{\rm dec}_{k,t}$ be obtained by flipping the last bit. The
distinguished failure symbol $\bot$, as well as every malformed output, is
by convention in neither set. Item~2 of \Cref{prop:encoded-bcmvv} says that
every efficient algorithm $\mathcal B$ has negligible prediction advantage:
\[
 \left|
  \Pr_{(k,t)\leftarrow\Gen_{\Eq}(1^\lambda)}
       [\mathcal B(k)\in H^{\rm dec}_{k,t}]
  -
  \Pr_{(k,t)\leftarrow\Gen_{\Eq}(1^\lambda)}
       [\mathcal B(k)\in\overline H^{\rm dec}_{k,t}]
 \right|
 \leq\negl(\lambda).
\]
It is essential here not to intersect the relation with
the secret-dependent event $\bx_{1-b}\in X_{1-b}$.  Such an intersection
does not inherit the BCMVV adaptive-hardcore-bit security guarantee.  For
example, take $b=0$,
choose $j^\star$ in the first half, set $x_{j^\star}=Q-1$ and
$x_j=1$ for $j\ne j^\star$, sample $\bd$ from the fiber
$I_{0,\bx}(\bd)=\boldsymbol\delta_{j^\star}$, where
$\boldsymbol\delta_{j^\star}$ is the corresponding standard basis
vector, and output $1$.  The partner lies in
$X_1$ exactly when $s_{j^\star}=1$; on that event the predicted bit is
correct, and the opposite-branch direction test fails only with
probability $2^{-\Omega(n)}$.  The artificially truncated relation would
therefore give this adversary prediction advantage $1/2-2^{-\Omega(n)}$.

The second and third results instead use the following
safe-fiber procedure.  Its selected points lie in both protocol branch
domains for every binary secret, although the ambient BCMVV relation
itself does not require this extra property.

\begin{lemma}[Fiber realization]
\label{lem:fiber-realization}
There is a polynomial-time procedure that, given $\chi\in\{\pm1\}$ and
$\boldsymbol\phi\in\bits^{n}$, outputs $(b,\bx,\bd)$ with
$(-1)^b=\chi$, with $\bx\in X_b$ and
$\bx-\chi\bs\in X_{1-b}$ for every $\bs$ in the support of
the key generation, where we write $\bx_b=\bx$ and
$\bx_{1-b}=\bx-\chi\bs$, and with $\bd$ uniform in
$\{\bd:I_{b,\bx}(\bd)=\boldsymbol\phi\}$; in particular, pointwise,
\[
 \bd\cdot\bigl(J_{q}(\bx_0)\oplus
 J_{q}(\bx_1)\bigr)
 =\langle\boldsymbol\phi,\bs\rangle.
\]
Moreover:
\begin{enumerate}
 \item if $\boldsymbol\phi$ is nonzero on the branch-$b$ half, then
  $\bd\in G_{k,b,\bx_b}$ with certainty; and
 \item except with probability $2^{-\Omega(n)}$ over the
  key,
  $\Pr_{\bd}[\bd\notin G_{k,1-b,\bx_{1-b}}]\leq
  2^{-n/8}$.
\end{enumerate}
\end{lemma}

\begin{proof}
Take $\bx=\mathbf 1$ for $b=0$ and
$\bx=(Q-2)\mathbf 1$ for $b=1$; both branch domains are then
respected for every binary $\bs$. Independently for each $j$,
sample $\bd_j$ uniformly among the solutions of
$\bd_j\cdot g_j=\phi_j$, where
$g_j=J_{q}(x_j)\oplus
J_{q}(x_j-\chi)\neq\mathbf 0$; both values of
$\phi_j$ have solutions, and the resulting $\bd$ is uniform in the
fiber. The displayed identity is
\Cref{eq:decoded-linearization}, and item 1 is the definition of
$G_{k,b,\bx_b}$.

For item 2, consider the opposite-branch coefficient
$I_{1-b,\bx_{1-b}}(\bd)$. On coordinates with $s_j=1$, the
opposite-branch word equals $g_j$, so the coefficient equals
$\phi_j$. On coordinates with $s_j=0$, the opposite-branch word
is $g'_j=J_{q}(x_j)\oplus
J_{q}(x_j+\chi)$, which is nonzero and distinct from
$g_j$ because $2\chi\not\equiv0\pmod{q}$; over
$\mathbb F_2$, distinct nonzero words are linearly independent, so,
conditioned on $\bd_j\cdot g_j=\phi_j$, the bit
$\bd_j\cdot g'_j$ is uniform. Hence the opposite-branch
coefficient restricted to the zero coordinates of $\bs$ in the
$(1-b)$-half is uniform, and it vanishes on the whole half with
probability at most $2^{-z}$, where $z$ is the number of such
coordinates. Since $\bs$ is uniform, $z\geq n/8$ except
with probability $2^{-\Omega(n)}$ over the key, by a Chernoff
bound.
\end{proof}

\subsection{Uniform independent challenges}
\label{app:uniform-challenges}

Fix a key, a sign $\chi$, and an assignment of base points
$\xi:[m]\times[n-1]\to\mathbb Z_{q}$, and
collect the bits unknown to the adversary into the string
\[
 \mathbf Z=\Bigl(\bs,\ \bigl(J_{q}(\xi_{i,j})\oplus
 J_{q}(\xi_{i,j}+\chi\,\sigma_{i,j}(\bs))\bigr)_{i\in[m],\,
 j\in[n-1]}\Bigr)
 \in\bits^{N_Z},
 \qquad
 N_Z=n+m(n-1)L.
\]
A \emph{challenge} is a pair
$(\ba,\bD^h)\in\bits^{n}\times
(\bits^{L})^{m\times(n-1)}$, with
$\bD^h=(\bD^h_{i,j})_{i\in[m],j\in[n-1]}$. Writing
$\mathcal P(\bD^h,\xi)$ for the predicate list
$\{(i,j,\xi_{i,j},\bD^h_{i,j}):\bD^h_{i,j}\neq\mathbf 0\}$, a challenge is a
parity of $\mathbf Z$:
\[
 \langle(\ba,\bD^h),\mathbf Z\rangle
 =\langle\ba,\bs\rangle\oplus
 \Phi^{\chi}_{\mathcal P(\bD^h,\xi)}(\bs).
\]
Answering a challenge correctly is therefore winning the game of
\Cref{ass:carry-ahcb} at the pair $(\ba,\mathcal P(\bD^h,\xi))$, up to
admissibility; a uniform challenge is admissible except with probability
$2^{-\Omega(n)}$, by \Cref{prop:encoded-good-density} applied
at the branch $b$ with $(-1)^b=\chi$ to the honestly encoded
transcript in the proof of \Cref{lem:realizability}, with coins
$\bh_{i,j}=\xi_{i,j}$ for $j\leq n-1$; since the words
$J_{q}(x_{b,j})\oplus J_{q}(x_{b,j}-\chi)$ are nonzero and the
$\bD^x,\bD^r$ blocks are disjoint from the $\bD^h$ blocks, the pair
$(\ba_b(\bD),\mathcal P_b(\bD))$ of a uniform direction has exactly the
law of $(\ba,\mathcal P(\bD^h,\xi))$ for a uniform challenge.

\begin{definition}[Resettable uniform-challenge predictor]
\label{def:uniform-challenge-predictor}
A \emph{resettable uniform-challenge predictor} consists of two classical
algorithms
$\mathcal A=(\mathcal A_{\rm com},\mathcal A_{\rm pred})$.
On input an equation key $k$, the first algorithm outputs a sign
$\chi$, an assignment of base points $\xi$, and a classical snapshot
$\mathsf{st}$ that can be copied and restored. Given $\mathsf{st}$ and a
challenge $(\ba,\bD^h)$, a fresh invocation of $\mathcal A_{\rm pred}$,
with fresh random coins $\varrho$, outputs one bit. Let
$\mathsf{tr}=(k,t,\chi,\xi,\mathsf{st})$ denote the complete first-stage transcript
and let $\mathbf Z=\mathbf Z(\mathsf{tr})$ be as above. Define
\begin{equation}
 \mathrm{corr}_{\mathsf{tr}}=
 \mathbb E_{\substack{
      (\ba,\bD^h)\ \text{uniform}\\
      \varrho}}
 \left[
  (-1)^{
   \mathcal A_{\rm pred}(\mathsf{st},\ba,\bD^h;\varrho)
   \oplus \langle(\ba,\bD^h),\mathbf Z\rangle}
 \right],
 \label{eq:uniform-challenge-correlation}
\end{equation}
and the \emph{uniform-challenge advantage}
$\Adv_{\rm uc}(\mathcal A)
 =\left|\mathbb E_{\mathsf{tr}}[\mathrm{corr}_{\mathsf{tr}}]\right|$.
\end{definition}

\begin{proposition}[Uniform challenges]
\label{prop:encoded-uniform-challenge}
Let $\mathcal A=(\mathcal A_{\rm com},\mathcal A_{\rm pred})$ be a
resettable uniform-challenge predictor
(\Cref{def:uniform-challenge-predictor}) with stage running times
$T_{\rm com},T_{\rm pred}$ and $\Adv_{\rm uc}(\mathcal A)\geq\mu>0$ for a
known $\mu$. Then there is a distinguisher between
the equation- and image-key modes of running time
\[
 T_{\rm com}+
 \poly(N_Z,1/\mu)\,
 \bigl(T_{\rm pred}+N_Z\bigr)
\]
and advantage at least $\mu/3-\negl(\lambda)$.
Consequently, every efficient resettable uniform-challenge predictor has
negligible advantage: an inverse-polynomial $\mu$ makes the reduction
efficient and its advantage non-negligible, contrary to
\Cref{prop:encoded-bcmvv}.
\end{proposition}

\begin{proof}
Suppose that
$\left|\mathbb E_{\mathsf{tr}}\mathrm{corr}_{\mathsf{tr}}\right|\geq\mu$.
Then
$\mathbb E_{\mathsf{tr}}|\mathrm{corr}_{\mathsf{tr}}|\geq\mu$, and hence
\[
 \Pr_{\mathsf{tr}}
 \left[
   |\mathrm{corr}_{\mathsf{tr}}|\geq\mu/2
 \right]
 \geq\frac{\mu}{2-\mu}
 \geq\frac\mu2,
\]
since otherwise the expectation of $|\mathrm{corr}_{\mathsf{tr}}|$, which is always at
most one, would be smaller than $\mu$.

Fix such a transcript. Reset access to
$\mathcal A_{\rm pred}$ gives a randomized Boolean oracle whose Fourier
coefficient at $\mathbf Z$ has magnitude at least $\mu/2$, by
\Cref{eq:uniform-challenge-correlation}. Apply the Goldreich--Levin
list-decoding algorithm~\cite{GoldreichLevin1989} at threshold $\mu/4$;
since the oracle is randomized, we use its form for $[-1,1]$-valued
functions. Using
$\poly(N_Z,1/\mu)$ reset invocations, it returns a list of
$\poly(1/\mu)$ candidates that contains
$\mathbf Z$ with probability at least $2/3$. Fresh prediction
coins cause no difficulty: repeated invocations estimate, to additive
accuracy $O(\mu)$, the bounded oracle
\[
 g_{\mathsf{st}}(\ba,\bD^h)=
 \mathbb E_{\varrho}
 \left[(-1)^{\mathcal A_{\rm pred}(\mathsf{st},\ba,\bD^h;\varrho)}\right].
\]
The stated success probability includes amplification of these estimates and
a union bound over all adaptive list-decoder queries.

For each list element, parse its first $n$ bits as a
candidate secret $\bs'$ and accept if
$\|\centerq(\bu-\bA\bs')\|_\infty\leq B_V$; the distinguisher
outputs ``equation'' if some candidate is accepted. On an equation key,
the true $\mathbf Z$ is on the list except with probability $1/3$ whenever
$|\mathrm{corr}_{\mathsf{tr}}|\geq\mu/2$, and its $\bs$-block passes the test; the
acceptance probability is at least $(\mu/2)(2/3)=\mu/3$. On an image
key, $[\bA\mid\bu]$ is statistically close to uniform, and for uniform
$\bu$ the probability that any of the $2^{n}$ binary
vectors passes the test is at most
$2^{n}\bigl((2B_V+1)/q
\bigr)^{m}
=2^{-\Omega(mL)}$. The distinguishing advantage is
therefore at least $\mu/3-\negl(\lambda)$, and
\Cref{prop:encoded-bcmvv} bounds it by a negligible function. Formally,
if the predictor's advantage were non-negligible, fix an
inverse-polynomial lower bound attained on an infinite subsequence and use
it as $\mu$.
\end{proof}

If the uniform-challenge experiment also filters through the admissibility
predicate, removing the filter changes the correlation by at most the
probability of an inadmissible challenge, $2^{-\Omega(n)}$, so the same conclusion
holds with $\mu-2^{-\Omega(n)}$ in place of $\mu$.
\Cref{ass:encoded-prime-lwe} absorbs the polynomial overhead of the
reduction for inverse-polynomial $\mu$. The proposition nevertheless
remains structurally weaker than \Cref{ass:carry-ahcb}: it requires a
predictor that commits to its base points and can then be reset and
queried on independent uniform challenges.

\subsection{High-min-entropy equation samplers}
\label{app:high-entropy-samplers}

The honest Hadamard experiment produces a uniformly random direction, and
the induced challenge of \Cref{app:uniform-challenges} then has the
maximum possible min-entropy, namely $N_Z$ bits. We next ask how far a
sampler may depart from this benchmark while still being ruled out by the
ordinary property. The conclusion in this subsection requires equations
that are correct with overwhelming probability; samplers with
inverse-polynomial bias are treated in \Cref{app:linear-predicates} for
the linear case.

\begin{definition}[Resettable equation sampler]
\label{def:equation-sampler}
A \emph{resettable equation sampler} is a pair
$(\mathcal A_{\rm com},\mathcal S)$
of classical algorithms. The first algorithm, $\mathcal A_{\rm com}(k)$,
produces a sign $\chi$, an assignment of base points $\xi$, and one
fixed, copyable classical snapshot $\mathsf{st}$. Let
$\mathsf{tr}=(k,t,\chi,\xi,\mathsf{st})$ be the complete transcript, including the
private key data, and let $\mathbf Z=\mathbf Z(\mathsf{tr})\in\bits^{N_Z}$ be as in
\Cref{app:uniform-challenges}. The sampler is invoked as
\[
 (\ba^{(\tau)},\bD^{h,(\tau)},C^{(\tau)})\leftarrow
 \mathcal S\bigl(\mathsf{st},(\ba^{(\tau')},\bD^{h,(\tau')},C^{(\tau')})_{\tau'<\tau}\bigr)
\]
with fresh coins; all persistent copyable state is included in $\mathsf{st}$
and the displayed history. For an integer
$0\leq\kappa_{\rm ent}\leq N_Z$, the sampler has \emph{conditional
min-entropy deficiency at most $\kappa_{\rm ent}$} and \emph{equation
error at most $\eps$} if, for every transcript and every
positive-probability history,
\begin{align}
 \max_{(\ba,\bD^h)}
 \Pr[(\ba^{(\tau)},\bD^{h,(\tau)})=(\ba,\bD^h)\mid
       \mathsf{tr},(\ba^{(\tau')},\bD^{h,(\tau')},C^{(\tau')})_{\tau'<\tau}]
 &\leq 2^{-N_Z+\kappa_{\rm ent}},
 \label{eq:conditional-min-entropy}\\
 \Pr[C^{(\tau)}\neq\langle(\ba^{(\tau)},\bD^{h,(\tau)}),\mathbf Z\rangle\mid
       \mathsf{tr},(\ba^{(\tau')},\bD^{h,(\tau')},C^{(\tau')})_{\tau'<\tau}]
 &\leq\eps.
 \nonumber
\end{align}
\end{definition}

Deficiency zero forces a uniform challenge, while a challenge uniform on a
codimension-$\kappa_{\rm ent}$ subspace attains the bound in
\Cref{eq:conditional-min-entropy}. This is worst-case conditional
min-entropy, rather than the average or smooth variant. If samples are
conditionally independent given each fixed transcript, it suffices to
impose the corresponding pointwise bound for every transcript.

\begin{proposition}[High-entropy correct-equation samplers]
\label{prop:encoded-high-entropy-sampler}
Let $(\mathcal A_{\rm com},\mathcal S)$ be a resettable equation
sampler (\Cref{def:equation-sampler}) with conditional min-entropy
deficiency at most $\kappa_{\rm ent}$ and equation error at most $\eps$.
For any $0<\delta_{\rm rk}<1$, put
\[
 T=\left\lceil
    8\bigl(N_Z+\ln(1/\delta_{\rm rk})\bigr)
   \right\rceil.
\]
There is an ordinary decoded-family adaptive-hardcore-bit adversary, using $T$ invocations
of $\mathcal S$ and polynomial additional time, whose prediction advantage is
at least
\begin{equation}
 1-2^{\kappa_{\rm ent}-n/2}
   -2^{-n/8}
   -2^{-\Omega(n)}
   -2(\delta_{\rm rk}+T\eps).
 \label{eq:high-entropy-base-advantage}
\end{equation}
Thus such an efficient sampler is ruled out whenever the quantity in
\Cref{eq:high-entropy-base-advantage} is non-negligibly positive.
For example, with $\delta_{\rm rk}=2^{-n/8}$ we have $T=O(N_Z)$, and
\Cref{prop:encoded-bcmvv} rules out efficient resettable samplers satisfying
\[
 \kappa_{\rm ent}\leq n/4,
 \qquad
 T\eps=2^{-\Omega(n)}.
\]
\end{proposition}

\begin{proof}
Write $N=N_Z$. The reduction collects $T$ equations.
For the following rank and correctness estimates, fix a transcript. Let
$V_\tau=\operatorname{span}\{(\ba^{(1)},\bD^{h,(1)}),\ldots,(\ba^{(\tau)},\bD^{h,(\tau)})\}$. While
$\dim V_{\tau-1}<N-\kappa_{\rm ent}$, the conditional min-entropy bound gives
\[
 \Pr[(\ba^{(\tau)},\bD^{h,(\tau)})\in V_{\tau-1}\mid\text{history}]
 \leq |V_{\tau-1}|2^{-N+\kappa_{\rm ent}}
 \leq\frac12.
\]
Thus the rank increments stochastically dominate independent
Bernoulli$(1/2)$ trials until rank $N-\kappa_{\rm ent}$. To verify the
constant in $T$, put $\mu_T=T/2$. Then
$\mu_T\geq4(N+\ln(1/\delta_{\rm rk}))$ and
$N-\kappa_{\rm ent}\leq N\leq\mu_T/4$. A Chernoff bound, coupled to the
adaptive increments until the target rank is reached, gives
\[
 \Pr\left[
  \operatorname{rank}\bigl((\ba^{(1)},\bD^{h,(1)}),\ldots,(\ba^{(T)},\bD^{h,(T)})\bigr)
       <N-\kappa_{\rm ent}
 \right]
 \leq \exp(-9\mu_T/32)
 \leq\delta_{\rm rk}.
\]

Conditional on the transcript, except with probability at
most $T\eps$, all sampled equations are correct. Let $\boldsymbol\Gamma$ be the binary
matrix whose rows are the sampled challenges and let
$\mathbf C=(C^{(1)},\ldots,C^{(T)})$. If the sampled matrix has rank
below $N-\kappa_{\rm ent}$, or if $\boldsymbol\Gamma\bz=\mathbf C$ is inconsistent, the
reduction outputs $\bot$. Otherwise, Gaussian elimination gives
\[
 \{\bz\in\bits^N:\boldsymbol\Gamma\bz=\mathbf C\}=\bz_*+\ker\boldsymbol\Gamma,
 \qquad \dim\ker\boldsymbol\Gamma\leq\kappa_{\rm ent},
\]
and on the correct-equations event the true string $\mathbf Z$ belongs to this
affine space. Projecting onto the $\bs$-coordinates gives
\[
 \bs\in\ba_*+U,
 \qquad
 \dim U\leq\kappa_{\rm ent}.
\]

Sample $\boldsymbol\phi$ uniformly from $U^\perp\subseteq
\bits^{n}$, set
$c=\langle\boldsymbol\phi,\ba_*\rangle$; since
$\bs\oplus\ba_*\in U$, the equation
$\langle\boldsymbol\phi,\bs\rangle=c$ is exact. Run the fiber
realization (\Cref{lem:fiber-realization}) on
$(\chi,\boldsymbol\phi)$ to obtain $(b,\bx,\bd)$, and output
$(b,\bx,\bd,c)$ to the decoded game; the equation is exact by the
displayed identity there. For the direction conditions:
$U^\perp$ has density at least $2^{-\kappa_{\rm ent}}$ in
$\bits^{n}$, and the set of vectors vanishing on the
branch-$b$ half has density $2^{-n/2}$, so
$\boldsymbol\phi$ is nonzero there except with probability
$2^{\kappa_{\rm ent}-n/2}$; the opposite-branch condition
fails with probability at most $2^{-n/8}$ on a
$1-2^{-\Omega(n)}$ key event.

On the full success event and for a direction satisfying both conditions,
the reduction's output lies in $H^{\rm dec}_{k,t}$. If a direction
condition fails, the output lies in neither set and contributes zero; only
when the success event fails can the output lie in
$\overline H^{\rm dec}_{k,t}$, and on that event its contribution is at
least $-1$. The success event fails with probability at most
$\beta_{\rm fail}=\delta_{\rm rk}+T\eps$, and on the success event the
direction conditions fail with probability at most
$\gamma_{\rm dir}=2^{\kappa_{\rm ent}-n/2}+2^{-n/8}+2^{-\Omega(n)}$.
Therefore its prediction advantage is at least
\[
 (1-\beta_{\rm fail}-\gamma_{\rm dir})-\beta_{\rm fail}
 =1-\gamma_{\rm dir}-2\beta_{\rm fail},
\]
which proves \Cref{eq:high-entropy-base-advantage}.
\end{proof}

In particular, \Cref{prop:encoded-high-entropy-sampler} rules out an
exact sampler whose challenges, for one fixed snapshot, are independent
uniform, or equivalently remain uniform conditioned on every prior
history. More generally, the entropy may be lowered from its maximum value
$N_Z$ by as many as $n/4$ bits: each challenge may be up to a factor
$2^{n/4}$ more likely than it would be under the uniform distribution, and
the conclusion still follows. The successive challenge distributions may
depend on all preceding equations, but the equations themselves must be
correct with overwhelming probability.

\subsection{Linear predicates}
\label{app:linear-predicates}

We now consider adversaries whose predicate lists are empty, so that the
predicted bit is a linear function of the secret. For this class the
assumption is not merely supported by evidence: at the security level of
\Cref{prop:encoded-bcmvv}, it follows from
\Cref{ass:encoded-prime-lwe}.

\begin{proposition}[Empty predicate lists]
\label{prop:linear-predicates}
Let $\mathcal A$ be a uniform quantum polynomial-time algorithm in the
game of \Cref{ass:carry-ahcb} whose output always has
$\mathcal P=\emptyset$. Under \Cref{ass:encoded-prime-lwe}, the prediction
advantage of $\mathcal A$ is negligible.
\end{proposition}

\begin{proof}
An empty predicate list makes every bit $\delta_{i,j}(\rho)$ in
\Cref{eq:adm-error-profile} zero, so the sensitivity condition fails, and
admissibility is equivalent to the linearity condition: $\ba$ nonzero on the half
designated for $\chi$.

The reduction runs $\mathcal A(k)$ to obtain
$(\chi,\ba,\emptyset,C)$, and outputs $\bot$ if the output is
inadmissible; for empty predicate lists this is the public condition that
$\ba$ be nonzero on the designated half, so the reduction can check it. Otherwise it runs the fiber realization
(\Cref{lem:fiber-realization}) on $(\chi,\ba)$ and outputs
$(b,\bx,\bd,C)$ to the decoded game. The equation
$\bd\cdot(J_{q}(\bx_0)\oplus J_{q}(\bx_1))
=\langle\ba,\bs\rangle$ is exact, so the output lies in
$H^{\rm dec}_{k,t}$ exactly when $\mathcal A$'s output lies in
$H^{\rm carry}_{k,t}$, and in the flipped set exactly when it lies in
$\overline H^{\rm carry}_{k,t}$, provided
$\bd\in G_{k,0,\bx_0}\cap G_{k,1,\bx_1}$. The branch-$b$ condition
(item~1 of \Cref{lem:fiber-realization})
holds with certainty by the linearity condition, and the opposite-branch condition
(item~2) fails with probability at most $2^{-n/8}$ on a
$1-2^{-\Omega(n)}$ key event; on the failure events the
output lies in neither decoded set and contributes zero to the
correct-minus-incorrect difference. Hence the reduction's prediction
advantage is within
$2^{-n/8}+2^{-\Omega(n)}$ of $\mathcal A$'s,
and \Cref{prop:encoded-bcmvv} bounds the former by a negligible function.
\end{proof}

In particular, the case $\mathcal P=\emptyset$ of
\Cref{ass:carry-ahcb} follows from \Cref{ass:encoded-prime-lwe} with
negligible prediction advantage, at the security level of
\Cref{prop:encoded-bcmvv}.  This is sufficient for a
$3/4+\negl(\lambda)$ soundness conclusion in that restricted case, but
it does not supply the $2^{-c\sqrt n}$ advantage bound in the full
statement of \Cref{ass:carry-ahcb}.

The proposition extends to samplers that are only weakly correct, provided
their linear functions are sufficiently spread; this rules out a family of
attacks that need not respect the admissibility predicate at all.

\begin{definition}[Linear sampler]
\label{def:linear-sampler}
A \emph{linear sampler} is a pair $(\mathcal A_{\rm com},\mathcal S)$ of
uniform
polynomial-time classical algorithms. On input $k$, the first algorithm
outputs a sign $\chi$ and one copyable snapshot $\mathsf{st}$.
One invocation of $\mathcal S(\mathsf{st})$ outputs
$(\ba,C)\in\bits^{n}\times\bits$ satisfying, conditioned on
every fixed $(k,t,\mathsf{st})$,
\begin{align}
 H_\infty(\ba\mid k,t,\mathsf{st})&\geq h_{\min},
 \label{eq:linear-coefficient-entropy}\\
 \Pr[C=\langle\ba,\bs\rangle\mid k,t,\mathsf{st}]
 &\geq\frac12+\gamma.
 \label{eq:linear-correctness-bias}
\end{align}
The joint distribution of $(\ba,C)$ is otherwise arbitrary; in
particular, the error probability may depend on $\ba$.
\end{definition}

\begin{proposition}[Biased linear samplers]
\label{prop:linear-biased}
Let $(\mathcal A_{\rm com},\mathcal S)$ be a linear sampler
(\Cref{def:linear-sampler}) with parameters $h_{\min},\gamma$. Then there is
a polynomial-time ordinary decoded-family adaptive-hardcore-bit
adversary whose prediction advantage is at least
\begin{equation}
 2\gamma-\gamma_{\rm good},
 \qquad
 \gamma_{\rm good}
 \leq2^{n/2-h_{\min}}
 +2^{-n/8}+2^{-\Omega(n)}.
 \label{eq:linear-biased-advantage}
\end{equation}
Consequently, \Cref{prop:encoded-bcmvv} rules out the sampler whenever
$2\gamma-\gamma_{\rm good}$ is non-negligibly positive. In particular,
for any inverse-polynomial $\gamma$ it suffices that
\begin{equation}
 h_{\min}\geq
 \left\lceil
  \frac{n}{2}+\log_2(2/\gamma)
 \right\rceil.
 \label{eq:linear-entropy-threshold}
\end{equation}
This conclusion uses only one sampler invocation.
\end{proposition}

\begin{proof}
The reduction invokes $\mathcal S$ once, obtains $(\ba,C)$, runs the
fiber realization on $(\chi,\ba)$, and outputs $(b,\bx,\bd,C)$ to the
decoded game. The equation bit
$\bd\cdot(J_{q}(\bx_0)\oplus J_{q}(\bx_1))
=\langle\ba,\bs\rangle$ is exact and does not depend on the fiber
choice.

Write $Y=(-1)^{C\oplus\langle\ba,\bs\rangle}$; by
\Cref{eq:linear-correctness-bias},
$\mathbb E[Y\mid k,t,\mathsf{st}]\geq2\gamma$. The output enters the correct
or flipped decoded set according to $Y$ whenever
$\bd\in G_{k,0,\bx_0}\cap G_{k,1,\bx_1}$, and neither set otherwise, so
filtering changes the expectation above by at most the failure
probability of the direction conditions. For the branch-$b$ condition, the
min-entropy bound \Cref{eq:linear-coefficient-entropy} gives
$\Pr[\ba\text{ vanishes on the half designated for }\chi]
\leq2^{n/2-h_{\min}}$, since that set has
$2^{n/2}$ elements, each of probability at most $2^{-h_{\min}}$.
The opposite branch fails with probability at most
$2^{-n/8}$ on a $1-2^{-\Omega(n)}$ key event.
The resulting prediction advantage is therefore at least
$2\gamma-\gamma_{\rm good}$ with $\gamma_{\rm good}$ as displayed.
Finally, \Cref{eq:linear-entropy-threshold} gives
$2^{n/2-h_{\min}}\leq\gamma/2$, so
\Cref{eq:linear-biased-advantage} is at least
$3\gamma/2-2^{-\Omega(n)}$, which is non-negligible for every
inverse-polynomial $\gamma$ and contradicts
\Cref{prop:encoded-bcmvv}.
\end{proof}

When the linear functions are uniform and repeated labeled samples are
independent, the equations above can also be viewed as LPN with
coefficient-dependent, or adversarial, classification noise. Known
agnostic-learning algorithms reduce parity learning under the uniform
distribution with such adversarial classification noise to parity learning
with random classification noise; combining this reduction with the Blum--Kalai--Wasserman algorithm gives a
subexponential algorithm
\cite{FeldmanEtAl2009,BlumKalaiWasserman2003}. The direct argument
above is stronger for our purpose: it needs neither repeated samples nor a
noise model, and it preserves any inverse-polynomial correlation without
first recovering the secret.

The three results are complementary, and their restrictions delimit what
remains open. The first requires reset access to a predictor after its
base points are fixed, the second requires many correct equations for one
fixed snapshot, and the third covers only empty predicate lists. These
restrictions leave open the fully adaptive game in
\Cref{ass:carry-ahcb}, where the linear function, the predicate list, and
the base points may be chosen jointly and used once.

%% file: sections/appendix_encoded_parameters.tex
\section{\texorpdfstring{Parameters for the Encoded Construction}{Parameters for the Encoded Construction}}
\label{app:encoded-parameters}

In this appendix, we prove the parameter estimates used in the completeness
and soundness arguments. Recall from \Cref{sec:encoded-parameters} and \appref{app:encoded-ahcb-evidence} that
\[
\begin{aligned}
 L
  &=4\left\lceil(\log_2(\lambda+2))^2\right\rceil
    =\Theta(\log^2\lambda),
 &3\cdot2^{L-2}\leq q&<2^{L}
    &&\text{with $q$ prime},\\
 \ell&=\Theta(\lambda^2/L),
 &n&=\Theta(\lambda^2),
 &m&=\Theta(nL),
\end{aligned}
\]
that the trapdoor inversion radius is
\[
 r_G
 =\frac{q}
 {C_T\sqrt{(n+1)L}},
\]
and that
\[
 P=\frac{r_G}{2m},
 \qquad
 Q=2^{L/2},
 \qquad
 B_L=\left\lceil2\sqrt{\ell}\right\rceil,
 \qquad
 B_V=2^{L/4},
 \qquad
 B_P
 =\frac{q}
 {2C_T\sqrt{mn\log_2q}}.
\]

\begin{lemma}[Parameter estimates]
\label{lem:parameter-estimates}
For all sufficiently large $\lambda$, the following hold.
\begin{enumerate}
 \item \textbf{Noise hierarchy.}
  $B_L<B_V<B_P$, and the ratios
  $B_V/B_L$ and
  $B_P/B_V$ are superpolynomial in $\lambda$.
 \item \textbf{Completeness errors.}
  $mB_V/W
    +2n/Q=\negl(\lambda)$.
 \item \textbf{Interval overlap and no wrap-around.}
  $2B_V<W$ and
  $2W+2B_V<q$.
 \item \textbf{Trapdoor uniqueness radii.}
  $\sqrt{m}\,B_V<r_G$, every
  $\boldeta\in E^{m}$ satisfies
  $\|\centerq(\boldeta)\|_2<r_G$, and
  $\|\centerq(\boldeta+\be)\|_2<r_G$ for every key error $\be$ in the
  support of
  $D_{\mathbb Z_{q},B_V}^{m}$.
 \item \textbf{LK-$1/4$ radius.}  For every generated matrix
  $\bK$ satisfying \Cref{eq:encoded-lattice-separation}, every
  $\boldeta\in E^{m}$ satisfies
  $\|\centerq(\boldeta)\|_\infty
    \leq\frac14\lambda_1^{\infty}(\mathcal L(\bK))$.
\end{enumerate}
\end{lemma}

\begin{proof}
The three noise scales satisfy
\[
 B_L
 =\Theta\!\left(\frac{\lambda}{\sqrt{L}}\right),
 \qquad
 B_V=2^{L/4},
 \qquad
 B_P
 =\frac{2^{L+O(1)}}{\lambda^2L}.
\]
Thus $B_L<B_V<B_P$ for all sufficiently
large parameters, and
\[
 \frac{B_V}{B_L}
 =2^{L/4-O(\log\lambda)},\qquad
 \frac{B_P}{B_V}
 =2^{3L/4-O(\log\lambda)}
\]
are both superpolynomial in $\lambda$. These are precisely the two
noise-gap conditions in BCMVV~\cite{BCMVV2021}, as in
\Cref{eq:encoded-bcmvv-superpoly-gaps}; this proves the first item.

For the second item, note that for sufficiently large $\lambda$, the
definition in
\Cref{eq:encoded-width-definition} makes $W$ the largest power
of two satisfying the integer inequality
\[
 4m^2C_T^2(n+1)L\,
 W^2\leq q^2.
\]
Equivalently, $W\leq P$, and maximality gives
$W>P/2$. Consequently,
\begin{align*}
 \frac{mB_V}{W}
 &<
 \frac{4m^2B_V}{r_G}\\
 &=
 O\!\left(
   \lambda^5L^{5/2}\,2^{-3L/4}
 \right)
 =\negl(\lambda).
\end{align*}
Also,
\[
 \frac{2n}{Q}
 =O\!\left(\lambda^2 2^{-L/2}\right)
 =\negl(\lambda).
\]

The same estimates give
\[
 \frac{B_V}{W}
 =\negl(\lambda),\qquad
 \frac{2W+2B_V}{q}
 =\lambda^{-3+o(1)}.
\]
The second estimate is used only as a deterministic no-wrap inequality;
it is not a completeness-error term. In particular,
$2B_V<W$ and
$2W+2B_V<q$ for all sufficiently large
$\lambda$, proving the third item.

For the fourth item, every centered noise vector in
$E^{m}$ satisfies
\[
 \|\centerq(\boldeta)\|_2
 \leq\frac{\sqrt{m}\,W}{2}
 \leq\frac{r_G}{4\sqrt{m}}
 <r_G.
\]
The key error also satisfies
\[
 \frac{\sqrt{m}\,B_V}{r_G}
 =
 O\!\left(
  \lambda^2L2^{-3L/4}
 \right)
 =\negl(\lambda).
\]
Hence, for all sufficiently large $\lambda$,
\[
 \|\centerq(\boldeta+\be)\|_2
 \leq\|\centerq(\boldeta)\|_2+\sqrt{m}\,B_V
 <r_G.
\]
These are the two trapdoor-uniqueness bounds used for equation keys.

Finally,
$\lambda_1^{\infty}(\mathcal L(\bK))
\geq r_G/\sqrt{m}$ by
\Cref{eq:encoded-lattice-separation}, and hence
\[
 \|\centerq(\boldeta)\|_\infty
 \leq\frac{W}{2}
 \leq\frac{r_G}{4m}
 \leq\frac14\lambda_1^{\infty}(\mathcal L(\bK)).
\]
This is the LK-$1/4$ inequality used in
\Cref{lem:encoded-image-extraction}, and it uses only the lower
bound of \Cref{eq:encoded-lattice-separation}.
\end{proof}

The main register lengths are
\[
nL=\Theta(\lambda^2\log^2\lambda),
\qquad
mL=\Theta(\lambda^2\log^4\lambda),
\qquad
mnL
 =\Theta(\lambda^4\log^4\lambda).
\]
The last quantity is, up to a constant factor, the number of qubits in
the source and output registers of \Cref{lem:encoded-circuit}.